\documentclass[12pt, a4paper]{article}

\usepackage{amsmath}
\usepackage{amssymb}
\usepackage{amsthm}
\usepackage{graphicx}
\usepackage{psfrag}
\usepackage{enumerate}

\newtheorem{thm}{Theorem}[section]

\newtheorem{lem}[thm]{Lemma}

\newtheorem{assumption}[thm]{Assumption}

\newtheorem{definition}[thm]{Definition}

\newtheorem{example}[thm]{Example}
\newenvironment{exmp}{\begin{example}\rm}{\end{example}}
\newtheorem{remark}[thm]{Remark}
\newenvironment{rem}{\begin{remark}\rm}{\end{remark}}
\newtheorem{tab}{Table}

\newcommand\mcm{\mathcal{M}}
\newcommand{\mstar}{\mcm^*}
\newcommand{\ceil}[1]{\left\lceil{#1}\right\rceil}
\newcommand{\floor}[1]{\left\lfloor{#1}\right\rfloor}
\newcommand{\rbrac}[1]{\left(#1\right)}

\title{A Graph Theoretical Approach to Network Encoding Complexity}

\author{\small \begin{tabular}{ccc}
Li Xu & Weiping Shang & Guangyue Han\\
The University of Hong Kong & Zhengzhou University & The University of Hong Kong\\
email: xuli@hku.hk & email: shangwp@zzu.edu.cn& email: ghan@hku.hk\\
\end{tabular}}

\date{{\normalsize \today}}

\begin{document}
\maketitle

\begin{abstract}
Consider an acyclic directed network $G$ with sources $S_1, S_2, \ldots, S_l$ and distinct sinks $R_1, R_2, \ldots, R_l$. For $i=1, 2, \ldots, l$, let $c_i$ denote the min-cut between $S_i$ and $R_i$. Then, by Menger's theorem, there exists a group of $c_i$ edge-disjoint paths from $S_i$ to $R_i$, which will be referred to as a group of Menger's paths from $S_i$ to $R_i$ in this paper. Although within the same group they are edge-disjoint, the Menger's paths from different groups may have to merge with each other. It is known that by choosing Menger's paths appropriately, the number of mergings among different groups of Menger's paths is always bounded by a constant, which is independent of the size and the topology of $G$. The tightest such constant for the all the above-mentioned networks is denoted by $\mathcal{M}(c_1, c_2, \ldots, c_2)$ when all $S_i$'s are distinct, and by $\mathcal{M}^*(c_1, c_2, \ldots, c_2)$ when all $S_i$'s are in fact identical. It turns out that $\mathcal{M}$ and $\mathcal{M}^*$ are closely related to the network encoding complexity for a variety of networks, such as multicast networks, two-way networks and networks with multiple sessions of unicast. Using this connection, we compute in this paper some exact values and bounds in network encoding complexity using a graph theoretical approach.
\end{abstract}

\section{Introduction and Notations}

Let $G(V, E)$ denote an acyclic directed graph, where $V$ denotes the set of all the vertices (or points) in $G$ and $E$ denotes the set of all the edges in $G$. In this paper, a \emph{path} in $G$ is treated as a set of concatenated edges. For $k$ paths $\beta_1, \beta_2, \ldots, \beta_k$ in $G(V, E)$, we say these paths {\it merge}~\cite{ha2011} at an edge $e \in E$ if
\begin{enumerate}
\item $e \in \bigcap_{i=1}^k \beta_i$,
\item there are at least two distinct edges $f, g \in E$ such that $f, g$ are immediately ahead of $e$ on some $\beta_i, \beta_j$, $i \neq j$, respectively.
\end{enumerate}
\noindent We call the maximal subpath that starts with $e$ and that is shared by all $\beta_i$'s (i.e., $e$ together with the subsequent concatenated edges shared by all $\beta_i$'s until some $\beta_i$ branches off) {\it merged subpath} (or simply {\it merging}) by all $\beta_i$'s at $e$; see Figure~\ref{picmergings} for a quick example.

For any two vertices $u, v \in V$, we call any set consisting of the maximum number of pairwise edge-disjoint directed paths from $u$ to $v$ a set of {\em Menger's paths} from $u$ to $v$. By Menger's theorem~\cite{me1927}, the cardinality of Menger's paths from $u$ to $v$ is equal to the min-cut between $u$ and $v$. Here, we remark that the Ford-Fulkerson algorithm~\cite{fo1956} can find the min-cut and a set of Menger's paths from $u$ to $v$ in polynomial time.

Assume that $G(V, E)$ has $l$ sources $S_1, S_2, \ldots, S_l$ and $l$ distinct sinks $R_1, R_2, \ldots, R_l$. For $i=1, 2, \ldots, l$, let $c_i$ denote the min-cut between $S_i$ and $R_i$, and let $\alpha_i=\{\alpha_{i, 1}, \alpha_{i, 2}, \ldots, \alpha_{i, c_i}\}$ denote a set of Menger's paths from $S_i$ to $R_i$. We are interested in the number of mergings among paths from different $\alpha_i$'s, denoted by $|G|_\mcm(\alpha_1,\alpha_2,\ldots,\alpha_l)$. In this paper we will count the number of mergings {\bf without} multiplicity: all the mergings at the same edge $e$ will be counted as one merging at $e$.
\begin{figure}
\psfrag{be1}{$\beta_1$} \psfrag{be2}{$\beta_2$} \psfrag{be3}{$\beta_3$}
\psfrag{aaa}{$A$} \psfrag{bbb}{$B$}\psfrag{ccc}{$C$}\psfrag{ddd}{$D$}
\centerline{\includegraphics[width=0.48\textwidth]{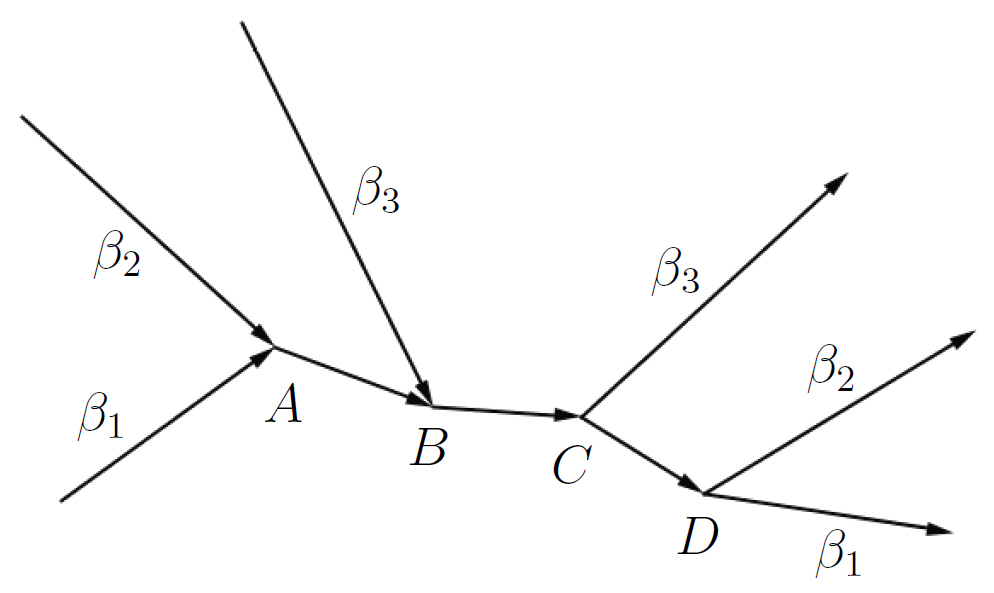}}
\caption{Paths $\beta_1, \beta_2$ merge at edge $A \to B$ and at merged subpath (or merging) $A \to B \to C \to D$, and paths $\beta_1, \beta_2, \beta_3$ merge at edge $B \to C$ and at merged subpath (or merging) $B\to C$.} \label{picmergings}
\end{figure}
The motivation for such consideration is more or less obvious in transportation networks: mergings among different groups of transportation paths can cause congestions, which may either decrease the whole network throughput or incur unnecessary cost. The connection between the number of mergings and the encoding complexity in computer networks, however, is a bit more subtle, which can be best illustrated by the following three examples in network coding theory (for a brief introduction to this theory, see~\cite{ye2006}).

The first example is the famous ``butterfly network''~\cite{li2003}. As depicted in Figure~\ref{picthreeeg}(a), for the purpose of transmitting messages $a, b$ simultaneously from the sender $S$ to the receivers $R_1, R_2$, network encoding has to be done at node $C$. Another way to interpret the necessity of network coding at $C$ (for the simultaneous transmission to $R_1$ and $R_2$) is as follows: If the transmission to $R_2$ is ignored, Menger's paths $S \to A \to R_1$ and $S \to B \to C \to D \to R_1$ can be used to transmit messages $a, b$ from $S$ to $R_1$; if the transmission to $R_1$ is ignored, Menger's paths $S \to A \to C \to D \to R_2$ and $S \to B \to R_2$ can be used to transmit messages $a, b$ from $S$ to $R_2$. For the simultaneous transmission to $R_1$ and $R_2$, merging by these two groups of Menger's paths at $C \to D$ becomes a ``bottle neck'', therefore network coding at $C$ is required to avoid the possible congestions.

The second example is a variant of the classical butterfly network (see Example 17.2 of~\cite{ye2008}; cf. the two-way channel in Page $519$ of~\cite{co2006}) with two senders and two receivers, where the sender $S_1$ is attached to the receiver $R_2$ to form a group and the sender $S_2$ is attached to the receiver $R_1$ to form the other group. As depicted in Figure~\ref{picthreeeg}(b), the two groups wish to exchange messages $a$ and $b$ through the network. Similarly as in the first example, the edge $A \to B$ is where the Menger's paths $S_1 \to A \to B \to R_1$ and $S_2 \to A \to B \to R_2$ merge with each other, which is a bottle neck for the simultaneous transmission of messages $a, b$. The simultaneous transmission is achievable if upon receiving the messages $a$ and $b$, network encoding is performed at the node $A$ and the newly derived message $a+b$ is sent over the channel $AB$.

The third example is concerned with two sessions of unicast in a network \cite{so2010}. As shown in Figure~\ref{picthreeeg}(c), the sender $S_1$ is to transmit message $a$ to the receiver $R_1$ using path $S_1 \to A \to B \to E \to F \to C \to D \to R_1$. And the sender $S_2$ is to transmit message $b$ to the receiver $R_2$ using two Menger's paths $S_2 \to A \to B \to C \to D \to R_2$ and $S_2 \to E \to F \to R_2$. Since mergings $A \to B$, $C \to D$ and $E \to F$ become bottle necks for the simultaneous transmission of messages $a$ and $b$, network coding at these bottle necks, as shown in Figure~\ref{picthreeeg}(c), is performed to ensure the simultaneous message transmission.

Generally speaking, for a network with multiple groups of Menger's paths, each of which is used to transmit a set of messages to a particular sink, network encoding is needed at mergings by different groups of Menger's paths. As a result, the number of mergings is the number of network encoding operations required in the network. So, we are interested in the number of mergings among different groups of Menger's paths in such networks.
\begin{figure}
\psfrag{axx}{$\footnotesize\textrm{(a)}$}\psfrag{bxx}{$\footnotesize\textrm{(b)}$}\psfrag{cxx}{$\footnotesize\textrm{(c)}$}
\psfrag{sss}{$S$}\psfrag{s1a}{$S_1$}\psfrag{s2a}{$S_2$}
\psfrag{r1a}{$R_1$}\psfrag{r2a}{$R_2$}
\psfrag{aaa}{$a$}\psfrag{bbb}{$b$}
\psfrag{abx}{$a\hspace{-0.1cm}+\hspace{-0.1cm}b$}
\psfrag{abb}{$a\hspace{-0.1cm}+\hspace{-0.1cm}2b$}
\psfrag{aandb}{$a\hspace{-0.1cm}+\hspace{-0.1cm}b$}
\psfrag{AAA}{$A$}\psfrag{BBB}{$B$}\psfrag{CCC}{$C$}
\psfrag{DDD}{$D$}\psfrag{EEE}{$E$}\psfrag{FFF}{$F$}
  \centering
  \includegraphics[width=0.85\textwidth]{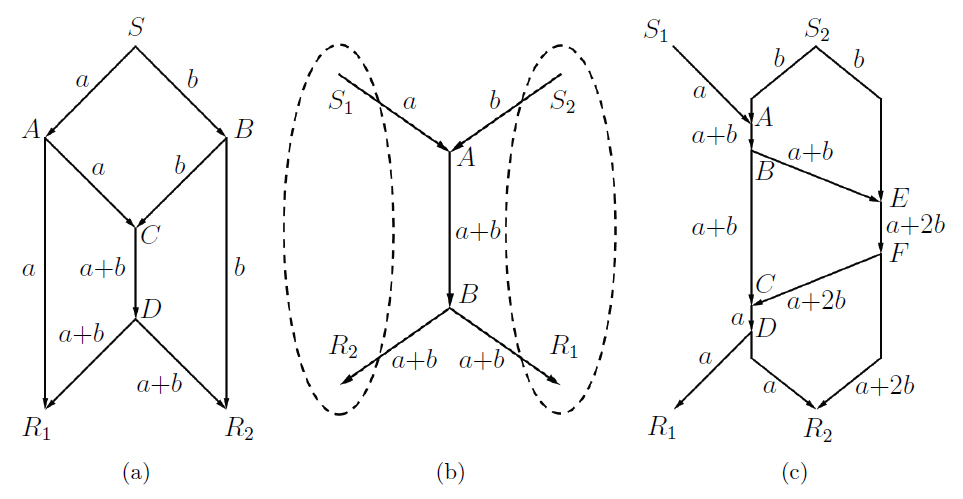}
  \caption{(a) Network coding on the butterfly network (b) Network coding on a variant of the butterfly network (c) Network coding on two sessions of unicast}\label{picthreeeg}
\end{figure}

For the case when all sources in $G$ are in fact identical, $M^*(G)$ is defined as the minimum of $|G|_\mcm(\alpha_1,\alpha_2,\ldots,\alpha_l)$ over all possible Menger's path sets $\alpha_i$'s, $i=1, 2, \ldots, l$, and $\mathcal{M}^*(c_1, c_2, \ldots, c_l)$ is defined as the supremum of $M^*(G)$ over all possible choices of such $G$. It is clear that $M^*(G)$ is the least number of network encoding operations required for a given $G$, and $\mathcal{M}^*(c_1, c_2, \ldots, c_l)$ is the largest such number among all such $G$ (with the min-cut between the $i$-th pair of source and sink being $c_i$). As for $\mathcal{M}^*$, the authors of~\cite{fr2006} used the idea of ``subtree decomposition'' to first prove that
$$
\mathcal{M}^*(\underbrace{2, 2, \ldots, 2}_l)=l-1.
$$
Although their idea seems to be difficult to generalize to other parameters, it does allow us to gain deeper understanding about the topological structure of the graphs achieving $l-1$ mergings for this special case. It was first shown in~\cite{la2006} that $\mathcal{M}^*(c_1, c_2)$ is finite for all $c_1, c_2$ (see Theorem $22$ in~\cite{la2006}), and subsequently $\mathcal{M}^*(c_1, c_2, \ldots, c_l)$ is finite for all $c_1, c_2, \ldots, c_l$.

For the case when all sources in $G$ are distinct, $M(G)$ is defined as the minimum of $|G|_\mcm(\alpha_1,\alpha_2,\ldots,\alpha_l)$ over all possible Menger's path sets $\alpha_i$'s, $i=1, 2, \ldots, n$, and $\mathcal{M}(c_1, c_2, \ldots, c_l)$ is defined as the supremum of $M(G)$ over all possible choices of such $G$. Again, the encoding idea for the second example can be easily generalized to networks, where each receiver is attached to all senders except its associated one. It is clear that the number of mergings is a tight upper bound for the number of network encoding operations required. For networks with several unicast sessions, in~\cite{so2010}, an upper bound for the encoding complexity of a network with two unicast sessions was given, as a result of a more general treatment (to networks with two multicast sessions) by the authors. It is easy to see that for networks with multiple unicast sessions (straightforward generalizations of the third example), $\mathcal{M}$ with appropriate parameters can serve as an upper bound on network encoding complexity. It was first conjectured that $\mathcal{M}(c_1, c_2, \ldots, c_l)$ is finite in~\cite{ta2003}. More specifically the authors proved that (see Lemma $10$ in~\cite{ta2003}) if $\mathcal{M}(c_1, c_2)$ is finite for all $c_1, c_2$, then $\mathcal{M}(c_1, c_2, \ldots, c_l)$ is finite as well. Here, we remark that we have rephrased the work in~\cite{fr2006, la2006, ta2003}, since all of them are done using very different languages from ours.

In~\cite{ha2011}, we have shown that for any $c_1, c_2, \ldots, c_l$, $\mathcal{M}^*(c_1, c_2, \ldots, c_l)$, $\mathcal{M}(c_1, c_2, \ldots, c_l)$ are both finite, and we further studied the behaviors of $\mathcal{M}^*, \mathcal{M}$ as functions of the min-cuts. In this paper, further continuing the work in~\cite{ha2011}, we compute exact values of and give tighter bounds on $\mathcal{M}^*$ and $\mathcal{M}$ for certain parameters.

For a path $\beta$ in $G$, let $h(\beta), t(\beta)$ denote \emph{head} (or \emph{starting point}) and \emph{tail} (or \emph{ending point}) of path $\beta$, respectively; let $\beta[u,v]$ denote the subpath of $\beta$ with the starting point $u$ and the ending point $v$. For two distinct paths $\xi, \eta$ in $G$, we say $\xi$ is {\em smaller} than $\eta$ (or, $\eta$ is {\em larger} than $\xi$) if there is a directed path from $t(\xi)$ to $h(\eta)$; if $\xi, \eta$ and the connecting path from $t(\xi)$ to $h(\eta)$ are subpaths of path $\beta$, we say $\xi$ is {\em smaller} than $\eta$ on $\beta$. Note that this definition also applies to the case when paths degenerate to vertices/edges; in other words, in the definition, $\xi, \eta$ or the connecting path from $t(\xi)$ to $h(\eta)$ can be vertices/edges in $G$, which can be viewed as degenerated paths. If $t(\xi)=h(\eta)$, we use $\xi \circ \eta$ to denote the path obtained by concatenating $\xi$ and $\eta$ subsequently. For a set of vertices $v_1, v_2, \ldots, v_k$ in $G$, define $G|v_1, \ldots, v_k)$ to be the subgraph of $G$ induced on the set of vertices, each of which is smaller or equal to some $v_i$, $i=1, 2, \ldots, k$.

$G$ is said to be a {\em $(c_1, c_2, \ldots, c_l)$-graph} if every edge in $G$ belongs to some $\alpha_i$-path, or, in loose terms, all $\alpha_i$'s ``cover'' the whole $G$. For a $(c_1, c_2, \ldots, c_l)$-graph, the number of mergings is the number of vertices whose in-degree is at least $2$. It is clear that to compute $\mathcal{M}(c_1, c_2, \ldots, c_l)$ ($\mathcal{M}^*(c_1, c_2, \ldots, c_l)$), it is enough to consider all the $(c_1, c_2, \ldots, c_l)$-graphs with distinct (identical) sources. For a $(c_1, c_2, \ldots, c_l)$-graph $G$, we say $\alpha_i$ is {\em reroutable} if there exists a different set of Menger's paths $\alpha'_i$ from $S_i$ to $R_i$, and we say $G$ is {\em reroutable} (or alternatively, there is a \emph{rerouting} in $G$), if some $\alpha_i$, $i=1, 2, \ldots, l$, is reroutable. Note that for a non-reroutable $G$, the choice of $\alpha_i$'s is unique, so we  often write $|G|_\mcm(\alpha_1, \alpha_2, \ldots, \alpha_l)$ as $|G|_\mcm$ for notational simplicity.

Now, for a fixed $i$, reverse the directions of edges that only belong to $\alpha_i$ to obtain a new graph $G'$. For any two mergings $\lambda, \mu$, if there exists a directed path in $G'$ from the head (or tail) of $\lambda$ to the head (or tail) of $\mu$, we say the head (or tail) of $\lambda$ \emph{semi-reaches} the head (or tail) of $\mu$ against $\alpha_i$, or alternatively, $\lambda$ \emph{semi-reaches} against $\alpha_i$ from head (or tail) to head (or tail). It is easy to check that $G$ is reroutable if and only if there exists $i$ and a merging $\lambda$ such that $\lambda$ semi-reaches itself against $\alpha_i$ from head to head, which is equivalent to the condition that there exists $i$ and a merging $\eta$ such that $\eta$ semi-reaches itself against $\alpha_i$ from tail to tail.

\begin{figure}
\psfrag{s1a}{$S_1$} \psfrag{s2a}{$S_2$} \psfrag{r1a}{$R_1$} \psfrag{r2a}{$R_2$}
\psfrag{phi1}{$\phi_1$} \psfrag{phi2}{$\phi_2$} \psfrag{psi1}{$\psi_1$} \psfrag{psi2}{$\psi_2$}
\psfrag{gg1}{$\gamma_1$} \psfrag{gg2}{$\gamma_2$} \psfrag{gg3}{$\gamma_3$} \psfrag{gg4}{$\gamma_4$}
\centerline{\includegraphics[width=0.4\textwidth]{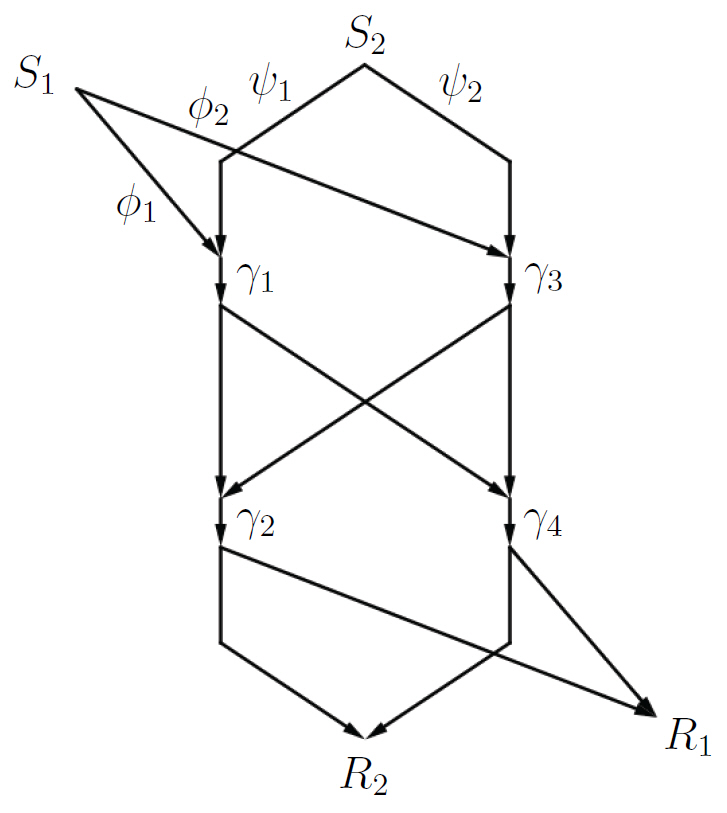}}
\caption{An example of a reroutable graph} \label{picGv}
\end{figure}
\begin{exmp}
For the graph depicted in Figure~\ref{picGv}, the source $S_1$ is connected to the sink $R_1$ by a group of Menger's paths
\begin{align*}
\phi=&\{\phi_1, \phi_2\}=\{S_1 \to h(\gamma_1) \to t(\gamma_1) \to h(\gamma_4) \to t(\gamma_4) \to R_1, \\
&S_1 \to h(\gamma_3) \to t(\gamma_3) \to h(\gamma_2) \to t(\gamma_2) \to R_1\}
\end{align*}
and the source $S_2$ is connected to the sink $R_2$ by a group of Menger's paths
\begin{align*}
\psi=&\{\psi_1, \psi_2\}=\{S_2 \to h(\gamma_1) \to t(\gamma_1) \to h(\gamma_2) \to t(\gamma_2) \to R_2, \\
&S_2 \to h(\gamma_3) \to t(\gamma_3) \to h(\gamma_4) \to t(\gamma_4) \to R_2\}.
\end{align*}
Then $\gamma_1, \gamma_2, \gamma_3, \gamma_4$ are mergings by $\phi$-paths and $\psi$-paths. $\gamma_1$, $\gamma_3$ are smaller than $\gamma_2$ and $\gamma_4$. $G|S_1, S_2)$ only consists of two isolated vertices $S_1, S_2$; $G|h(\gamma_1), h(\gamma_3))$ is the subgraph of $G$ induced on the set of vertices $\{S_1, S_2, h(\gamma_1), h(\gamma_3)\}$; $G|t(\gamma_2), t(\gamma_4))$ is the subgraph of $G$ induced on the set of vertices
$$
\{S_1, S_2, h(\gamma_1), h(\gamma_3), t(\gamma_1), t(\gamma_3), h(\gamma_2), h(\gamma_4), t(\gamma_2), t(\gamma_4)\};
$$
and $G|R_1, R_2)$ is just $G$ itself.

The group of Menger's paths $\phi$ is reroutable, since there exists another group of Menger's paths
\begin{align*}
\phi'=&\{\phi'_1,\phi'_2\}=\{S_1 \to h(\gamma_1) \to t(\gamma_1) \to h(\gamma_2) \to t(\gamma_2) \to R_1, \\
&S_1 \to h(\gamma_3) \to t(\gamma_3) \to h(\gamma_4) \to t(\gamma_4) \to R_1\}
\end{align*}
from $S_1$ to $R_1$. Similarly, $\psi$ is also reroutable. It is easy to check, by definition, that $\gamma_2$ semi-reaches $\gamma_4$ against $\psi$ from head to head, $\gamma_1$ semi-reaches $\gamma_4$ against $\psi$ from tail to head, and $\gamma_4$ semi-reaches itself against $\phi$ (or alternatively $\psi$) from head to head. Hence, $G$ is reroutable.
\end{exmp}

\section{Related Sequences}

\subsection{Merging sequences}

For any $m, n$, consider the following procedure to ``draw'' an $(m, n)$-graph: for ``fixed'' edge-disjoint paths $\psi_1, \psi_2, \ldots, \psi_n$ from $S_2$ to $R_2$, we extend edge-disjoint paths $\phi_1, \phi_2, \ldots, \phi_m$ from $S_1$ to merge with $\psi$-paths until we reach $R_1$. More specifically, the procedure of extending $\phi$-paths is done step by step, and for each step, we choose to extend one of $m$ $\phi$-paths to merge with one of $n$ $\psi$-paths. Thus for each step, we have $mn$ ``strokes'' to choose from the following set
$$
\{(\phi_1, \psi_1), (\phi_1, \psi_2), \ldots, (\phi_m, \psi_{n-1}), (\phi_m, \psi_n)\},
$$
here, by ``drawing'' the \emph{path pair} $(\phi_i, \psi_j)$, we mean further extending path $\phi_i$ to merge with path $\psi_j$, while ensuring the new merged subpath is larger than any existing merged subpaths on path $\psi_j$. Apparently, the procedure, and thus the graph, is uniquely determined by the sequence of strokes (see Example~\ref{stack-graph}), which will be referred to as a {\em merging sequence} of this $(m, n)$-graph. It is also easy to see that any $(m, n)$-graph can be generated by some merging sequence.

\begin{exmp} \label{stack-graph}
Consider the following two graphs in Figure~\ref{picmergingsequence} (here and hereafter, all the mergings in this paper are represented by solid dots instead). Listing the elements in the merging sequence, graph (a) can be described by $[(\phi_1, \psi_2), (\phi_2, \psi_1)]$, or alternatively $[(\phi_2, \psi_1), (\phi_1, \psi_2)]$. When the context is clear, we often omit $\phi, \psi$ in the merging sequence for notational \mbox{simplicity}. For example, graph (b) can be described by a merging sequence $[(1, 1), (2, 1), (2, 2), (3, 2)]$. Note that it cannot be described by $[(1, 1), (2, 1), (3, 2), (2, 2)]$, since $(3, 2)$ (or, more precisely, the merging corresponding to $(3, 2)$) is larger than $(2, 2)$ on $\psi_2$.
\end{exmp}
\begin{figure}
\psfrag{axx}{$\footnotesize\textrm{(a)}$}\psfrag{bxx}{$\footnotesize\textrm{(b)}$}\psfrag{cxx}{$\footnotesize\textrm{(c)}$}
\psfrag{s1a}{$S_1$}\psfrag{s2a}{$S_2$}
\psfrag{r1a}{$R_1$}\psfrag{r2a}{$R_2$}
\psfrag{s1b}{$S_1$}\psfrag{s2b}{$S_2$}
\psfrag{r1b}{$R_1$}\psfrag{r2b}{$R_2$}
\psfrag{phi1}{$\phi_1$}\psfrag{phi2}{$\phi_2$}\psfrag{phi3}{$\phi_3$}
\psfrag{psi1}{$\psi_1$}\psfrag{psi2}{$\psi_2$}\psfrag{psi3}{$\psi_3$}
\centerline{\includegraphics[width=0.58\textwidth]{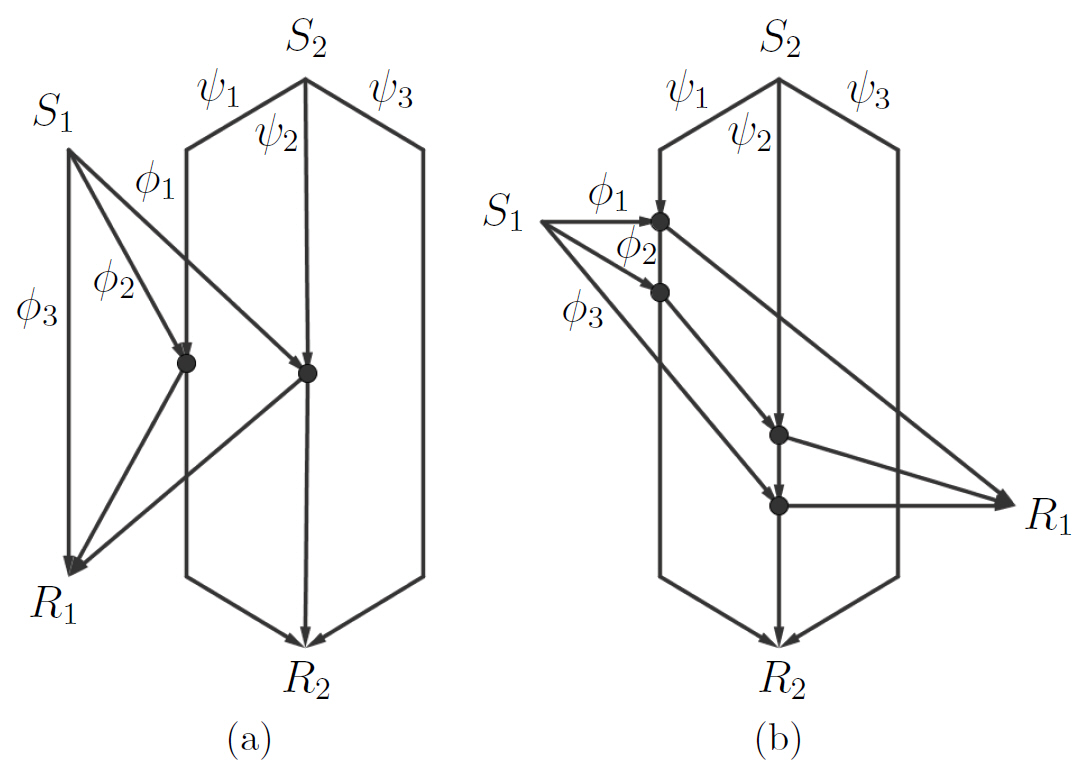}}
\caption{Two examples of merging sequences}
\label{picmergingsequence}
\end{figure}

\subsection{AA-sequences}  \label{AA-sequences}

Consider a non-reroutable $(m, n)$-graph $G$ with two sources $S_1, S_2$, two distinct sinks $R_1, R_2$, a set of Menger's paths $\phi=\{\phi_1, \phi_2, \ldots, \phi_m\}$ from $S_1$ to $R_1$, and a set of Menger's paths $\psi=\{\psi_1, \psi_2, \ldots, \psi_n\}$ from $S_2$ to $R_2$.

For the case when $S_1$ and $S_2$ are distinct, consider the following procedure on $G$. Starting from $S_1$, go along path $\phi_i$ until we reach a merged subpath (or more precisely, the terminal vertex of a merged subpath), we then go against the associated $\psi$-path (corresponding to the merged subpath just visited) until we reach another merged subpath, we then go along the associated $\phi$-path, $\ldots$ Continue this procedure (of alternately going along $\phi$-paths or going against $\psi$-paths until we reach a merged subpath) in the same manner as above, then the fact that $G$ is non-reroutable and acyclic guarantees that eventually we will reach $R_1$ or $S_2$. By sequentially listing all the terminal vertices of any merged subpaths visited, such a procedure produces a {\em $\phi_i$-AA-sequence}. Apparently, there are $m$ $\phi$-AA-sequences.

Similarly, consider the following procedure on $G$. Starting from $R_2$, go against path $\psi_j$ until we reach a merged subpath, we then go along the associated $\phi$-path (corresponding to the merged subpath just visited) until we reach another merged subpath, we then go against the associated $\psi$-path, $\ldots$ Continue this procedure in the same manner, again, eventually, we are guaranteed to reach $R_1$ or $S_2$. By sequentially listing all the terminal vertices of any merged subpaths visited, such a procedure produces a {\em $\psi_j$-AA-sequence}. Apparently, there are $n$ $\psi$-AA-sequences.

The {\em length} of an AA-sequence $\pi$, denoted by $\textrm{Length}(\pi)$, is defined to be the number of terminal vertices of merged subpaths visited during the procedure. Since each such terminal vertex in an AA-sequence is associated with a path pair, equivalently, the length of an AA-sequence can be also defined as the number of the associated path pairs. For the purpose of computing $\mathcal{M}(m, n)$, we can assume that each Menger's path in $G$ merges at least once, which implies that each AA-sequence is of positive length.

For the case when $S_1$ and $S_2$ are identical, by Proposition $3.6$ in~\cite{ha2011}, we can restrict our attention to the case when $m=n$. For the purpose of computing $\mathcal{M}^*(n, n)$, by the proof of Proposition $3.6$ in~\cite{ha2011}, we can assume that paths $\phi_i$ and $\psi_i$ share a \emph{starting subpath} (a maximal shared subpath by $\phi_i$ and $\psi_i$ starting from the source) for $i=1, 2, \ldots, n$, and due to non-reroutability of $G$, $\phi_n$ and $\psi_1$ do not merge with any other path. Then, $\psi$-AA-sequences and their lengths can be similarly defined as in the case when $S_1$ and $S_2$ are distinct, except that we have to replace ``merged subpath'' by ``merged subpath or starting subpath''. (Here, let us note that the procedure of defining $\phi$-AA-sequences does NOT carry over.) It can be checked that the existence of $m$ starting subpaths implies that any $\psi$-AA-sequence is of positive length and will always terminate at $R_1$.

It turns out that the lengths of AA-sequences are related to the number of mergings in $G$.
\begin{lem} \label{AA-to-number-of-merings}
For a non-reroutable $(m, n)$-graph $G$ with distinct sources,
$$|G|_\mcm=\frac{1}{2}\sum_{\pi} \textrm{Length}(\pi);$$
for a non-reroutable $(n, n)$-graph $G$ with identical sources and $n$ starting subpaths,
$$
|G|_\mcm=\frac{1}{2}\left(\sum_{\pi} \textrm{Length}(\pi)-n\right),
$$
where the two summations above are over the all the possible AA-sequences.
\end{lem}

\begin{exmp}  \label{AA-examples}
Consider the two graphs in Figure~\ref{picaasequences}. Let ``$\Rightarrow$'' and ``$\Leftarrow$'' denote ``go along'' and ``go against'', respectively. In graph (a), sequentially listing the terminal vertices of merged subpaths visited during the procedure, two $\phi$-AA-sequences can be represented by $S_1 \Rightarrow h(\gamma_1) \Leftarrow S_2$ and $S_1 \Rightarrow h(\gamma_2) \Leftarrow t(\gamma_1) \Rightarrow h(\gamma_5) \Leftarrow t(\gamma_4) \Rightarrow R_1$. Similarly, two $\psi$-AA-sequences can be represented by $R_2 \Leftarrow t(\gamma_3) \Rightarrow R_1$ and $R_2 \Leftarrow t(\gamma_5) \Rightarrow h(\gamma_3) \Leftarrow t(\gamma_2) \Rightarrow h(\gamma_4) \Leftarrow S_2$. One also checks that the number of mergings is $5$, which is half of $(1+4+1+4)$, the sum of lengths of all AA-sequences.

In graph (b), sequentially listing the terminal vertices of merged subpaths and starting subpaths visited during the procedure, three $\psi$-AA-sequences can be represented by $R_2\Leftarrow t(\omega_1) \Rightarrow h(\gamma_1)\Leftarrow t(\omega_2)\Rightarrow h(\gamma_4)\Leftarrow t(\gamma_3)\Rightarrow R_1$, $R_2\Leftarrow t(\gamma_2)\Rightarrow R_1$ and $R_2\Leftarrow t(\gamma_4)\Rightarrow h(\gamma_2) \Leftarrow t(\gamma_1)\Rightarrow h(\gamma_3)\Leftarrow t(\omega_3)\Rightarrow R_1$. One also checks the number of mergings is $4$, which is half of $(5+1+5-3)$.
\begin{figure}
\psfrag{axx}{$\footnotesize\textrm{(a)}$}\psfrag{bxx}{$\footnotesize\textrm{(b)}$}\psfrag{cxx}{$\footnotesize\textrm{(c)}$}
\psfrag{s1a}{$S_1$} \psfrag{s2a}{$S_2$} \psfrag{sss}{$S$}
\psfrag{r1a}{$R_1$} \psfrag{r2a}{$R_2$}
\psfrag{phi1}{$\phi_1$} \psfrag{phi2}{$\phi_2$} \psfrag{phi3}{$\phi_3$} \psfrag{psi1}{$\psi_1$} \psfrag{psi2}{$\psi_2$}\psfrag{psi3}{$\psi_3$}
\psfrag{pi1}{$\omega_1$} \psfrag{pi2}{$\omega_2$}\psfrag{pi3}{$\omega_3$}
\psfrag{g1a}{$\gamma_1$} \psfrag{g2a}{$\gamma_2$} \psfrag{g3a}{$\gamma_3$} \psfrag{g4a}{$\gamma_4$} \psfrag{g5a}{$\gamma_5$}
\centerline{\includegraphics[width=0.62\textwidth]{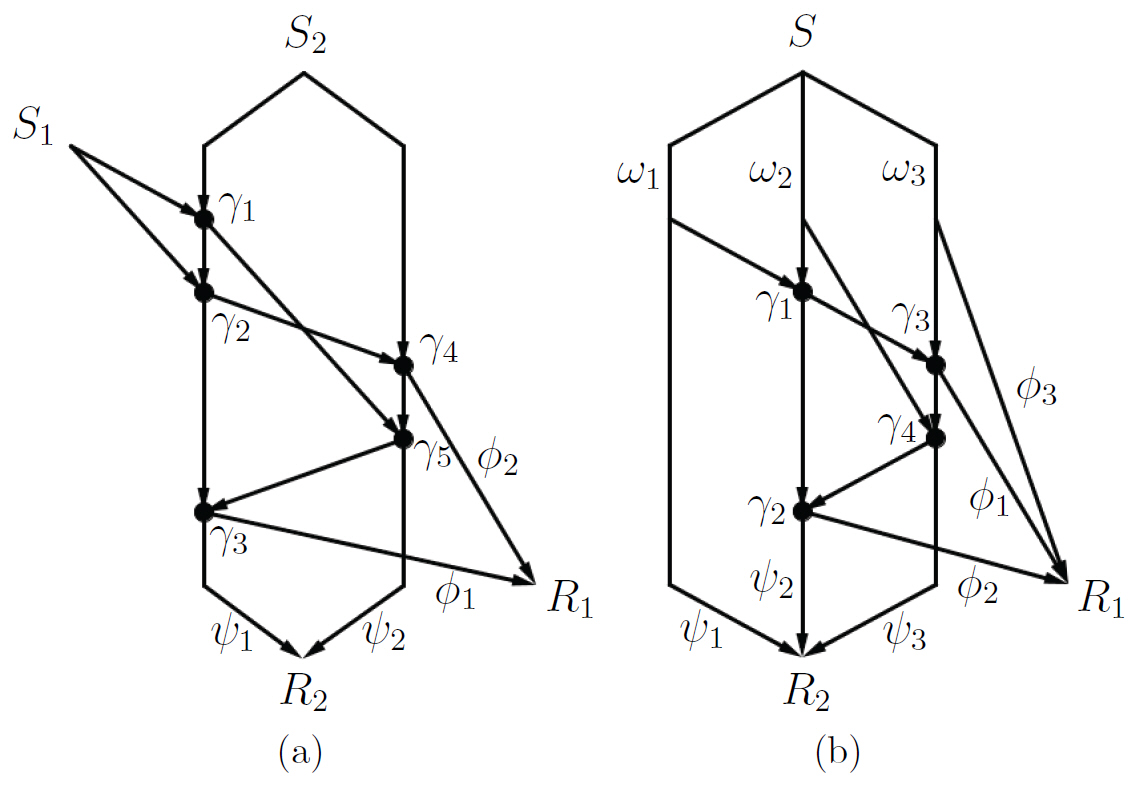}}
\caption{Two examples of AA-sequences}
\label{picaasequences}
\end{figure}
\end{exmp}
\begin{lem}\label{lengthofAAsequence}
The shortest $\phi$-AA-sequence ($\psi$-AA-sequence) is of length at most $1$.
\end{lem}

\begin{proof}
Suppose, by contradiction, that the shortest $\phi$-AA-sequence is of length at least $2$. Pick any $\phi$-path, say $\phi_{i_0}$. Assume that $\phi_{i_0}$ first merges with $\psi_{j_0}$ at merged subpath $\lambda_{i_0, j_0}$. Since the $\phi_{i_0}$-AA-sequence is of length at least $2$, there exists a $\phi$-path, say $\phi_{i_1}$, such that $\phi_{i_1}$ has a merged subpath, say $\mu_{i_1, j_0}$, smaller than $\lambda_{i_0, j_0}$ on $\psi_{j_0}$. Now assume that $\phi_{i_1}$ first merges with $\psi_{j_1}$ at merged subpath $\lambda_{i_1, j_1}$, then similarly there exists a $\phi$-path, say $\phi_{i_2}$, such that $\phi_{i_2}$ has a merged subpath, say $\mu_{i_2, j_1}$, smaller than $\lambda_{i_1, j_1}$ on $\psi_{j_1}$. Continue this procedure in the similar manner to \mbox{obtain} $\psi_{j_2}, \lambda_{i_2, j_2}, \phi_{i_3}, \mu_{i_3, j_2}, \psi_{j_3}, \lambda_{i_3, j_3}, \phi_{i_4}, \mu_{i_4, j_3},\ldots$ Apparently, there exists $k < l$ such that $i_l = i_k$. One then checks that
\begin{align*}
&\phi_{i_k}[h(\lambda_{i_k,j_k}), h(\mu_{i_l,j_{l-1}})]
\circ \psi_{j_{l-1}}[h(\mu_{i_l,j_{l-1}}), h(\lambda_{i_{l-1},j_{l-1}})]\
\circ \phi_{i_{l-1}}[h(\lambda_{i_{l-1},j_{l-1}}), h(\mu_{i_{l-1}, j_{l-2}})]\\
\circ& \psi_{j_{l-2}}[h(\mu_{i_{l-1}, j_{l-2}}), h(\lambda_{i_{l-2}, j_{l-2}})]
\circ \cdots \circ \phi_{i_{k+1}}[h(\lambda_{i_{k+1}, j_{k+1}}), h(\mu_{i_{k+1}, j_k})]
\circ \psi_{j_k}[h(\mu_{i_{k+1}, j_k}), h(\lambda_{i_k,j_k})]
\end{align*}
constitutes a cycle, which contradicts the assumption that $G$ is acyclic.

A parallel argument can be applied to the shortest $\psi$-AA-sequence.

\end{proof}

\begin{lem} \label{each-path-pair-at-most-once}
For a non-reroutable graph $G$, any path pair occurs at most once in any given AA-sequence.
\end{lem}

\begin{proof}
By contradiction, suppose that the same path pair occurs in an AA-sequence twice. As in the proof of Lemma 2.7 in~\cite{ha2011}, one can prove $G$ is reroutable, which is a contradiction.
\end{proof}

\begin{rem} It then immediately follows from Lemma~\ref{each-path-pair-at-most-once} that in a non-reroutable $(m,n)$-graph with distinct sources,
\begin{itemize}
\item the longest $\phi$-AA-sequence ($\psi$-AA-sequence) is of length at most $mn$;
\item any $\phi$-path ($\psi$-path) merges at most $mn$ times.
\end{itemize}
And in a non-reroutable $(m,m)$-graph with identical sources,
\begin{itemize}
\item the longest $\psi$-AA-sequence is of length at most $m^2$;
\item any $\phi$-path ($\psi$-path) merges at most $m^2$ times.
\end{itemize}
\end{rem}

\section{Exact Values} \label{exact-values-section}

In this section, we give exact values of $\mathcal{M}$ and $\mathcal{M}^*$ for certain special parameters.

\begin{thm} \label{ThreeProofs}
$$
\mathcal{M}(2,n)=3n-1.
$$
\end{thm}

\begin{proof}\ We first show that $\mathcal{M}(2,n) \geq 3n-1$. Consider the following $(2,n)$-graph specified by the following merging sequence (for a simple example, see Figure~\ref{pictwomm2}(a)): $\Omega=[\Omega_k: 1\le k\le 3n-1]$, where
\begin{displaymath}
\Omega_k=
\left\{
      \begin{array}{llll}
            ([i]_2, 1)     & \mathrm{if}\ k=3i-2 & \mathrm{for}\ 1\leq i\leq n, \\
            ([i]_2, i+1)   & \mathrm{if}\ k=3i-1 & \mathrm{for}\ 1\leq i \leq n-1,\\
            ([i+1]_2, i+1) & \mathrm{if}\ k=3i   & \mathrm{for}\ 1\leq i \leq n-1,\\
            ([n+1]_2, 1)   & \mathrm{if}\ k=3n-1. &
      \end{array}
\right.
\end{displaymath}
where $[x]_2=1$ when $x$ is odd, $[x]_2=2$ when $x$ is even.
\begin{figure}
\psfrag{axx}{$\footnotesize\textrm{(a)}$}\psfrag{bxx}{$\footnotesize\textrm{(b)}$}\psfrag{cxx}{$\footnotesize\textrm{(c)}$}
\psfrag{s1a}{$S_1$}\psfrag{s2a}{$S_2$}
\psfrag{r1a}{$R_1$}\psfrag{r2a}{$R_2$}
\psfrag{A}{\scriptsize $A$}\psfrag{B}{\scriptsize $B$}
\psfrag{C}{\scriptsize $C$}\psfrag{D}{\scriptsize $D$}
\psfrag{E}{\scriptsize $E$}\psfrag{F}{\scriptsize $F$}
\psfrag{J}{\scriptsize $J$}\psfrag{K}{\scriptsize $K$}
\psfrag{L}{\scriptsize $L$}\psfrag{M}{\scriptsize $M$}
\psfrag{N}{\scriptsize $N$}
\psfrag{phi1}{$\phi_1$}\psfrag{phi2}{$\phi_2$}
\centering
\includegraphics[width=0.72\textwidth]{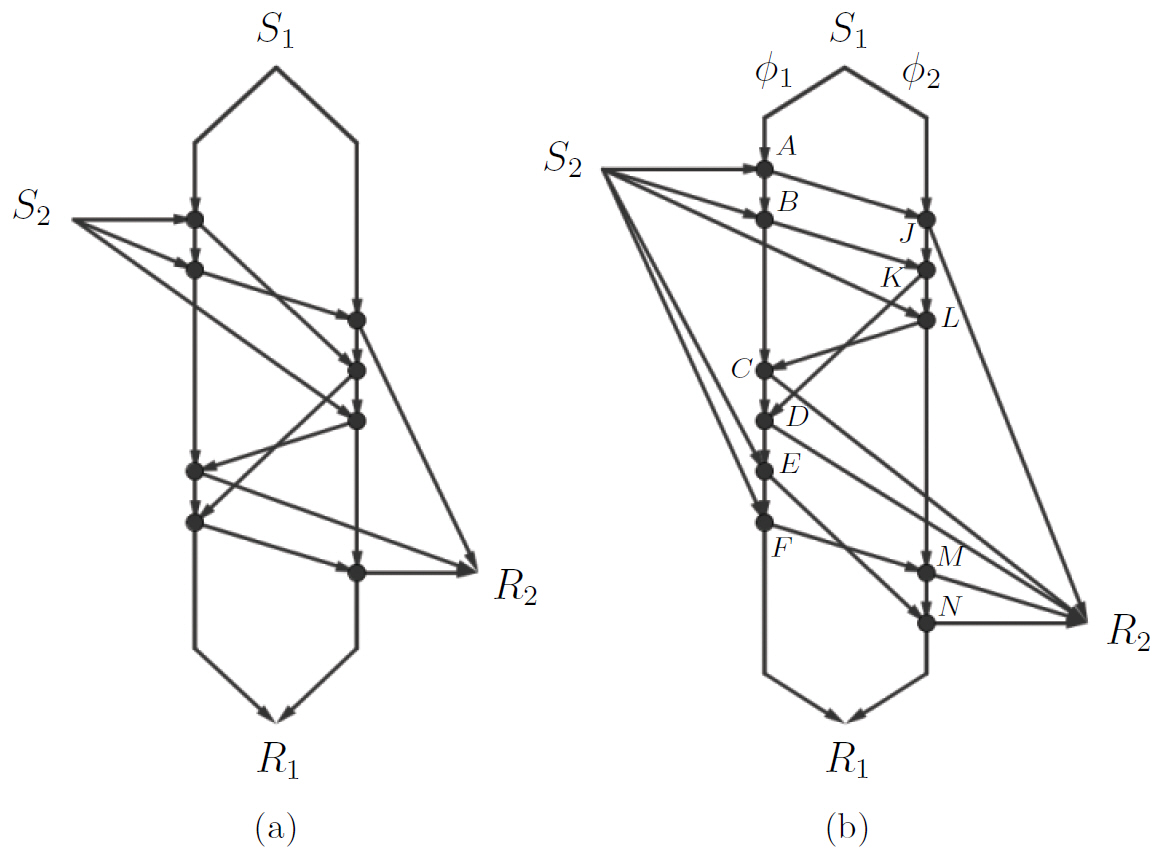}
\caption{(a) A non-reroutable $(2,3)$-graph with 8 mergings (b) An example of a $(2,5)$-graph}\label{pictwomm2}
\end{figure}

One checks that the above graph is non-reroutable with $3n-1$ mergings, which implies that $\mathcal{M}(2,n) \geq 3n-1$.

Next, we show that $\mcm(2,n)\le 3n-1$. Consider a non-reroutable $(2,n)$-graph $G$ with distinct sources $S_1, S_2$, sinks $R_1, R_2$, a set of Menger's paths $\phi=\{\phi_1,\phi_2\}$ from $S_1$ to $R_1$, and a set of Menger's paths $\psi=\{\psi_1,\psi_2,\ldots,\psi_n\}$ from $S_2$ to $R_2$.
Define
\begin{align*}
\Sigma=\{&(\lambda,\mu): \textrm{merging $\lambda$ is smaller than merging $\mu$ on some $\psi$-path}\\
&\textrm{and there is no other merging between them on this path}\}.
\end{align*}

Note that for any $(\lambda,\mu)\in \Sigma$, $\lambda,\mu$ must belong to different $\phi$-paths. We say $(\lambda,\mu)\in \Sigma$ is of \emph{type I}, if $\lambda$ belongs to $\phi_1$, and $(\lambda,\mu)\in \Sigma$ is of \emph{type II}, if $\lambda$ belongs to $\phi_2$. For any two different elements $(\lambda_1,\mu_1),(\lambda_2,\mu_2) \in \Sigma$. We say $(\lambda_1,\mu_1)\prec (\lambda_2,\mu_2)$ if either (they are of the same type and $\lambda_1$ is smaller than $\lambda_2$) or (they are of different types and $\lambda_1$ is smaller than $\mu_2$). One then checks that the relationship defined by $\prec$ is a strict total order.

Letting $x$ denote the number of elements in $\Sigma$, we define
$$
\Theta=(\Theta_1,\Theta_2,\ldots,\Theta_x)
$$
to be the sequence of the ordered (by $\prec$) elements in $\Sigma$. Now we consecutively partition $\Theta$ into $z$ ``\emph{medium-blocks}'' $B_1,B_2,\ldots,B_z$, and further consecutively partition each $B_i$ into $y_i$ ``\emph{mini-blocks}'' $B_{i,1},B_{i,2},\ldots,B_{i,y_i}$ (see Example~\ref{mediumandmini} for an example) such that

\begin{itemize}
\item for any $i,j$, the elements in $B_{i,j}$ are of the same type.

\item for any $i,j$, $B_{i,j}$ is \emph{linked} to $B_{i,j+1}$ in the following sense: let $(\lambda_{i,j},\mu_{i,j})$ denote the element with the largest second component in $B_{i,j}$ and let $(\lambda_{i,j+1},\mu_{i,j+1})$ denote the element with the smallest first component in $B_{i,j+1}$, then $\mu_{i,j}=\lambda_{i,j+1}$.

\item for any $i$, $B_{i,y_i}$ is not linked to $B_{i+1,1}$.
\end{itemize}

A mini-block is said to be a \emph{singleton} if it has only one element. We then have the following lemma, whose proof is omitted.

\begin{lem}  \label{two-singletons}
Between any two ``adjacent'' singletons (meaning there is no singleton between these two singletons) in a medium-block, there must exist a mini-block containing at least three elements.
\end{lem}

Letting $y$ denote the number of mini-blocks in $\Theta$ and $x_i$ denote the number of elements in medium-block $B_i$ for $1\le i\le z$, we then have
\begin{align*}
x&=x_1+x_2+\cdots+x_z,\\
y&=y_1+y_2+\cdots+y_z.
\end{align*}

Suppose there are $k$ singletons in $\Theta$, then by Lemma~\ref{two-singletons}, we can find $(k-1)$ mini-blocks, each of which has at least three elements. Hence, for $1\le i\le z$,
\begin{equation} \label{eq-0}
x_i\ge 1\cdot k+3\cdot (k-1)+2\cdot [y_i-k-(k-1)]=2y_i-1,
\end{equation}
which implies
\begin{equation} \label{eq-1}
x=\sum_{i=1}^z x_i\ge \sum_{i=1}^z (2y_i-1)=2y-z.
\end{equation}

For any two linked mini-blocks $B_{i,j}$ and $B_{i,j+1}$, let $(\lambda_{i,j},\mu_{i,j})$ denote the element with the largest second component in $B_{i,j}$, and let $(\lambda_{i,j+1},\mu_{i,j+1})$ denote the element with the smallest first component in $B_{i,j+1}$. By the definition (of two mini-blocks being linked), we have $\mu_{i,j}=\lambda_{i,j+1}$, which means $B_{i,j}$ and $B_{i,j+1}$ share a common merging. Together with the fact that each element in $\Sigma$ is a pair of mergings, this further implies that the number of mergings in $G$ is
\begin{equation}  \label{eq-2}
|G|_\mcm=2x-(y-z).
\end{equation}

Notice that $\lambda_{i,j}, \lambda_{i,j+1}, \mu_{i,j+1}$ belong to the same $\psi$-path, and furthermore, there exists only one $\phi$-path passing by both an element (more precisely, passing by both its mergings) in $B_{i, j}$ and an element in $B_{i, j+1}$. So, $n$, the number of $\psi$-paths in $G$ can be computed as
\begin{equation} \label{eq-3}
n=x-(y-z).
\end{equation}

It then follows from (\ref{eq-1}), (\ref{eq-2}), (\ref{eq-3}) and the fact $t\ge 1$ that
\begin{equation}  \label{eq-4}
n=x-y+z \ge (2y-z)-y+z=y
\end{equation}
and furthermore
\begin{equation} \label{eq-5}
|G|_\mcm=2x-y+z=2n+y-z \le 2n+n-1=3n-1,
\end{equation}
which establishes the theorem.
\end{proof}

\begin{exmp}\label{mediumandmini}
Consider the graph in Figure~\ref{pictwomm2}(b) and assume the context is as in the proof of Theorem~\ref{ThreeProofs}. Then we have,
$$
\Sigma=\{(A,J),(B,K),(L,C),(K,D),(F,M),(E,N)\}.
$$

Among all the elements in $\Sigma$, $(A,J)$, $(B,K)$, $(F,M)$ and $(E,N)$ are of type I, and $(L,C)$, $(K,D)$ are of type II. It is easy to check that
$$
\Theta=((A,J),(B,K),(K,D),(L,C),(E,N),(F,M)),
$$
which is partitioned into three mini-blocks $((A,J),(B,K))$, $((K,D),(L,C))$ and $((E,N),(F,M))$. The first mini-block is linked to the second one, but the second one is not linked to the third, so $\Theta$ is partitioned into two medium-blocks:
$$
((A,J),(B,K),(K,D),(L,C)) \mbox{ and } ((E,N),(F,M)).
$$
\end{exmp}

\begin{rem} \label{Pell}
The result in Theorem~\ref{ThreeProofs} in fact has already been proved in~\cite{ha2011} using a different approach. The proof in this paper, however, is more intrinsic in the sense that it reveals in greater depth the topological structure of non-reroutable $(2,n)$-graphs achieving $3n-1$ mergings, and further helps to determine the number of such graphs.

Assume a non-reroutable $(2,n)$-graph $G$ has $3n-1$ mergings. One then checks that in the proof of Theorem~\ref{ThreeProofs}, equalities hold for (\ref{eq-5}). It then follows that
\begin{itemize}
\item $t=1$, namely, there is only one medium-block in $\Theta$;
\item equalities hold necessarily for (\ref{eq-4}), (\ref{eq-1}) and eventually (\ref{eq-0}), which further implies that between two adjacent singletons, only one mini-block has three elements and any other mini-block has two elements.
\end{itemize}
Furthermore, one checks that
\begin{itemize}
\item for a mini-block with two elements $((\lambda_1,\mu_1),(\lambda_2,\mu_2))$, $\mu_2$ is smaller than $\mu_1$;
\item for a mini-block with three elements $((\lambda_1,\mu_1)$, $(\lambda_2,\mu_2)$, $(\lambda_3,\mu_3))$, either ($\mu_2$ is smaller than $\mu_3$ and $\mu_3$ is smaller than $\mu_1$) or ($\mu_3$ is smaller than $\mu_1$ and $\mu_1$ is smaller than $\mu_2$).
\end{itemize}

Assume that $G$ is ``reduced'' in the sense that, other than $S_1, S_2, R_1, R_2$, each vertex in $G$ is a terminal vertex of some merging. The properties above allow us to count how many reduced non-reroutable $(2,n)$-graphs (up to graph isomorphism) can achieve $3n-1$ mergings: suppose that there are $k$ ($1 \le k \le \floor{\frac{n+1}{2}}$) singletons in $G$, then necessarily, there are $(k-1)$ three-element mini-blocks and $(m-2k+1)$ two-element mini-blocks in $\Theta$. It can be checked that the number of ways for these $n$ mini-blocks to form $\Theta$ for some $(2,n)$-graph is $\binom{n}{2k-1}2^{k-1}$. This implies that the number of $(2,n)$-graph, whose $\Theta$ consists of $k$ singletons, $(k-1)$ three element mini-blocks and $(n-2k+1)$ two element mini-blocks, is $\binom{n}{2k-1}2^{k-1}$. Through a computation summing over all feasible $k$, the number of reduced non-reroutable $(2,n)$-graphs with $3n-1$ mergings can be computed as
$$\sum_{k=1}^{\floor{\frac{n+1}{2}}} \binom{n}{2k-1}2^{k-1} = \frac{1}{2\sqrt{2}}[(1+\sqrt{2})^n-(1-\sqrt{2})^n] = P_n,$$
where $P_n$ is the $n$-th Pell number~\cite{bi1975}.
\end{rem}

\begin{thm} \label{mstar44}
$$\mstar(4,4)=9.$$
\end{thm}

\begin{proof}

Consider a non-reroutable $(4, 4)$-graph $G$ with one source $S$, two sinks $R_1, R_2$, a set of Menger's paths $\phi=\{{\phi}_1,{\phi}_2,{\phi}_3,{\phi}_4\}$ from $S$ to $R_1$ and a set of Menger's paths $\psi=\{{\psi}_1,{\psi}_2,{\psi}_3,{\psi}_4\}$ from $S$ to $R_2$. As discussed in Section~\ref{AA-sequences}, we assume that $\psi_i$ and $\phi_i$ share a starting subpath $\omega_i$ from $S$ for $i=1,2,3,4$, and furthermore, we assume $\phi_4,\psi_1$ do not merge with any other paths, directly ``flowing'' to the sinks.

Consider the four $\psi$-AA-sequences, which will be referred to as $\pi_1, \pi_2, \pi_3, \pi_4$ in the following. It is easy to check that each $\pi_i$, $i=1, 2, 3, 4$, will be of odd length. Without loss of generality, assume that $\pi_4$ is the shortest such sequence, and thus by Lemma~\ref{lengthofAAsequence}, $\pi_4$ is of length $1$; let $\sigma$ be the merging associated with $\pi_4$. By Lemma~\ref{each-path-pair-at-most-once}, each $\pi_i$ can only be associated with each path pair $(\phi_j, \psi_k)$, $j=1,2,3$ and $k=2,3,4$ at most once. It then follows that excluding $\sigma$, each $\psi_k$, $k=2, 3, 4$, can only merge with each $\phi_j$, $j=1, 2, 3$, at most once. One then further checks that each $\pi_i$, $i=1, 2, 3$, can only be associated with $(\phi_j, \psi_k)$, $j=1,2,3$ and $k=2,3,4$ for $7$ times in total. By Lemma \ref{AA-to-number-of-merings}, we then derive
$$
|G|_{\mathcal{M}} \leq (9+7+7+1-4)/2=10.
$$

We next prove that $|G|_{\mathcal{M}}$ cannot be $10$. Suppose, by contradiction, that $|G|_{\mathcal{M}}$ is $10$. Then, necessarily, the longest $\psi$-AA-sequence, say $\pi_1$, will be of length $9$. It then follows that the two pairs, $(\phi_1, \psi_1)$ and $(\phi_4, \psi_4)$ must be associated with $\pi_1$. It also follows that $\pi_2,\pi_3$ must be of length $7$.

Now we prove that $\sigma$ belongs to $\phi_1$ and $\psi_4$. It suffices to prove that each of
$\psi_2, \psi_3, \phi_2, \phi_3$ cannot have four mergings. Suppose, by contradiction, there are four mergings in $\psi_2$, say $\mu_1,\mu_2,\mu_3,\mu_4$, in the ascending order; here $\mu_4$ is necessarily $\sigma$. Then, there are two mergings belonging to the same $\phi$-path, say $\phi_k$, $k\neq 2,4$. Now we consider two cases:

If $\mu_1, \mu_4 \in \phi_k$, then $h(\mu_1),t(\mu_1),h(\mu_4),t(\mu_4)$ must belong to different $\psi$-AA-sequences. Suppose $h(\mu_1) \in \pi_{j_1}$, $t(\mu_1) \in \pi_{j_2}$, $h(\mu_4) \in \pi_{j_3}$ and $t(\mu_4) \in \pi_{j_4}$, where $\{j_1, j_2, j_3, j_4\}=\{1, 2, 3, 4\}$. Note that $t(\mu_i)$ and $h(\mu_{i+1})$ belong to the same $\psi$-AA-sequence for $i=1,2,3$, $h(\mu_1)$ and $t(\omega_{2})$ belong to the same $\psi$-AA-sequence. This implies that $t(\omega_2) \in \pi_{j_1}$, $h(\mu_2) \in \pi_{j_2}$ and $t(\mu_3)\in \pi_{j_3}$. It then follows that $t(\mu_2),h(\mu_3)$ cannot belong to $\pi_{j_2}$ or $\pi_{j_3}$, so it must belong to $\pi_{j_1}$. On the other hand, either $\mu_2$ or $\mu_3$ must belong to $\phi_2$, the same $\phi$-path to which $t(\omega_2)$ belongs. Then $(\phi_2,\psi_2)$ occurs at least twice in $\pi_{j_1}$, which violates the Lemma~\ref{each-path-pair-at-most-once}.

If $\mu_2, \mu_4\in \phi_k$, then $h(\mu_2),t(\mu_2),h(\mu_4),t(\mu_4)$ must belong to different $\psi$-AA-sequences. Suppose $h(\mu_2) \in \pi_{j_1}$, $t(\mu_2) \in \pi_{j_2}$, $h(\mu_4) \in \pi_{j_3}$ and $t(\mu_4) \in \pi_{j_4}$, where $\{j_1, j_2, j_3, j_4\}=\{1, 2, 3, 4\}$. Note that $t(\mu_i)$ and $h(\mu_{i+1})$ belong to the same $\psi$-AA-sequence for $i=1,2,3$; $h(\mu_1)$ and $t(\omega_2)$ belong to the same $\psi$-AA-sequence. This implies that $t(\mu_1) \in \pi_{j_1}$, $h(\mu_3) \in \pi_{j_2}$ and $t(\mu_3) \in \pi_{j_3}$. In this case $\mu_3$ must belong to $\phi_2$, the same path to which $t(\omega_2)$ belongs. It then follows that $t(\omega_2)$ cannot belong to $\pi_{j_2},\pi_{j_3}$, so it must belong to $\pi_{j_1}$. But then we have $h(\mu_1)$, $t(\mu_1) \in \pi_{j_1}$, which violates the Lemma~\ref{each-path-pair-at-most-once}.

Combining the above two cases, we conclude that there cannot be four mergings on $\psi_2$. With a parallel argument applied to $\phi_2,\phi_3,\psi_3$, we conclude that there are four mergings on $\psi_4$, say $\gamma_1,\gamma_2,\gamma_3,\gamma_4=\sigma$, in the ascending order.

\begin{figure}
\psfrag{axx}{$\footnotesize\textrm{(a)}$}\psfrag{bxx}{$\footnotesize\textrm{(b)}$}\psfrag{cxx}{$\footnotesize\textrm{(c)}$}
\psfrag{dxx}{$\footnotesize\textrm{(d)}$}\psfrag{exx}{$\footnotesize\textrm{(e)}$}\psfrag{fxx}{$\footnotesize\textrm{(f)}$}
\psfrag{gxx}{$\footnotesize\textrm{(g)}$}\psfrag{hxx}{$\footnotesize\textrm{(h)}$}
\psfrag{sss}{$S$}\psfrag{r1a}{$R_2$}\psfrag{r2a}{$R_1$}
\psfrag{psi1}{\footnotesize $\psi_1$}\psfrag{psi2}{\footnotesize $\psi_2$}\psfrag{psi3}{\footnotesize $\psi_3$}\psfrag{psi4}{\footnotesize $\psi_4$}
\psfrag{phi1}{\footnotesize $\phi_1$}\psfrag{phi2}{\footnotesize $\phi_2$}\psfrag{phi3}{\footnotesize $\phi_3$}\psfrag{phi4}{\footnotesize $\phi_4$}
\psfrag{eta1}{\footnotesize $\lambda_1$}\psfrag{eta2}{\footnotesize $\lambda_2$}
\psfrag{gm1}{\footnotesize $\gamma_1$}\psfrag{gm2}{\footnotesize $\gamma_2$}\psfrag{gm3}{\footnotesize $\gamma_3$}\psfrag{gm4}{\footnotesize $\gamma_4$}
\psfrag{dtt}{\footnotesize $\mu$}
\centering
\includegraphics[width=0.9\textwidth]{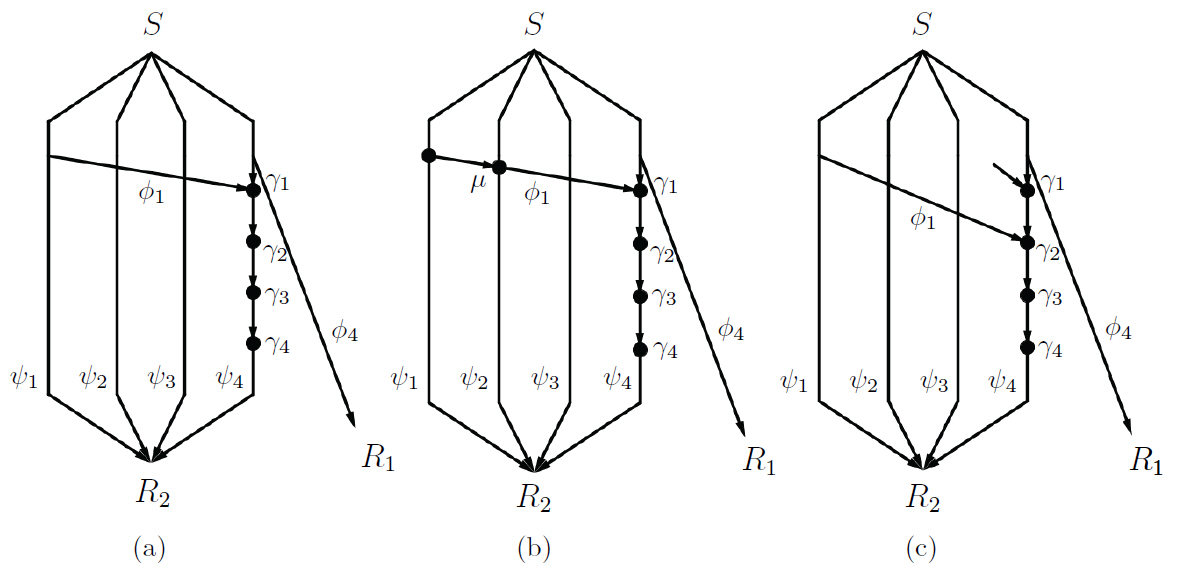}
\includegraphics[width=0.9\textwidth]{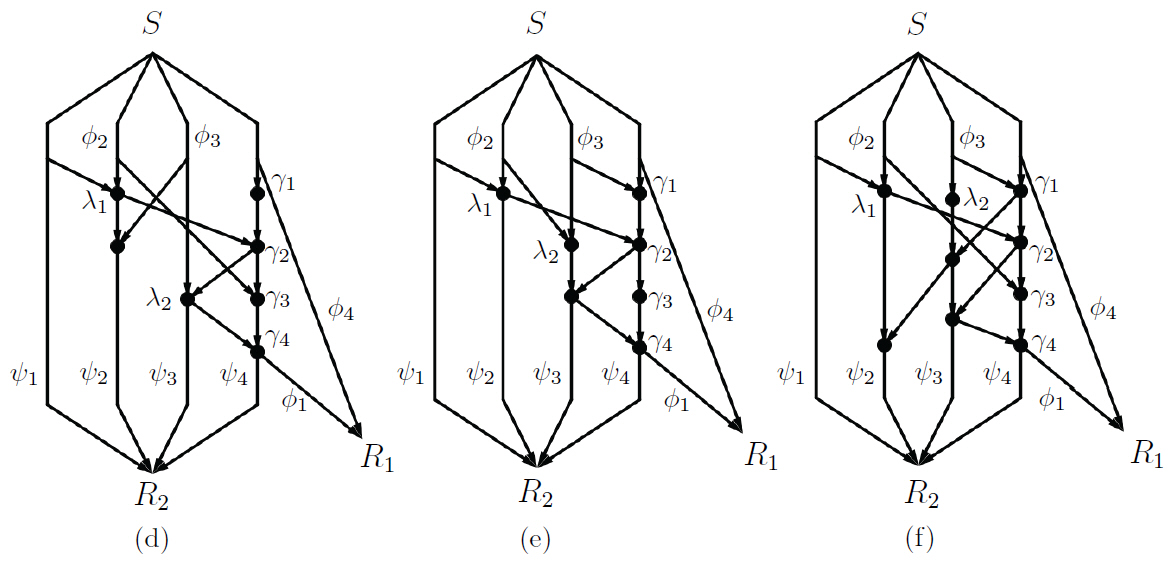}
\includegraphics[width=0.58\textwidth]{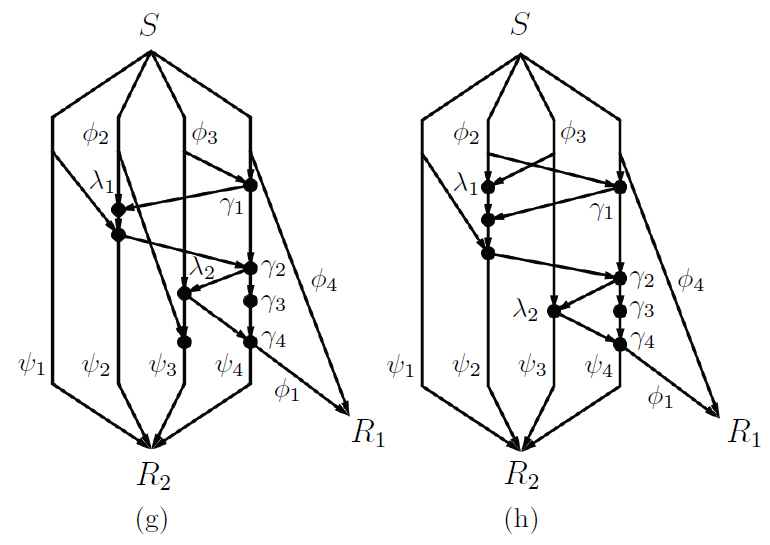}
\caption{(a) Case 1.1 (b) Case 1.2 (c) Case 2.1 (d) Case 2.2.1 (e)(f) Case 2.2.2 (g)(h) Case 2.2.3}
\label{pic44}
\end{figure}
Next, we examine all the following cases to show that $|G|_{\mathcal{M}}$ cannot be $10$.

\underline{Case 1:} paths $\phi_1$ and $\psi_4$ merge at $\gamma_1$ and $\gamma_4$. For this case, we have the following two subcases.

\underline{Case 1.1:} $\phi_1$ first merges with $\psi_4$. For this case, it
is easy to check that $\pi_1$ is of length $3$, which contradicts the fact that it is of length $9$ (see Figure~\ref{pic44}(a)).

\underline{Case 1.2:} $\phi_1$ first merges with $\psi_2$ or $\psi_3$; without loss of generality, assume that $\phi_1$ first merges with $\psi_2$ at the merging denoted by $\mu$. Then necessarily, $\phi_1$ immediately merges with $\psi_4$ at the merging $\gamma_1$. Then, we have $h(\mu),t(\mu)\in \pi_1$, which violates Lemma~\ref{each-path-pair-at-most-once} (see Figure~\ref{pic44}(b)).

\underline{Case 2:} paths $\phi_1$ and $\psi_4$ merge at $\gamma_2$ and $\gamma_4$.. For this case, we consider the following subcases.

\underline{Case 2.1:} $\phi_1$ first merges with $\psi_4$. Then, we have $h(\gamma_1),t(\gamma_1) \in \pi_1$, which violates Lemma~\ref{each-path-pair-at-most-once} (see  Figure~\ref{pic44}(c)).

\underline{Case 2.2:} $\phi_1$ first merges with $\psi_2$ or $\psi_3$; without loss of generality, assume that $\phi_1$ first merges with $\psi_2$. Then necessarily, $\phi_1$ will subsequently merges with $\psi_4$, $\psi_3$ and $\psi_4$. Let $\lambda_1,\lambda_2$ be the smallest mergings in $\psi_2, \psi_3$, respectively. It is clear that at least one of $\lambda_1$ and $\lambda_2$ belongs to $\phi_1$, since otherwise $\lambda_1, \lambda_2$ would belong to $\phi_3, \phi_2$, respectively, and thus $\lambda_1$ would semi-reach itself from head to head again $\psi$, which implies the existence of a rerouting, a contradiction.

\underline{Case 2.2.1:} both the first mergings on $\psi_2$, $\psi_3$ belong to $\phi_1$. Then, $\gamma_1$ is the largest merging on either $\phi_2$ or $\phi_3$, that is, from $\gamma_1$, the associated path cannot go forward to merge anymore. $\phi_3$ can only first merges with $\psi_2$ and $\phi_2$ can only first merges with $\psi_4$ at $\gamma_3$, which implies the existence of a rerouting ($\gamma_3$ semi-reaches itself against $\phi$ from head to head). See Figure~\ref{pic44}(d) for an example.

\underline{Case 2.2.2:} the first merging $\lambda_1$ on $\psi_2$ belongs to $\phi_1$, and the first merging on $\psi_3$ belongs to $\phi_2$. If $\phi_2$ first merges with $\psi_3$, then $\phi_3$ can only first merges $\psi_4$ at $\gamma_1$, one check that $\pi_1$ is of length $8$, a contradiction (see Figure~\ref{pic44}(e)); if $\phi_2$ first merges with $\psi_4$ at $\gamma_3$ (if $\phi_2$ first merges with $\psi_4$ at $\gamma_1$, then $\pi_1$ is of length $6$, a contradiction), then $\phi_3$ can only first merge with $\psi_4$ at $\gamma_1$, and then merges with $\psi_3$, $\psi_2$, which implies the existence of a rerouting ($\gamma_3$ semi-reaches itself against $\phi$ from head to head). See Figure~\ref{pic44}(f) for an example.

\underline{Case 2.2.3:} the first merging $\lambda_2$ on $\psi_3$ belongs to $\phi_1$, and the first merging $\lambda_1$ on $\psi_2$ belongs to $\phi_3$. If $\phi_3$ first merges with $\psi_4$, then necessarily the merging is $\gamma_1$, and $\phi_3$ further merges with $\psi_2$ at $\lambda_1$. In this case $\phi_2$ cannot go forward to merge anymore, which contradicts the fact that
$\phi_2$ merges with $\psi$-paths just three times (see Figure~\ref{pic44}(g)); if $\phi_3$ first merges $\psi_2$ at $\lambda_1$, then $\phi_2$ can only first merges with $\psi_4$ at $\gamma_1$, and then merge with
$\psi_2$. In this case, $\phi_2$ cannot go forward to merge anymore, which also contradicts the fact that $\phi_2$ merges with $\psi$-paths exactly three times (see Figure~\ref{pic44}(h)).

All the above cases combined imply that $|G|_{\mathcal{M}}$ is at most $9$. On the other hand, one can find a non-reroutable $(4, 4)$-graph with one source, two sinks and $9$ mergings as in Figure~\ref{picmstar44}, which implies $|G|_{\mathcal{M}} \ge 9$ (see a more general result in Theorem~\ref{lowerboundmstar}). We then have established the theorem.
\end{proof}

\begin{thm} \label{m33}
$$
\mathcal{M}(3, 3)=13.
$$
\end{thm}

\begin{proof}
Consider a non-reroutable $(3, 3)$-graph $G$ with two source $S_1, S_2$ and two sinks $R_1, R_2$. Let $\phi=\{{\phi}_1,{\phi}_2,{\phi}_3\}$, $\psi=\{{\psi}_1,{\psi}_2,{\psi}_3\}$ denote the set of Menger's paths from $S_1, S_2$ to $R_1, R_2$, respectively.

As discussed in Section~\ref{AA-sequences}, we assume each AA-sequences is of  positive length. Then, by Lemma~\ref{lengthofAAsequence}, the shortest AA-sequence is of length $1$. It can also be checked that the longest AA-sequence in $G$ is of length at most $7$. So, by Lemma~\ref{AA-to-number-of-merings}, we have
$$
\mathcal{M}(3, 3) \leq (7+7+1+7+7+1)/2=15.
$$
It follows from Theorem~\ref{Mmnlb} (this is proven later in Section~\ref{bounds-section}) that $\mcm(3,3) \ge 13$. We next show $\mathcal{M}(3, 3)$ cannot be $15$ or $14$. Note that any non-reroutable $(3, 3)$-graph having 15 mergings implies that \begin{equation} \label{15-mergings}
\mbox{ its } (\phi\mbox{-AA-sequences}; \psi\mbox{-AA-sequences}) \mbox{ are of length } (7, 7, 1; 7, 7, 1), \mbox{ respectively};
\end{equation}
and $14$ mergings implies that
$$
\mbox{ its } (\phi\mbox{-AA-sequences}; \psi\mbox{-AA-sequences}) \mbox{ are of length } (7, 6, 1; 7, 6, 1), (7, 7, 1; 7, 5, 1),
$$
\begin{equation} \label{14-mergings}
(7, 7, 1; 6, 6, 1), (7, 5, 1; 7, 7, 1), (6, 6, 1; 7, 7, 1), \mbox{ respectively}.
\end{equation}
The idea of the proof is that we first preprocess to eliminate many cases by checking if (\ref{15-mergings}) and (\ref{14-mergings}) are satisfied, then we can exhaustively investigate all the remaining cases to prove $\mathcal{M}(3, 3)$ cannot be equal to $14$ or $15$.

\begin{figure}
\psfrag{axx}{$\footnotesize\textrm{(a)}$}\psfrag{bxx}{$\footnotesize\textrm{(b)}$}\psfrag{cxx}{$\footnotesize\textrm{(c)}$}
\psfrag{s1a}{$S_1$} \psfrag{r1a}{$R_1$} \psfrag{s2a}{$S_2$} \psfrag{r2a}{$R_2$}
\psfrag{aaa}{$\phi_1$} \psfrag{bbb}{$\phi_2$} \psfrag{ccc}{$\phi_3$}
\psfrag{phi1}{$\phi_1$} \psfrag{phi2}{$\phi_2$} \psfrag{phi3}{$\phi_3$}
\psfrag{gg0}{$\gamma_0$} \psfrag{psi1}{$\psi_1$} \psfrag{psi2}{$\psi_2$} \psfrag{psi3}{$\psi_3$}
\psfrag{gg1}{$\gamma_1$} \psfrag{gg2}{$\gamma_2$} \psfrag{gg3}{$\gamma_3$} \psfrag{gg4}{$\gamma_4$} \psfrag{gg5}{$\gamma_5$} \psfrag{gg6}{$\gamma_6$} \psfrag{gg7}{$\gamma_7$} \psfrag{gg8}{$\gamma_8$} \psfrag{gg9}{$\gamma_9$} \psfrag{gg10}{$\gamma_{10}$} \psfrag{gg11}{$\gamma_{11}$} \psfrag{gg12}{$\gamma_{12}$}
\centerline{\includegraphics[width=0.83\textwidth]{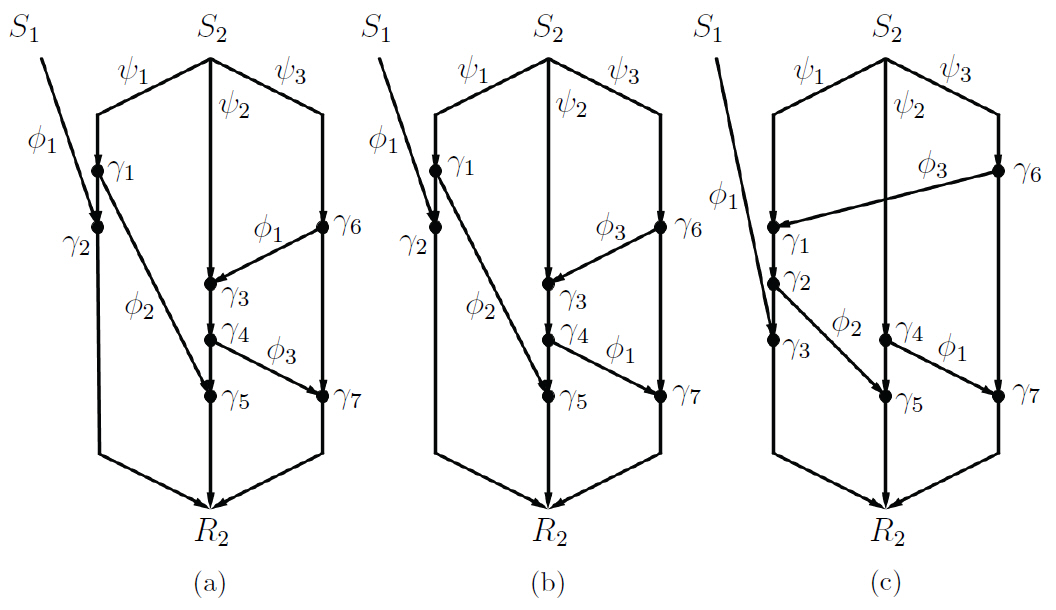}}
\caption{Three possible cases for the $\phi_1$-AA-sequence}
\label{M33IIIIII}
\end{figure}
Suppose, by contradiction, that $G$ has $14$ or $15$ mergings. Then, as before, at least one of AA-sequences of $G$ is of length $7$. Without loss of generality, we assume the $\phi_1$-AA-sequences is of length $7$. One then checks that, up to obvious symmetry, as depicted in Figure~\ref{M33IIIIII}, we only have three possible cases for the $\phi_1$-AA-sequence: for Case $1$, the $\phi_1$-AA-sequence is $S_1 \Rightarrow h(\gamma_2) \Leftarrow t(\gamma_1) \Rightarrow h(\gamma_5) \Leftarrow t(\gamma_4) \Rightarrow h(\gamma_7) \Leftarrow t(\gamma_6) \Rightarrow h(\gamma_3) \Leftarrow S_2$; for Case $2$, the $\phi_1$-AA-sequence is $S_1 \Rightarrow h(\gamma_2) \Leftarrow t(\gamma_1) \Rightarrow h(\gamma_5) \Leftarrow t(\gamma_4) \Rightarrow h(\gamma_7) \Leftarrow t(\gamma_6) \Rightarrow h(\gamma_3) \Leftarrow S_2$; for Case $3$, the $\phi_1$-AA-sequence is $S_1 \Rightarrow h(\gamma_3) \Leftarrow t(\gamma_2) \Rightarrow h(\gamma_5) \Leftarrow t(\gamma_4) \Rightarrow h(\gamma_7) \Leftarrow t(\gamma_6) \Rightarrow h(\gamma_1) \Leftarrow S_2$.

Note that the graphs in Figure~\ref{M33IIIIII} only show the segments of paths $\phi_1,\phi_2,\phi_3$ associated with the $\phi_1$-AA-sequence. Next, for each of the above-mentioned cases, we will extend these segments either backward or forward in all possible ways, and we shall show that no matter how we extend, the number of mergings in $G$ will not exceed $13$.
\begin{figure}
\psfrag{axx}{$\footnotesize\textrm{(a)}$}\psfrag{bxx}{$\footnotesize\textrm{(b)}$}\psfrag{cxx}{$\footnotesize\textrm{(c)}$}
\psfrag{s1a}{$S_1$} \psfrag{r1a}{$R_1$} \psfrag{s2a}{$S_2$} \psfrag{r2a}{$R_2$}
\psfrag{aaa}{$\phi_1$} \psfrag{bbb}{$\phi_2$} \psfrag{ccc}{$\phi_3$}
\psfrag{psi1}{$\psi_1$} \psfrag{psi2}{$\psi_2$} \psfrag{psi3}{$\psi_3$}
\psfrag{sg1}{$\sigma_1$}\psfrag{sg2}{$\sigma_2$}\psfrag{sg3}{$\sigma_3$}
\psfrag{gg0}{$\gamma_0$}
\psfrag{gg1}{$\gamma_1$} \psfrag{gg2}{$\gamma_2$} \psfrag{gg3}{$\gamma_3$} \psfrag{gg4}{$\gamma_4$} \psfrag{gg5}{$\gamma_5$} \psfrag{gg6}{$\gamma_6$} \psfrag{gg7}{$\gamma_7$} \psfrag{gg8}{$\gamma_8$} \psfrag{gg9}{$\gamma_9$} \psfrag{gg10}{$\gamma_{10}$} \psfrag{gg11}{$\gamma_{11}$} \psfrag{gg12}{$\gamma_{12}$}
\centerline{\includegraphics[width=0.8\textwidth]{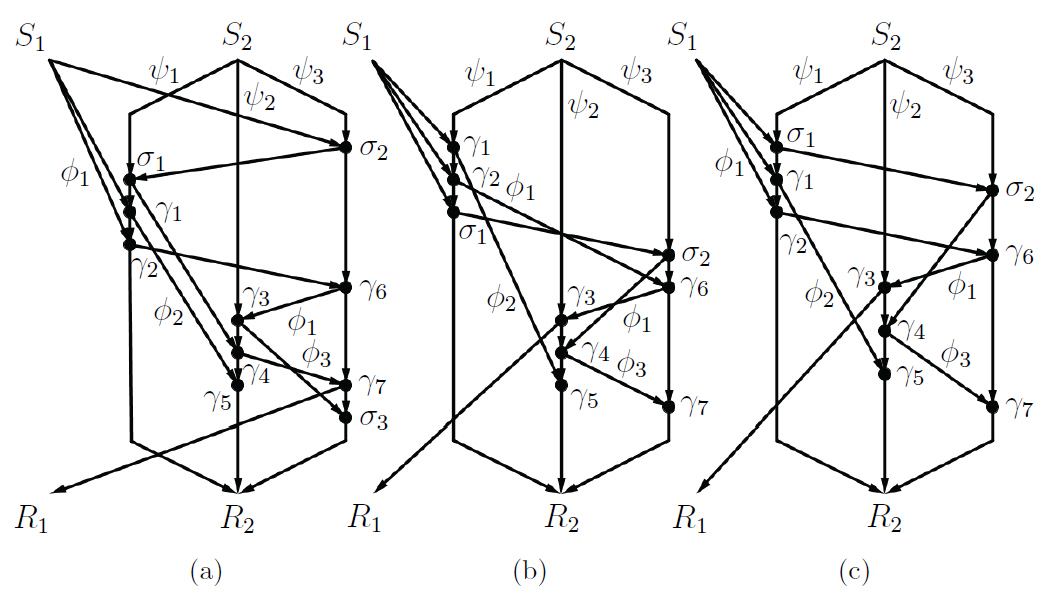}}
\caption{Case $1$, Subcase $(2,0)$}
\label{M33I20abc}
\end{figure}

\underline{Case 1:} as shown in Figure~\ref{M33IIIIII}(a).

For this case, one checks that after $\phi_1$ first merge with $\psi_1$ at $\gamma_2$, it must immediately merge with $\psi_3$ at $\gamma_6$; one also checks that for paths $\phi_2, \phi_3$, each of them can only go backward to merge at most twice. In the following, by Subcase ($l_1$, $l_2$), we mean the case when path $\phi_3$ goes backward to merge $l_1$ times and path $\phi_2$ goes backward to merge $l_2$ times. It suffices to check the following nine subcases: $(0, 0)$, $(0, 1)$, $(0, 2)$, $(1, 0)$, $(1, 1)$, $(1, 2)$, $(2, 0)$, $(2, 1)$, $(2, 2)$.

The checking procedure is rather mechanical and tedious, so we only go through Subcase $(2, 0)$, as shown in Figure~\ref{M33I20abc}, for illustrative purposes. For this case, we have three choices for path $\phi_3$.

For Choice $1$ as shown in Figure~\ref{M33I20abc}(a), the $\phi_3$-AA-sequence is of length $1$, so path $\phi_1$ must go forward to merge further to make sure the $\phi_2$-AA-sequence is of length more than $6$. Therefore, from $\gamma_7$, path $\phi_3$ cannot go forward to merge any more and it must go to $R_2$ directly. Then one exhaustively checks that from $\gamma_9$ and $\gamma_5$, paths $\phi_1$ and $\phi_2$ cannot go forward to merge more than four times in total.

For Choice $2$ as shown in Figure~\ref{M33I20abc}(b), the $\phi_2$-AA-sequence is of length $1$, and path $\phi_1$ cannot go forward to merge anymore. One exhaustively checks that from $\gamma_5$ and $\gamma_7$, paths $\phi_2$ and $\phi_3$ cannot go forward to merge five times in total.

For Choice $3$ as shown in Figure~\ref{M33I20abc}(c), the $\phi_3$-AA-sequence is of length $1$, and the $\phi_2$-AA-sequence is of length $3$. So, (\ref{15-mergings}) or (\ref{14-mergings}) is not satisfied.

\underline{Case 2:} as shown in Figure~\ref{M33IIIIII}(b).

For this case, one checks that each of paths $\phi_2$ and $\phi_3$ cannot go backward to merge more than three times. One also checks that path $\phi_1$, after merging with $\psi_1$ at $\gamma_2$, will immediately merge with $\psi_2$ at $\gamma_4$. Since otherwise, one verifies that the total number of mergings is strictly less than $14$: path $\phi_3$ can go backward to merge for at most twice and path $\phi_2$ cannot go backward to merge; furthermore, path $\phi_3$ cannot go forward to merge anymore from $\gamma_3$ and paths $\phi_1$ and $\phi_2$ cannot go forward to merge four times in total. It suffices to check the following subcases: $(0, 0)$, $(1, 0)$, $(1, 1)$, $(2, 0)$, $(2, 1)$, $(2, 2)$, $(1, 3)$, $(0, 1)$, $(0, 2)$, $(1, 2)$.

\underline{Case 3:} as shown in Figure~\ref{M33IIIIII}(c).

For this case, after path $\phi_1$ merges with $\psi_1$ at $\gamma_3$, it has to immediately merge with $\psi_2$ at $\gamma_4$. Similarly as before, it suffices to check the following subcases: $(0, 0)$, $(1, 0)$, $(1, 1)$, $(2, 1)$, $(2, 2)$, $(0, 1)$, $(0, 2)$, $(1, 2)$.

\end{proof}

\begin{thm}
\begin{displaymath}
\mcm(1,2,n)=
\left\{
\begin{array}{ll}
4n&\textrm{ if }n=2,3,\\
4n+1&\textrm{ if }n=1 \textrm{ or }n\ge 4.
\end{array}
\right.
\end{displaymath}
\end{thm}

\begin{proof} It follows from~\cite{ha2011} that
\begin{equation}  \label{4n+1}
\mcm(1,2,n)\le \mcm(1,2)+\mcm(1,n)+\mcm(2,n)=2+n+(3n-1)=4n+1.
\end{equation}

To prove the theorem, we will consider the following four cases:

\underline{Case $1$:} $n=1$. It immediately follows from Theorem~\ref{one-one-two} that
$$
\mathcal{M}(1, 2, 1)=\mathcal{M}(1, 1, 2)=5.
$$

\underline{Case $2$:} $n=2$. It can be checked that the $(1, 2, 2)$-graph in Figure~\ref{picM122twopic}(a) is non-reroutable, which implies that $\mcm(1,2,2) \geq 8$. Since, by (\ref{4n+1}), $\mathcal{M}(1, 2, 2) \leq 9$, it suffices to prove that $\mathcal{M}(1, 2, 2)$ is not $9$.
\begin{figure}
\psfrag{axx}{$\footnotesize\textrm{(a)}$}\psfrag{bxx}{$\footnotesize\textrm{(b)}$}\psfrag{cxx}{$\footnotesize\textrm{(c)}$}
\psfrag{s1a}{$S_1$}\psfrag{r1a}{$R_1$}
\psfrag{s2a}{$S_2$}\psfrag{r2p}{$R_2$}
\psfrag{s3v}{$S_3$}\psfrag{r3v}{$R_3$}
\psfrag{s2a}{$S_2$}\psfrag{r2a}{$R_2$}
\psfrag{s3a}{$S_3$}\psfrag{r3a}{$R_3$}
\psfrag{aa1}{$\phi_1$}\psfrag{aa2}{$\phi_2$}\psfrag{aa3}{$\phi_3$}\psfrag{aa4}{$\phi_4$}
\psfrag{bb1}{$\psi_1$}\psfrag{bb2}{$\psi_2$}\psfrag{bb3}{$\psi_3$}
\psfrag{cc1}{$\xi_1$}\psfrag{cc2}{$\xi_2$}\psfrag{cc3}{$\xi_3$}\psfrag{cc4}{$\xi_4$}
\psfrag{dd1}{$\eta_1$}\psfrag{dd2}{$\eta_2$}\psfrag{dd3}{$\eta_3$}
\centering
\includegraphics[width=0.72\textwidth]{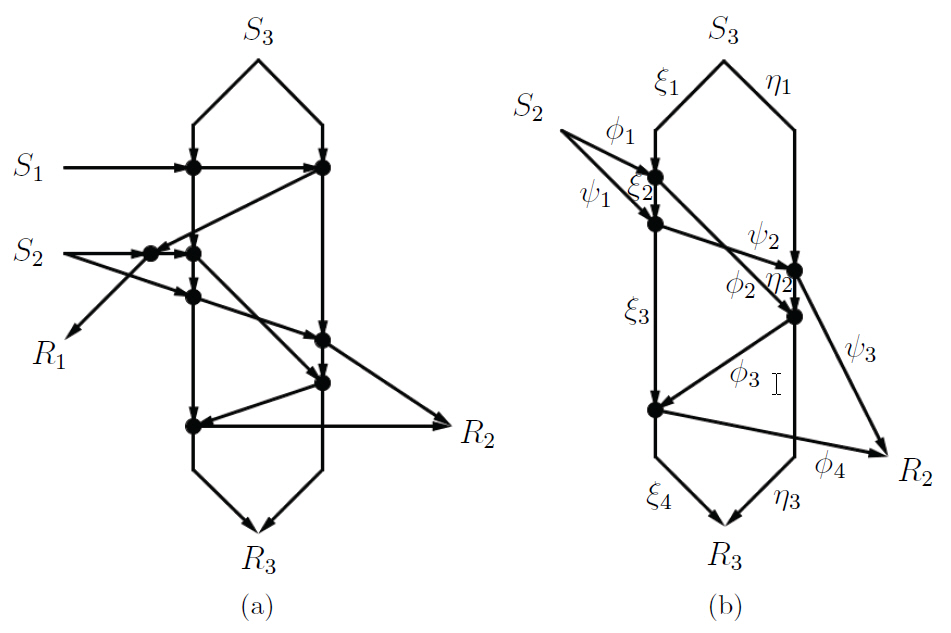}
\caption{(a) A non-reroutable $(1,2,2)$-graph with 8 mergings (b) The edge-labeled non-reroutable $(2,2)$-graph}\label{picM122twopic}
\end{figure}

Suppose, by contradiction, that a non-reroutable $(1, 2, 2)$-graph $G$ has $9$ mergings. Assume $G$ has distinct sources $S_1, S_2, S_3$, sinks $R_1, R_2, R_3$, path $\beta$ from $S_1$ to $R_1$, a set of two Menger's paths $\{\phi,\psi\}$ from $S_2$ to $R_2$ and a set of two Menger's paths $\{\xi,\eta\}$ from $S_3$ to $R_3$.

Since $G$ is non-reroutable, path $\beta$ merges with each of paths $\phi,\psi,\xi,\eta$ at most once (otherwise path $\beta$ is reroutable through the path with which $\beta$ merges twice). This, together with the assumption that $|G|_{\mathcal{M}}=9$ and the fact that $\mcm(2,2)=5$, implies that $|G'|_{\mathcal{M}}=5$, where $G'$ denotes the subgraph of $G$ induced on $\phi, \psi, \xi, \eta$, and $\beta$ must merge with each of $\phi,\psi,\xi,\eta$ exactly once. Here, by Remark~\ref{Pell}, $G'$ has only two ``reduced'' instances: the graph in Figure~\ref{picM122twopic}(b) and its ``reversed'' version obtained by reversing all its edges; so, we can assume $G'$ takes the form as in Figure~\ref{picM122twopic}(b). Moreover, since we count mergings without multiplicity, we can further assume that every merging in $G$ is by exactly two Menger's paths.

Now, we exhaustively examine all ways in which $\beta$ can merge with $G'$ without generating any reroutings or cycles. The following rule can be used to eliminate many cases: For any two paths $\beta'$, $\beta''$ in $G'$, if $\beta'$ is smaller than $\beta''$, then $\beta$ cannot merge with them both (since, otherwise, $\beta$ is reroutable through $\beta'$ and $\beta''$).

With the subpaths of $\phi,\psi,\xi,\eta$ labeled as in Figure~\ref{picM122twopic}(b), we obtain the following sets of subpaths, each of which consists of (unordered) subpaths, where $\beta$ can merge with $\phi, \psi, \xi, \eta$ without violating the above-mentioned rule: $\{\phi_1,\psi_1,\xi_1,\eta_1\}$, $\{\phi_2,\psi_1,\xi_2,\eta_1\}$, $\{\phi_2,\psi_2,\xi_3,\eta_1\}$, $\{\phi_2,\psi_3,\xi_3,\eta_2\}$, $\{\phi_3,\psi_3,\xi_3,\eta_3\}$, $\{\phi_4,\psi_3,\xi_4,\eta_3\}$. In the following, we examine each of the above possibilities, and conclude that there is no way one can add path $\beta$ without generating reroutings or cycles, which further implies that $\mathcal{M}(1, 2, 2)=8$.

Below, expression like ``$\eta_1\mapsto \phi_1\mapsto \psi_1 :\{\xi,\eta\}$'' means ``if after $\beta$ merges with $\eta_1$, it further immediately with $\phi_1$, and further immediately with $\psi_1$, then the group of Menger's paths $\{\xi, \eta\}$ are reroutable''.

\begin{enumerate}
\item $\phi_1,\psi_1,\xi_1,\eta_1$.

$\phi_1 \mapsto \xi_1 :\{\phi,\psi\}$, $\xi_1\mapsto \phi_1 :\{\xi,\eta\}$,
$\psi_1 \mapsto \eta_1 :\{\phi,\psi\}$, $\eta_1\mapsto \psi_1 :\{\xi,\eta\}$,\\
$\phi_1 \mapsto \eta_1 :\{\phi,\psi\}$, $\xi_1\mapsto \psi_1 :\{\xi,\eta\}$,
$\psi_1\mapsto \xi_1\mapsto \eta_1 :\{\phi,\psi\}$, $\eta_1\mapsto \phi_1\mapsto \psi_1 :\{\xi,\eta\}$.\\
It is easy to check we cannot find path $\beta$ without some of the above subpaths.

\item $\phi_2,\psi_1,\xi_2,\eta_1$.

For $\xi_2$ and $\phi_2$, $\xi_2\mapsto \phi_2 :\{\phi,\psi\}$, $\phi_2\mapsto \xi_2 :\{\xi,\eta\}$.\\
For $\xi_2$ and $\psi_1$, $\xi_2\mapsto \psi_1 :\{\xi,\eta\}$, $\psi_1\mapsto \xi_2 :\{\phi,\psi\}$.\\
For $\xi_2$ and $\eta_1$, $\xi_2\mapsto \eta_1 :\{\phi,\psi\}$, $\eta_1\mapsto \xi_2 :\{\xi,\eta\}$.\\
Hence, path $\beta$ cannot merge with $\xi_2$, if it merges the other three edges.

\item $\phi_2,\psi_2,\xi_3,\eta_1$.

For $\eta_1$ and $\phi_2$, $\eta_1\mapsto \phi_2 :\{\xi,\eta\}$, $\phi_2\mapsto \eta_1 :\{\phi,\psi\}$.\\
For $\eta_1$ and $\psi_2$, $\eta_1\mapsto \psi_2 :\{\xi,\eta\}$, $\psi_2\mapsto \eta_1 :\{\phi,\psi\}$.\\
For $\eta_1$ and $\xi_3$, $\eta_1\mapsto \xi_3 :\{\xi,\eta\}$, $\xi_3\mapsto \eta_1 :\{\phi,\psi\}$.\\
Hence, path $\beta$ cannot merge with $\eta_1$, if it merges the other three edges.

\item $\phi_2,\psi_3,\xi_3,\eta_2$.

For $\psi_3$ and $\phi_2$, $\psi_3\mapsto \phi_2 :\{\xi,\eta\}$, $\phi_2\mapsto \psi_3 :\{\phi,\psi\}$.\\
For $\psi_3$ and $\xi_3$, $\psi_3\mapsto \xi_3 :\{\xi,\eta\}$, $\xi_3\mapsto \psi_3 :\{\phi,\psi\}$.\\
For $\psi_3$ and $\eta_2$, $\psi_3\mapsto \eta_2 :\{\xi,\eta\}$, $\eta_2\mapsto \psi_3 :\{\phi,\psi\}$.\\
Hence, path $\beta$ cannot merge with $\psi_3$, if it merges the other three edges.

\item $\phi_3,\psi_3,\xi_3,\eta_3$.

For $\phi_3$ and $\psi_3$, $\phi_3\mapsto \psi_3 :\{\phi,\psi\}$, $\psi_3\mapsto \phi_3 :\{\xi,\eta\}$.\\
For $\phi_3$ and $\xi_3$, $\phi_3\mapsto \xi_3 :\{\phi,\psi\}$, $\xi_3\mapsto \phi_3 :\{\xi,\eta\}$.\\
For $\phi_3$ and $\eta_3$, $\phi_3\mapsto \eta_3 :\{\xi,\eta\}$, $\eta_3\mapsto \phi_3 :\{\phi,\psi\}$.\\
Hence, path $\beta$ cannot merge with $\phi_3$, if it merges the other three edges.

\item $\phi_4,\psi_3,\xi_4,\eta_3$.

$\phi_4 \mapsto \xi_4 :\{\xi,\eta\}$, $\xi_4 \mapsto \phi_4 :\{\phi,\psi\}$,
$\psi_3 \mapsto \eta_3 :\{\xi,\eta\}$, $\eta_3\mapsto \psi_3 :\{\phi,\psi\}$,\\
$\psi_3 \mapsto \xi_4 :\{\xi,\eta\}$, $\eta_3\mapsto \phi_4 :\{\phi,\psi\}$,
$\psi_3 \mapsto \phi_4 \mapsto \eta_3 :\{\xi,\eta\}$, $\eta_3\mapsto \xi_4\mapsto \psi_3 :\{\phi,\psi\}$.\\
It is easy to check we cannot find path $\beta$ without some of the above subpaths.

\end{enumerate}

\underline{Case $3:$} $n=3$. It can be checked that the $(1, 2, 3)$-graph as in Figure~\ref{picM123and124}(a) is non-reroutable, which implies that $\mcm(1, 2, 3) \geq 12$. Since, by (\ref{4n+1}), $\mathcal{M}(1, 2, 3) \leq 13$, it suffices to prove that $\mathcal{M}(1, 2, 2)$ is not $13$.

Suppose, by contradiction, that a non-reroutable $(1, 2, 3)$-graph $G$ has $13$ mergings. Assume $G$ has distinct sources $S_1, S_2, S_3$, sinks $R_1, R_2, R_3$, path $\beta$ from $S_1$ to $R_1$, a set of two Menger's paths $\{\phi,\psi\}$ from $S_2$ to $R_2$ and a set of three Menger's paths $\{\xi,\eta,\delta\}$ from $S_3$ to $R_3$.

Since $G$ is non-reroutable, path $\beta$ merges each of paths $\phi,\psi,\xi,\eta,\delta$ at most once (otherwise path $\beta$ is reroutable through the path with which $p$ merges twice). This, together with the fact that $|G|_{\mathcal{M}}=13$ and the fact that $\mcm(2,3)=8$, implies that $\beta$ must merge with each of $\phi,\psi,\xi,\eta,\delta$ exactly once and the number of mergings among $\{\phi,\psi\}$ and $\{\xi,\eta,\delta\}$ is $8$.

Similar to the proof for the case $n=2$, we consider the subgraph $G'$ of $G$ induced on paths $\phi,\psi,\xi,\eta,\delta$. One then checks that any $(2, 3)$-graph must have, up to relabeling, one of five merging sequences. We then exhaustively investigate how $\beta$ can be ``added'' to $G'$ to form $G$ without generating any reroutings or cycles. Through a similar discussion, we conclude that there is no way we can add such path $\beta$ to generate a non-reroutable $(1,2,3)$-graph with 13 mergings. As a result, $\mcm(1,2,3)=12$.

\begin{figure}
\psfrag{s1a}{$S_1$}\psfrag{s2a}{$S_2$}\psfrag{s3a}{$S_3$}
\psfrag{r1a}{$R_1$}\psfrag{r2a}{$R_2$}\psfrag{r3a}{$R_3$}
\psfrag{axx}{$\footnotesize\textrm{(a)}$}\psfrag{bxx}{$\footnotesize\textrm{(b)}$}\psfrag{cxx}{$\footnotesize\textrm{(c)}$}
\psfrag{s1a}{\footnotesize$S_1$}\psfrag{r1a}{\footnotesize$R_1$}\psfrag{s2a}{\footnotesize$S_2$}\psfrag{r2a}{\footnotesize$R_2$}\psfrag{s3a}{\footnotesize$S_3$}\psfrag{r3a}{\footnotesize$R_3$}
\psfrag{gm1}{\footnotesize$\gamma_1$}\psfrag{gm2}{\footnotesize$\gamma_2$}\psfrag{gm3}{\footnotesize$\gamma_3$}
\psfrag{gm4}{\footnotesize$\gamma_4$}\psfrag{gm5}{\footnotesize$\gamma_5$}\psfrag{gm6}{\footnotesize$\gamma_6$}
\psfrag{gm7}{\footnotesize$\gamma_7$}\psfrag{gm8}{\footnotesize$\gamma_8$}\psfrag{gm9}{\footnotesize$\gamma_9$}
\psfrag{gmt}{\footnotesize$\gamma_{10}$}\psfrag{gme}{\footnotesize$\gamma_{11}$}
\psfrag{uu1}{\footnotesize$\lambda_1$}\psfrag{uu2}{\footnotesize$\lambda_2$}
\psfrag{xi1}{\footnotesize$\xi_1$}\psfrag{xi2}{\footnotesize$\xi_2$}
\psfrag{vv1}{\footnotesize$\mu_1$}\psfrag{vv2}{\footnotesize$\mu_2$}\psfrag{vv3}{\footnotesize$\mu_3$}\psfrag{vv4}{\footnotesize$\mu_4$}
\psfrag{eta1}{\footnotesize$\eta_1$}\psfrag{eta2}{\footnotesize$\eta_2$}\psfrag{eta3}{\footnotesize$\eta_3$}\psfrag{eta4}{\footnotesize$\eta_4$}
\centering
\includegraphics[width=0.75\textwidth]{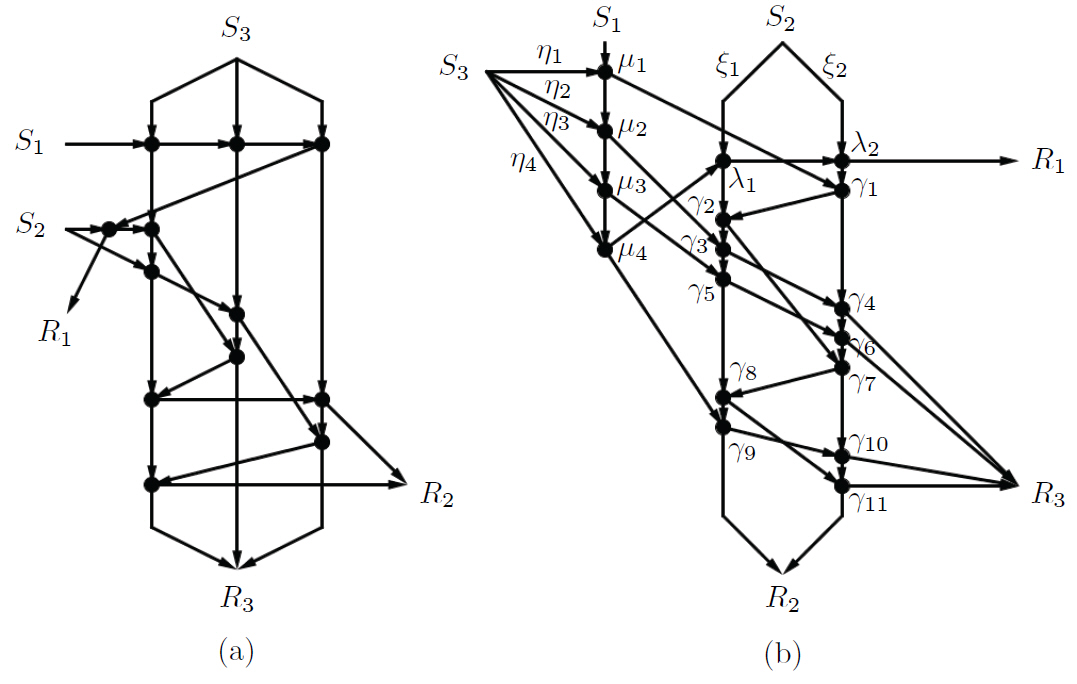}
\caption{(a) A non-reroutable $(1,2,3)$-graph with 12 mergings (b) A non-reroutable $(1,2,4)$-graph with 17 mergings}\label{picM123and124}
\end{figure}
\underline{Case $4$:} $n\ge 4$. By (\ref{4n+1}), we only need to construct a non-reroutable $(1,2,n)$-graph with $4n+1$ mergings. First, we consider a non-reroutable $(2,n)$-graph with distinct sources $S_2, S_3$, sinks $R_2, R_3$ and the following merging sequence:
$$
\Omega_1=(2,1),\Omega_2=(1,1),\Omega_3=(1,2),\Omega_4=(2,2),\Omega_5=(1,3),\Omega_6=(2,3),\Omega_7=(2,1);
$$
for $8\le k\le 3n-1$,
\begin{displaymath}
\Omega_k=\left\{
\begin{array}{lll}
([i]_2,1)      & \mathrm{if}\ k=3i-1 & \mathrm{for}\ 3\leq i\leq n, \\
([i]_2,i+1)    & \mathrm{if}\ k=3i   & \mathrm{for}\ 3\leq i\leq n-1, \\
([i+1]_2,i+1)  & \mathrm{if}\ k=3i+1 & \mathrm{for}\ 3\leq i\leq n-1,
\end{array}
\right.
\end{displaymath}
where $[x]_2=1$ when $x$ is odd, $[x]_2=2$ when $x$ is even. One can check that this $(2,n)$-graph is non-reroutable.

Assume that the two Menger's paths from $S_2$ to $R_2$ start with the subpaths $\xi_1,\xi_2$, respectively; and the $n$ Menger's paths from $S_3$ to $R_3$ start with the subpaths $\eta_1,\eta_2,\ldots,\eta_n$, respectively; and there are no mergings on $\xi_1, \xi_2, \eta_1, \eta_2, \ldots, \eta_n$. Next, we add a path $\beta$ to construct a non-reroutable $(1, 2, n)$-graph such that path $\beta$, starting from $S_1$, successively merges with $\eta_1,\eta_2,\ldots,\eta_n,\xi_1,\xi_2$ (these mergings are labeled as  $\mu_1,\mu_2,\ldots,\mu_n,\lambda_1,\lambda_2$ in Figure~\ref{picM123and124}(b)), and eventually reaches $R_1$.
It can be checked that this newly constructed $(1,2,n)$-graph is non-reroutable.
\end{proof}

\begin{rem}
Through exhaustive searching, we are able to compute exact values for $\mcm$ and $\mstar$ with some small parameters:
$\mathcal{M}(3, 4)=18$, $\mathcal{M}(3, 5)=23$, $\mathcal{M}(3, 6)=28$, $\mathcal{M}(4, 4)=27$, $\mathcal{M}^*(5, 5)=16$, $\mathcal{M}^*(6, 6)=27$, $\mathcal{M}(2, 2, 2)=11$, $\mathcal{M}(1, 3, 3)=17$, $
\mathcal{M}(2, 2, 3)=18$, $\mathcal{M}^*(2, 3, 3)=5$, $\mathcal{M}^*(2, 4, 4)=10$,
$\mathcal{M}^*(2, 5, 5)=17$, $\mathcal{M}^*(3, 3, 3)=8$, $\mathcal{M}^*(3, 4, 4)=13$, $\mathcal{M}^*(4, 4, 4)=18$.
Computations show that for $m \leq n \leq n'$ and $(m, n) \leq (3, 4) \mbox{ or } (2, 5)$,
$$
\mathcal{M}^*(m, n, n')=\mathcal{M}^*(m, n, n).
$$
\end{rem}

\begin{thm} \label{Turan}
$$
\mathcal{M}(\underbrace{1, 1, \ldots, 1}_k)=\left \lfloor \frac{k^2}{4} \right \rfloor.
$$
\end{thm}

\begin{proof}

For the ``$\geq$'' direction, by Proposition $2.12$ of~\cite{ha2011}, we deduce that
$$
\mathcal{M}(\underbrace{1, 1, \ldots, 1}_k) \geq \sum_{i \leq \lfloor k/2 \rfloor, j \geq \lfloor k/2 \rfloor+1} \mathcal{M}(1, 1)= \left \lfloor \frac{k^2}{4} \right \rfloor.
$$
To prove the ``$\leq$'' direction, consider a non-reroutable $(\underbrace{1, 1, \ldots, 1}_k)$-graph $G$ with distinct sources and sets of Menger's paths $\{\beta_1\}, \{\beta_2\}, \ldots, \{\beta_k\}$. It is easy to check that due to non-reroutability of $G$, any two $\beta$-paths can merge with each other at most once. Without loss of generality, assume that $\beta_k$ merges $j$ times with $\beta_1, \beta_2, \ldots, \beta_j$, $1\leq j \leq k-1$; and any other path $\beta_i$, $i \neq k$, merges at most $j$ times. Again, due to non-reroutability of $G$, there are no non-$\beta_k$-involved mergings among paths $\beta_1, \beta_2, \ldots, \beta_j$, where we say a merging at edge $e$ is \emph{$\beta_k$-involved} if $e$ belongs to $\beta_k$. It then follows that any non-$\beta_k$-involved merging in $G$ must be associated with one of paths from $\beta_{j+1}, \ldots, \beta_k$, each of which merges at most $j$ times. We then conclude that
$$
|G|_{\mathcal{M}} \leq j + (k-j-1)j = (k-j) j \le \left \lfloor \frac{k^2}{4} \right \rfloor.
$$

\end{proof}

\begin{rem}  \label{wriggle}
For a non-reroutable $(\underbrace{1, 1, \ldots, 1}_k)$-graph $G$, in order to prove
$$
|G|_\mcm \leq \left \lfloor \frac{k^2}{4} \right \rfloor,
$$
we only need the following two conditions:
\begin{enumerate}
\item any two $\beta_{i_1}, \beta_{i_2}$ can merge at most once;
\item there are at most two mergings in any subgraph of $G$ induced on any three $\beta_{i_1}, \beta_{i_2},\beta_{i_3}$.
\end{enumerate}
So, in some sense, Theorem~\ref{Turan} is a ``dual'' version of the classical Turan's theorem~\cite{tu1941}, which states that the number of edges in a graph is less than $\lfloor \frac{k^2}{4} \rfloor$ if
\begin{enumerate}
\item the graph is simple, i.e., there is at most one edge between any two vertices;
\item the graph does not have triangles, i.e., there are at most two edges among any three vertices.
\end{enumerate}
\end{rem}

\begin{thm}  \label{one-one-two}
\begin{equation*} \mcm(\underbrace{1,\ldots,1}_k,2) =
\begin{cases}
3k-1 & \text{ if } k\leq 6,\\
\lfloor\frac{k^2}{4}\rfloor+k+2 & \text{ if }  k> 6.
\end{cases}
\end{equation*}
\end{thm}

\begin{proof}

\underline{The upper bound direction:} Consider any non-reroutable $(\underbrace{1, 1, \ldots, 1}_k, 2)$-graph $G$ with distinct sources $S_1,S_2,\ldots,S_k,\widehat{S}$, sinks $R_1,R_2,\ldots,R_k,\widehat{R}$, a Menger's path $\beta_i$ from $S_i$ to $R_i$ for $1\le i\le k$, two Menger's paths $\psi_1,\psi_2$ from $\widehat{S}$ to $\widehat{R}$. Let $B_1(B_2)$ denote the set of $\beta$-paths, each of which first merges with $\psi_1(\psi_2)$ and then with $\psi_2(\psi_1)$. Let $A_1 (A_2)$ denote the set of $\beta$-paths, each of which only merges with $\psi_1 (\psi_2)$, and let $C$ denote the set of $\beta$-paths, each of which does not merge with $\psi_1$ or $\psi_2$. And we write
$$
A=A_1\cup A_2, \qquad B=B_1\cup B_2.
$$

Consider any path $\beta_k$ in $B$. Assume that $\beta_k$ merges with $\psi_1$ at merged subpath $\gamma_{k, 1}$ and with $\psi_2$ at merged subpath $\gamma_{k, 2}$. Now, pick any path $\beta_i \in B_1$ ($B_2$). If, for some $j \neq i$, $\gamma_{j, 1}$ overlaps (i.e., shares an edge) with $\gamma_{i, 1}$, then by the non-reroutability of $G$, we have
\begin{enumerate}
\item $\beta_j \in B_2$ ($B_1$), in which case $\gamma_{j, 2}$ does not overlap with $\gamma_{i, 2}$; or
\item $\beta_j \in B_1$ ($B_2$), in which case $\beta_j$ must share
    \begin{itemize}
    \item the edge on $\gamma_{i, 1}$ ($\gamma_{i, 2}$) ending at $t(\gamma_{i, 1})$ ($t(\gamma_{i, 2})$),
    \item the subpath $\beta_i[t(\gamma_{i, 1}), h(\gamma_{i, 2}))]$ ($\beta_i[t(\gamma_{i, 2}), h(\gamma_{i, 1}))]$),
    \item and the edge on $\gamma_{i, 2}$ ($\gamma_{i, 1}$) starting from $h(\gamma_{i, 2})$ ($h(\gamma_{i, 1})$)
    \end{itemize}
    with $\beta_i$. In the remainder of this proof, we say $\beta_j$ is in the same \emph{equivalence class} as $\beta_i$.
\end{enumerate}

In the following, we say a merging at edge $e$ is \emph{$\psi$-involved} if $e$ belongs to either $\psi_1$ or $\psi_2$. The following properties then follow from the non-reroutability of $G$:
\begin{enumerate}[1)]
\item All $B$-paths of the same type (meaning all of them belong to either $B_1$ or $B_2$) and their equivalent classes can be (partially) ordered in the following sense: Consider $\beta_i, \beta_j \in B$ of the same type. Assume that $\beta_i$ merges with $\psi_1, \psi_2$ at $\gamma_{i, 1}, \gamma_{i, 2}$, and $\beta_j$ merges with $\psi_1, \psi_2$ at $\gamma_{j, 1}, \gamma_{j, 2}$. If $\gamma_{i, 1}$ is smaller than $\gamma_{j, 1}$, then $\gamma_{i, 2}$ must be smaller than $\gamma_{j, 2}$; in this case, we say that $\beta_i$ is \emph{smaller} than $\beta_j$, and the equivalence class of $\beta_i$ is \emph{smaller} than that of $\beta_j$. As a result, we can list the equivalence classes of all $B_1$-paths in ascending order: $Q_1, Q_2, \ldots, Q_m$, and the equivalence classes of all $B_2$-paths in ascending order: $\widehat{Q}_1, \widehat{Q}_2, \ldots, \widehat{Q}_n$.

\item A merging by two equivalent $B$-paths or two $B$-paths of different types must be $\psi$-involved. If a merging by any two non-equivalent $B$-paths $\beta_i, \beta_j$ is non-$\psi$-involved, then $\beta_i$ and $\beta_j$ are of the same type. If furthermore $\beta_i$ is smaller than $\beta_j$, then there exists $u$ such that $\beta_i \in Q_u(\widehat{Q}_u)$ and $\beta_j \in Q_{u+1}(\widehat{Q}_{u+1})$. As a consequence, for any $u, v$
$$
|G[Q_u, Q_{u+1}]|_{\mathcal{M}} \leq \min\{|Q_u|, |Q_{u+1}|\}, \quad |G[\widehat{Q}_v,
\widehat{Q}_{v+1}]|_{\mathcal{M}} \leq \min\{|\widehat{Q}_v|, |\widehat{Q}_{v+1}|\},
$$
where $G[Q_u, Q_{u+1}]$ ($G[\widehat{Q}_v, \widehat{Q}_{v+1}]$) denotes the subgraph of $G$ induced on all the $Q_u(\widehat{Q}_v)$-paths and $Q_{u+1}(\widehat{Q}_{v+1})$-paths.

\item Any $A$-path can merge with at most one $B_1$-path and at most one $B_2$-path.

\item Any three $\beta$-paths can only merge with each other at most twice.
\end{enumerate}

Now, by the definition of $A$ and $B$, we have the number of $\psi$-involved mergings is upper bounded by
$$
|A|+2|B|=k+|B|-|C|=2k-|A|-2|C|,
$$
and by Theorem~\ref{Turan}, the number of non-$\psi$-involved mergings is upper bounded by $\lfloor\frac{k^2}{4}\rfloor$.
It then follows that
\begin{equation} \label{universal-upper-bound-1}
M(G) \leq k+|B|-|C|+\left\lfloor\frac{k^2}{4} \right\rfloor.
\end{equation}
Note that for any $1\leq k \leq 3$,
$$
M(G)\le 2k-|A|-2|C|+\floor{\frac{k^2}{4}}\le 2k+\floor{\frac{k^2}{4}}=3k-1.
$$

So, from now on, we only consider the case when $k \geq 4$. It can be easily checked that when $|B|=k$,
$$
M(G)\le 3k-1.
$$

Next, we show that when $|B|< k$,
$$
M(G)\leq \left\lfloor\frac{k^2}{4}\right\rfloor+k+2.
$$

If $|B|-|C|\leq 2$, by (\ref{universal-upper-bound-1}), the above inequality immediately holds.

If $|B|-|C|\geq 3$, we have the following cases to consider:

\underline{Case 1:} there exists some equivalence class that has more than one element. Without loss of generality, assume some $B_1$-class has more than one element, and let $Q_i$ be the smallest such class with $|Q_i|=m > 1$, and let $Q_{i'}$ be the largest such class with $|Q_{i'}|=m' > 1$.

\underline{Case 1.1:} the number of non-$\psi$-involved mergings between $Q_i$ and $Q_{i+1}$ is strictly less than $m$. Then, one checks that
\begin{itemize}
\item $|G[Q_i, \psi]|_{\mathcal{M}} \leq m+1$, where $G[Q_i, \psi]$ denotes the subgraph of $G$ induced on all the $Q_i$-paths and $\psi$-paths;
\item the number of $\psi$-involved mergings by $Q_j$-paths, $j \neq i$, and $\psi$-paths is upper bounded by $2(|B|-m)+|A|$.
\item the number of non-$\psi$-involved mergings by $Q_i$-paths and other $B_1$-classes is upper bounded by $m$.
\item By Theorem~\ref{Turan}, the number of non-$\psi$-involved mergings among $Q_j$-paths, $j \neq i$, is upper bounded by $\left \lfloor \frac{(k-m)^2}{4} \right \rfloor$.
\item The number of non-$\psi$-involved mergings by $Q_i$-paths and ($A$-paths or $C$-paths) is upper bounded by $|A|+|C|$.
\end{itemize}
Combining all the bounds above, we have
$$
M(G) \leq (m+1) + 2 (|B|-m)+ |A|  + m + \left \lfloor \frac{(k-m)^2}{4} \right \rfloor+|A| + |C|
$$
$$
= \left \lfloor \frac{(k-m)^2}{4} \right \rfloor + 2|B| + 2|A|+|C|+1 \leq \left \lfloor \frac{(k-2)^2}{4} \right \rfloor+2k+1 = \left \lfloor \frac{k^2}{4} \right \rfloor + k +2.
$$

\underline{Case 1.2:} the number of non-$\psi$-involved mergings between $Q_i$ and $Q_{i+1}$ is equal to $m$, which necessarily implies that $i' \neq i$. Then, for either $Q_i$ or $Q_{i'}$, the number of non-$\psi$-involved mergings with $A$-paths is at most $|A|-1$, so we have
$$
M(G) \leq (m +1) + 2 (|B|-m)+A + (m+1) + \left \lfloor \frac{(k-m)^2}{4} \right \rfloor + (|A|-1)+|C|
$$
$$
=\left \lfloor \frac{(k-m)^2}{4} \right \rfloor + 2|B| + 2|A|+|C|+1 \leq \left \lfloor \frac{(k-2)^2}{4} \right \rfloor+2k+1 = \left \lfloor \frac{k^2}{4} \right \rfloor + k +2.
$$

\underline{Case 2:} every equivalence class has exactly one element. For this case, since the number of $\psi$-involved mergings is upper bounded by $|A|+2|B|=2k-|A|-2|C|$, it suffices to show that the number of non-$\psi$-involved mergings is upper bounded by
$$
\left\lfloor \frac{(k-2)^2}{4} \right\rfloor+|A|+2|C|+1.
$$

\underline{Case 2.1:} there do not exist non-$\psi$-involved mergings among all equivalence classes. For this case, the total number of mergings by $\{\beta_i, \beta_j\}$, any two chosen $B$-paths of the same type, and ($A$-paths or $C$-paths) is at most $|A|+2|C|$. We then conclude that the number of non-$\psi$-involved mergings is upper bounded by
$$
\left\lfloor\frac{(k-2)^2}{4}\right\rfloor+|A|+2|C|.
$$

\underline{Case 2.2:} there exists a non-$\psi$-involved merging by two adjacent equivalent classes, say $Q_j, Q_{j+1}$, and these two classes merge with each other once, however they do not merge with any other $B_1$-classes. By Property $2)$, both of these two classes are of the same type. Moreover, by Property $3)$, the number of mergings between these two classes and ($A$-paths or $C$-paths) is at most $|A|+2|C|$. Hence, we have the number of non-$\psi$-involved mergings is upper bounded by
$$
\left\lfloor\frac{(k-2)^2}{4}\right\rfloor+|A|+2|C|+1.
$$

\underline{Case 2.3:} there exist at least three adjacent equivalent classes, say $Q_j, Q_{j+1},...,Q_{j+l}$, $l\geq 2$, such that $Q_{j+r}$ merges with $Q_{j+r+1}$, $r=0, 1,...,l-1$, however there are no mergings by $\{Q_j, Q_{j+1},...,Q_{j+l}\}$ and other $B_1$-classes. For this case, it can be checked that at least one of $\{Q_j, Q_{j+1}\}$, $\{Q_{j+l-1}, Q_{j+l}\}$ and $\{Q_{j}, Q_{j+l}\}$ merges with $A$-paths at most $|A|-1$ times. Since each of the above pair of paths merge with $B$-paths at most twice and merges with $C$-paths at most $2|C|$ times, we thus have the number of non-$\psi$-involved mergings is upper bounded by
$$
\left\lfloor\frac{(k-2)^2}{4}\right\rfloor+(|A|-1)+2+2|C|=\left\lfloor\frac{(k-2)^2}{4}\right\rfloor+|A|+2|C|+1.
$$

Now, combining all the cases above, we then have established the upper bound direction:
\begin{displaymath}
\mcm(\underbrace{1,\ldots,1}_k,2) \leq \max\{3k-1,\left\lfloor\frac{k^2}{4}\right\rfloor+k+2\}=\begin{cases}
3k-1 & \text{ if } 4\leq k\leq 6,\\
\lfloor\frac{k^2}{4}\rfloor+k+2 & \text{ if }  k> 6.
\end{cases}
\end{displaymath}

\noindent \underline{The lower bound direction:} First, consider the following $(\underbrace{1, 1, \ldots, 1}_k, 2)$-graph $G$ with distinct sources $S_1,S_2,\ldots,S_k,\widehat{S}$, sinks $R_1,R_2,\ldots,R_k,\widehat{R}$, a Menger's path $\beta_i$ from $S_i$ to $R_i$ for $1\le i\le k$, two Menger's paths $\psi_1,\psi_2$ from $\widehat{S}$ to $\widehat{R}$ such that
\begin{itemize}
\item every merging in $G$ is by exactly two paths;
\item for $i=1, 2, \ldots, k$, $\beta_i$ first merges with $\psi_1$ at $\lambda_{i, 1}$, and then merges with $\psi_2$ at $\lambda_{i, 2}$;
\item for any $i < j$, $\beta_i$ is smaller than $\beta_j$;
\item for any $i=1, 2, \ldots, k-1$, $\beta_i$ merges with $\beta_{i+1}$ at $\mu_{i, i+1}$ such that $\mu_{i, i+1}$ is larger than $\lambda_{i, 2}$ and smaller than $\lambda_{i+1, 1}$;
\item there are no other mergings.
\end{itemize}
See Figure~\ref{picM1k12twocases}(a) for an example. It can be verified that the above $G$ is a non-reroutable $(\underbrace{1, 1, \ldots, 1}_k, 2)$-graph with $3k-1$ mergings, which implies that
\begin{equation}  \label{greater-1}
\mathcal{M}(\underbrace{1, 1, \ldots, 1}_k, 2) \geq 3k-1.
\end{equation}

\begin{figure}
\psfrag{axx}{$\footnotesize\textrm{(a)}$}\psfrag{bxx}{$\footnotesize\textrm{(b)}$}
\psfrag{shat}{\footnotesize $\widehat{S}$}\psfrag{rhat}{\footnotesize $\widehat{R}$}
\psfrag{ss1}{\footnotesize $S_1$}\psfrag{rr1}{\footnotesize $R_1$}
\psfrag{ss2}{\footnotesize $S_2$}\psfrag{rr2}{\footnotesize $R_2$}
\psfrag{ss3}{\footnotesize $S_3$}\psfrag{rr3}{\footnotesize $R_3$}
\psfrag{ss4}{\footnotesize $S_4$}\psfrag{rr4}{\footnotesize $R_4$}
\psfrag{ss5}{\footnotesize $S_5$}\psfrag{rr5}{\footnotesize $R_5$}
\psfrag{d11}{\scriptsize $\lambda_{1\hspace{-0.02cm},\hspace{-0.02cm}1}$}\psfrag{d21}{\scriptsize $\lambda_{2\hspace{-0.02cm},\hspace{-0.02cm}1}$}\psfrag{d31}{\scriptsize $\lambda_{3\hspace{-0.02cm},\hspace{-0.02cm}1}$}\psfrag{d41}{\scriptsize $\lambda_{4\hspace{-0.02cm},\hspace{-0.02cm}1}$}\psfrag{d51}{\scriptsize $\lambda_{5\hspace{-0.02cm},\hspace{-0.02cm}1}$}
\psfrag{d12}{\scriptsize $\lambda_{1\hspace{-0.02cm},\hspace{-0.02cm}2}$}\psfrag{d22}{\scriptsize $\lambda_{2\hspace{-0.02cm},\hspace{-0.02cm}2}$}\psfrag{d32}{\scriptsize $\lambda_{3\hspace{-0.02cm},\hspace{-0.02cm}2}$}\psfrag{d42}{\scriptsize $\lambda_{4\hspace{-0.02cm},\hspace{-0.02cm}2}$}\psfrag{d52}{\scriptsize $\lambda_{5\hspace{-0.02cm},\hspace{-0.02cm}2}$}
\psfrag{m14}{\scriptsize $\mu_{1\hspace{-0.02cm},\hspace{-0.02cm}4}$}\psfrag{m24}{\scriptsize $\mu_{2\hspace{-0.02cm},\hspace{-0.02cm}4}$}\psfrag{m34}{\scriptsize $\mu_{3\hspace{-0.02cm},\hspace{-0.02cm}4}$}
\psfrag{m15}{\scriptsize $\mu_{1\hspace{-0.02cm},\hspace{-0.02cm}5}$}\psfrag{m25}{\scriptsize $\mu_{2\hspace{-0.02cm},\hspace{-0.02cm}5}$}\psfrag{m35}{\scriptsize $\mu_{3\hspace{-0.02cm},\hspace{-0.02cm}5}$}
\psfrag{m12}{\scriptsize $\mu_{1\hspace{-0.02cm},\hspace{-0.02cm}2}$}\psfrag{m23}{\scriptsize $\mu_{2\hspace{-0.02cm},\hspace{-0.02cm}3}$}\psfrag{m45}{\scriptsize $\mu_{4\hspace{-0.02cm},\hspace{-0.02cm}5}$}
\psfrag{be1}{\scriptsize $\beta_1$}\psfrag{be2}{\scriptsize $\beta_2$}\psfrag{be3}{\scriptsize $\beta_3$}\psfrag{be4}{\scriptsize $\beta_4$}\psfrag{be5}{\scriptsize $\beta_5$}
\psfrag{psi1}{\scriptsize $\psi_1$}\psfrag{psi2}{\scriptsize $\psi_2$}
  \centering
  \includegraphics[width=0.7\textwidth]{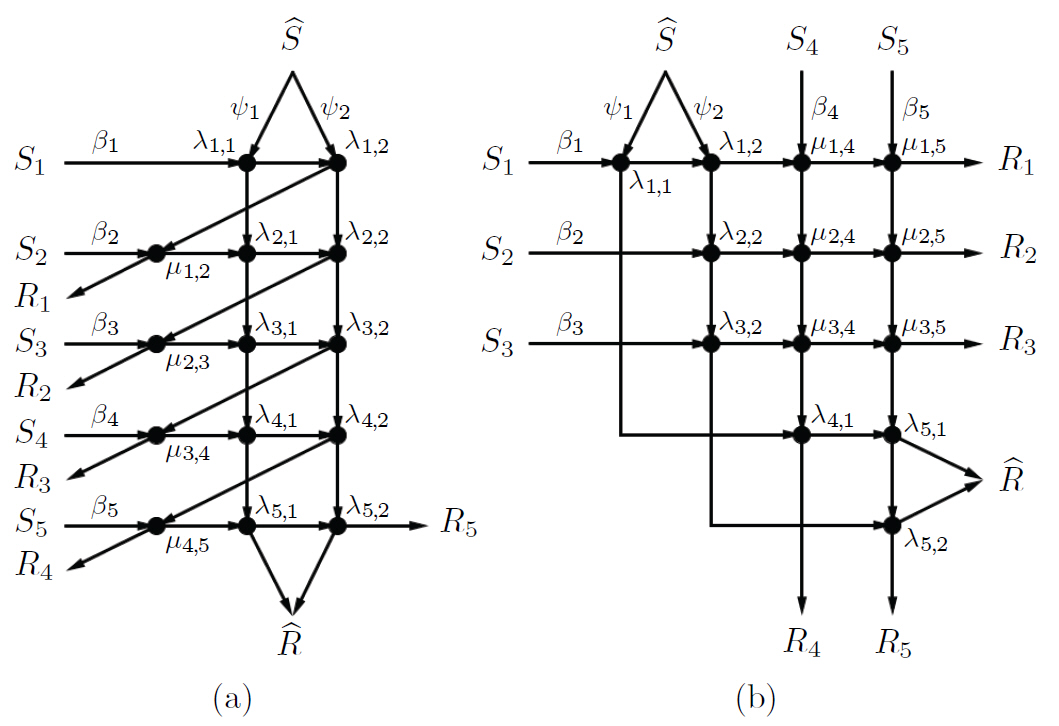}
  \caption{(a) A non-reroutable (1,1,1,1,1,2)-graph with 14 mergings (b) A non-reroutable (1,1,1,1,1,2)-graph with 13 mergings}\label{picM1k12twocases}
\end{figure}

Next, consider the following $(\underbrace{1, 1, \ldots, 1}_k, 2)$-graph $G$ with distinct sources $S_1,S_2,\ldots,S_k,\widehat{S}$, sinks $R_1,R_2,\ldots,R_k,\widehat{R}$, a Menger's path $\beta_i$ from $S_i$ to $R_i$ for $1\le i\le k$, two Menger's paths $\psi_1,\psi_2$ from $\widehat{S}$ to $\widehat{R}$ such that
\begin{itemize}
\item every merging in $G$ is by exactly two paths;
\item for $i=1, 2, \ldots, \lceil k/2 \rceil$, $j=\lceil k/2 \rceil+1, \ldots, k$, $\beta_i$ merges with $\beta_j$ at $\mu_{i, j}$;
\item for any $i =1, 2, \ldots, \lceil k/2 \rceil$ and any $\lceil k/2 \rceil+1 \leq j_1 < j_2 \leq k$, $\mu_{i, j_1}$ is smaller than $\mu_{i, j_2}$;
\item for any $j =\lceil k/2 \rceil+1, \ldots, k$ and any $1 \leq i_1 < i_2 \leq \lceil k/2 \rceil$, $\mu_{i_1, j}$ is smaller than $\mu_{i_2, j}$;
\item for $j=\lceil k/2 \rceil+1, \ldots, k$, $\psi_1$ merges with $\beta_j$ at $\lambda_{j,1}$ such that $\lambda_{j,1}$ is larger than $\mu_{\lceil k/2 \rceil, j}$;
\item for $i=1, 2, \ldots, \lceil k/2 \rceil$, $\psi_2$ merges with $\beta_i$ at $\lambda_{i,2}$ such that $\lambda_{i,2}$ is smaller than $\mu_{i, \lceil k/2 \rceil+1}$;
\item $\psi_1$ merges with $\beta_1$ at $\lambda_{1, 1}$ such that $\lambda_{1, 1}$ is smaller than $\lambda_{1, 2}$ and $\lambda_{\lceil k/2 \rceil+1,1}$;
\item $\psi_2$ merges with $\beta_k$ at $\lambda_{k,2}$ such that $\lambda_{k,2}$ is larger than $\lambda_{\lceil k/2 \rceil,2}$ and $\lambda_{k,1}$;
\item there are no other mergings.
\end{itemize}
See Figure~\ref{picM1k12twocases}(b) for an example. It can be verified that the above $G$ is a non-reroutable $(\underbrace{1, 1, \ldots, 1}_k, 2)$-graph with $\lfloor k^2/4 \rfloor+k+2$ mergings, which implies that
\begin{equation}  \label{greater-2}
\mathcal{M}(\underbrace{1, 1, \ldots, 1}_k, 2) \geq \left\lfloor \frac{k^2}{4} \right\rfloor+k+2.
\end{equation}
Combining (\ref{greater-1}) and (\ref{greater-2}), we then have established the lower bound direction.
\end{proof}

\begin{thm}
$$
\mcm(\underbrace{1,1,\ldots,1}_k,n)=nk+\floor{\frac{k^2}{4}} \ \textrm{ for }n\ge\frac{3k-1}{4}.
$$
\end{thm}

\begin{proof}
For the ``$\le$'' direction, it follows from \cite{ha2011} that
$$
\mcm(\underbrace{1,1,\ldots,1}_k,n)\le \mcm(\underbrace{1,1,\ldots,1}_k)+ k \mcm(1,n)= nk+\left\lfloor\frac{k^2}{4} \right\rfloor.
$$

Next, we show that the following $(\underbrace{1,1,\ldots,1}_k,n)$-graph $G$, which has distinct sources $S_1, S_2, \ldots, S_k, \widehat{S}$ and sinks $R_1, R_2, \ldots, R_k, \widehat{R}$, is non-reroutable with $nk+\lfloor\frac{k^2}{4}\rfloor$ mergings. The graph $G$ (see Figure~\ref{picM11113} for an example) can be described as follows:
\begin{itemize}
\item There is a path $\beta_i$ from $S_i$ to $R_i$ for $1\le i\le k$ and $n$ Menger's paths $\Psi=\{\psi_1,\psi_2,\ldots,\psi_k\}$ from $\widehat{S}$ to $\widehat{R}$;
\item For any feasible $i, j$, $\beta_i$ merges with $\psi_j$ exactly once at the merging $\lambda_{i,j}$;
\item For any feasible $i, j$, $\beta_i$ in $B_1$ or $B_3$ merges with $\beta_j$ in $B_2$ exactly once at the merging $\mu_{i,j}$, where
\begin{align*}
B_1&=\{\beta_1,\beta_2,\ldots,\beta_{k_1}\}, \\
B_2&=\{\beta_{k_1+1},\beta_{k_1+2},\ldots,\beta_{k_1+k_2}\}, \\
B_3&=\{\beta_{k_1+k_2+1},\beta_{k_1+k_2+2},\ldots,\beta_{k}\},
\end{align*}
here, $k_1=\ceil{\ceil{k/2}/2}$, $k_2=\floor{k/2}$, $k_3=\floor{\ceil{k/2}/2}$;
\item The mergings on each path can be sequentially listed in the ascending order as follows:\\
for $1\le i\le k$,
$$\psi_i: \lambda_{1,i},\lambda_{2,i},\ldots,\lambda_{n,i};$$
for $1\le i\le k_1$,
$$\beta_i: \lambda_{i,1}, \lambda_{i,2}, \ldots , \lambda_{i,k}, \mu_{i,k_1+1}, \mu_{i,k_1+2}, \ldots , \mu_{i,k_1+k_2};$$
for $k_1+1\le i\le k_1+k_2$,
$$\beta_i: \mu_{1,i}, \mu_{2,i}, \ldots , \mu_{k_1,i}, \lambda_{i,1}, \lambda_{i,2}, \ldots , \lambda_{i,k}, \mu_{k_1+k_2+1,i}, \mu_{k_1+k_2+2,i}, \ldots , \mu_{k,i};$$
for $k_1+k_2+1\le i\le k$,
$$\beta_i: \mu_{i,k_1+1}, \mu_{i,k_1+2}, \ldots , \mu_{i,k_1+k_2}, \lambda_{i,1}, \lambda_{i,2}, \ldots , \lambda_{i,k}.
$$
\end{itemize}

It can be checked that $G$ is non-reroutable with
\begin{align*}
|G|_{\mathcal{M}}=& n(k_1+k_2+k_3)+(k_1+k_3)k_2\\
=& n\left(\ceil{\frac{\ceil{\frac{k}{2}}}{2}}+\floor{\frac{k}{2}}+\floor{\frac{\ceil{\frac{k}{2}}}{2}}\right)+
\left(\ceil{\frac{\ceil{\frac{k}{2}}}{2}}+\floor{\frac{\ceil{\frac{k}{2}}}{2}}\right)\floor{\frac{k}{2}}\\
=& nk+\floor{\frac{k^2}{4}}.
\end{align*}
\begin{figure}
\psfrag{s1a}{$S_1$}\psfrag{r1a}{$R_1$}\psfrag{s2a}{$S_2$}\psfrag{r2a}{$R_2$}\psfrag{s3a}{$S_3$}\psfrag{r3a}{$R_3$}
\psfrag{s4a}{$S_4$}\psfrag{r4a}{$R_4$}\psfrag{spp}{$\widehat{S}$}\psfrag{rpp}{$\widehat{R}$}
\psfrag{p11}{\scriptsize $\lambda_{1\hspace{-0.02cm},\hspace{-0.02cm}1}$}\psfrag{p21}{\scriptsize $\lambda_{2\hspace{-0.02cm},\hspace{-0.02cm}1}$}\psfrag{p31}{\scriptsize $\lambda_{3\hspace{-0.02cm},\hspace{-0.02cm}1}$}\psfrag{p41}{\scriptsize $\lambda_{4\hspace{-0.02cm},\hspace{-0.02cm}1}$}
\psfrag{p12}{\scriptsize $\lambda_{1\hspace{-0.02cm},\hspace{-0.02cm}2}$}\psfrag{p22}{\scriptsize $\lambda_{2\hspace{-0.02cm},\hspace{-0.02cm}2}$}\psfrag{p32}{\scriptsize $\lambda_{3\hspace{-0.02cm},\hspace{-0.02cm}2}$}\psfrag{p42}{\scriptsize $\lambda_{4\hspace{-0.02cm},\hspace{-0.02cm}2}$}
\psfrag{p13}{\scriptsize $\lambda_{1\hspace{-0.02cm},\hspace{-0.02cm}3}$}\psfrag{p23}{\scriptsize $\lambda_{2\hspace{-0.02cm},\hspace{-0.02cm}3}$}\psfrag{p33}{\scriptsize $\lambda_{3\hspace{-0.02cm},\hspace{-0.02cm}3}$}\psfrag{p43}{\scriptsize $\lambda_{4\hspace{-0.02cm},\hspace{-0.02cm}3}$}
\psfrag{q12}{\scriptsize $\mu_{1\hspace{-0.02cm},\hspace{-0.02cm}2}$}\psfrag{q13}{\scriptsize $\mu_{1\hspace{-0.02cm},\hspace{-0.02cm}3}$}\psfrag{q42}{\scriptsize $\mu_{4\hspace{-0.02cm},\hspace{-0.02cm}2}$}\psfrag{q43}{\scriptsize $\mu_{4\hspace{-0.02cm},\hspace{-0.02cm}3}$}
\psfrag{be1}{\scriptsize $\beta_1$}\psfrag{be2}{\scriptsize $\beta_2$}\psfrag{be3}{\scriptsize $\beta_3$}\psfrag{be4}{\scriptsize $\beta_4$}
\psfrag{psi1}{\scriptsize $\psi_1$}\psfrag{psi2}{\scriptsize $\psi_2$}\psfrag{psi3}{\scriptsize $\psi_3$}
  \centering
  \includegraphics[width=0.54\textwidth]{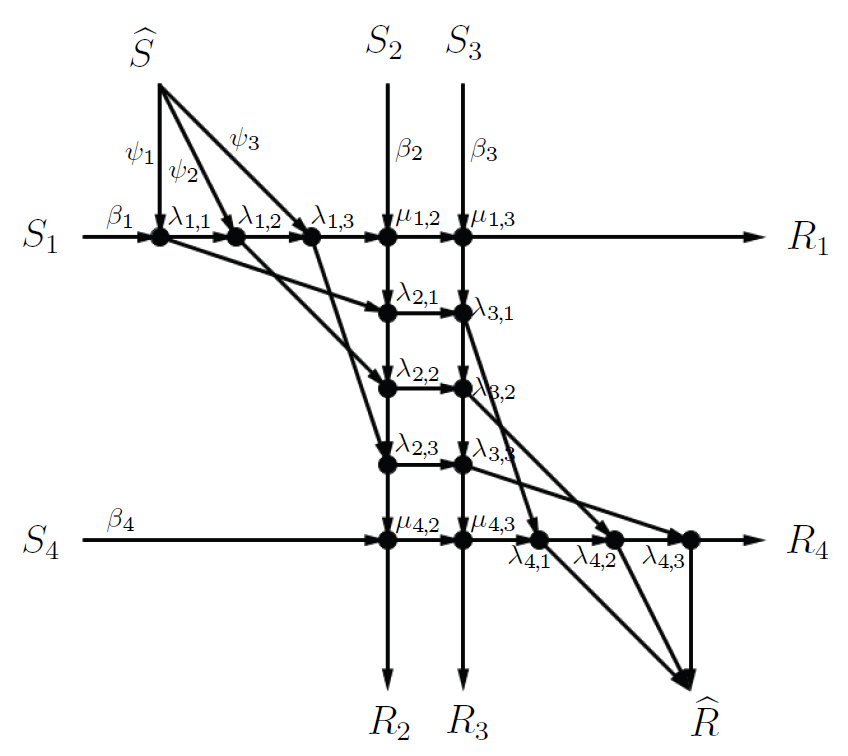}
  \caption{A non-reroutable $(1,1,1,1,3)$-graph with 16 mergings}\label{picM11113}
\end{figure}
\end{proof}

\section{Bounds} \label{bounds-section}

\subsection{Bounds on $\mathcal{M}^*(n,n)$}

In this section, we will construct a non-reroutable $(n,n)$-graph $\mathcal{E}(n,n)$ with one source $S$, two sinks $R_1, R_2$, a set of Menger's paths $\phi=\{\phi_0, \phi_1, \ldots, \phi_{n-1}\}$ from $S$ to $R_1$, a set of Menger's paths $\psi=\{\psi_0, \psi_1, \ldots, \psi_{n-1}\}$ from $S$ to $R_2$ and $(n-1)^2$ mergings for any positive integer $n$, thus giving a lower bound on $\mathcal{M}^*(n,n)$.

The graph $\mathcal{E}(n,n)$ can be described as follows: for each $0 \leq i \leq n-1$, paths $\phi_i$ and $\psi_i$ share a starting subpath $\omega_i$. After $\omega_{n-1}$, path $\phi_{n-1}$ does not merge any more, directly ``flowing'' to $R_1$; after $\omega_0$, path $\psi_0$ does not merge any more, directly ``flowing'' to $R_2$. The rest of the graph can be determined how paths $\phi_0, \phi_1, \ldots, \phi_{n-2}$ merge with $\psi_1, \psi_2, \ldots, \psi_{n-1}$. In more detail, for a given $n$, we define
\begin{align*}
X=&\{x_{i,j}=i(2n-i-2)+j:0 \le i \le n-2,1\le j \le n-i-1\}
\end{align*}
and
\begin{align*}
Y=&\{y_{i,j}=i(2n-i-3)+(n-1)+j: 0\le i\le n-3,1\le j\le n-i-2\}.
\end{align*}
\noindent It can be checked that all $x_{i,j}$'s, $y_{i,j}$'s are distinct and
$$X \cup Y=\{1,2,\ldots,(n-1)^2\}.$$
\noindent Now we define a mapping $f: \{1, 2, \ldots, (n-1)^2\} \mapsto \{(i, j): 0\le i, j\le n-1\}$ by
\begin{displaymath}
f(k)=
\begin{cases}
(i, j) &\textrm{ if }  k =x_{i,j},\\
(n-1-j, n-1-i) &\textrm{ if }  k=y_{i,j}.
\end{cases}
\end{displaymath}
Then the merging sequence of the rest of the graph can be defined as
$$
\Omega=[\Omega_k: \Omega_k=f(k), 1\le k \le (n-1)^2].
$$
For example, $\mathcal{E}(4,4)$, as illustrated in Figure~\ref{picmstar44}, is determined by the merging sequence
$$
\Omega=[(0,1),(0,2),(0,3),(2,3),(1,3),(1,1),(1,2),(2,2),(2,1)].
$$

Now, we prove that

\begin{figure}
\psfrag{sab}{$S$}
\psfrag{r1a}{$R_1$}\psfrag{r2a}{$R_2$}
\psfrag{q0a}{\footnotesize $\omega_0$}\psfrag{q1a}{\footnotesize $\omega_1$}\psfrag{q2a}{\footnotesize $\omega_2$}\psfrag{q3a}{\footnotesize $\omega_3$}
\psfrag{p01}{\footnotesize $\gamma_{0,1}$}\psfrag{p02}{\footnotesize $\gamma_{0,2}$}\psfrag{p03}{\footnotesize $\gamma_{0,3}$}
\psfrag{p11}{\footnotesize $\gamma_{1,1}$}\psfrag{p12}{\footnotesize $\gamma_{1,2}$}\psfrag{p13}{\footnotesize $\gamma_{1,3}$}
\psfrag{p21}{\footnotesize $\gamma_{2,1}$}\psfrag{p22}{\footnotesize $\gamma_{2,2}$}\psfrag{p23}{\footnotesize $\gamma_{2,3}$}
\centering
\includegraphics[width=0.5\textwidth]{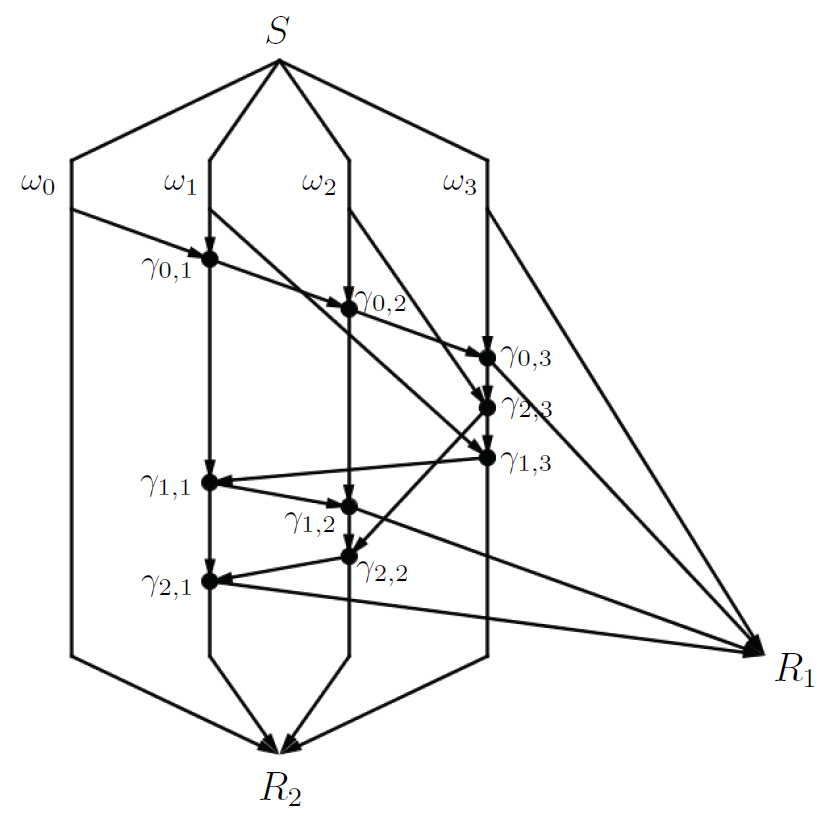}
\caption{Graph $\mathcal{E}(4,4)$ with $9$ mergings}\label{picmstar44}
\end{figure}

\begin{lem} \label{Fnn}
$\mathcal{E}(n,n)$ is non-reroutable.
\end{lem}

\begin{proof}\  Let $z=n-1$. For each $i,j=0, 1, \ldots, z$, label each merging $(i,j)$ in the merging sequence as $\gamma_{i,j}$ (it can be easily checked that no two mergings share the same label).

We only prove that there is only one possible set of Menger's paths from $S$ to $R_1$. The uniqueness of Menger's path sets from $S$ to $R_2$ can be established using a parallel argument.

Let $\alpha_1$ be an arbitrary yet fixed set of Menger's paths from $S$ to $R_1$. It suffices to prove that $\alpha_1$ is non-reroutable. Note that each path in $\alpha_1$ must end with either $\omega_z \rightarrow R_1$ or $\gamma_{i,z-i}\rightarrow R_1$, $i=0, 1, \ldots, z-1$ (here and hereafter, slightly abusing the notations ``$\rightarrow$'' and ``$\leftarrow$'', for paths (or vertices) $A_1, A_2, \ldots, A_k$, we use $A_1 \rightarrow A_2 \rightarrow \cdots \rightarrow A_k$ or $A_k \leftarrow \cdots \leftarrow A_2 \leftarrow A_1$ to denote the path which sequentially passes through $A_1,A_2,\ldots,A_k$; it can be checked that in this proof such an expression uniquely determines a path). In $\alpha_1$, label the Menger's path ending with $\gamma_{i,z-i}\rightarrow R_1$ as the $i$-th Menger's path for $0 \le i \le z-1$, and the Menger's path ending with $\omega_z \rightarrow R_1$ as the $z$-th one.

It is obvious that in $\mathcal{E}(m, m)$, there is only one path ending with $\omega_z \rightarrow R_1$, which implies that the $z$-th Menger's path in $\alpha_1$ is ``fixed'' (as $S \rightarrow \omega_z \rightarrow R_1$); or, more rigorously, for any set of Menger's paths $\alpha'_1$, the $z$-th Menger's path in $\alpha'_1$ is the same as the $z$-th one in $\alpha_1$. So, for the purpose of choosing other Menger's paths, all the edges on $S \rightarrow \omega_z \rightarrow R_1$ are ``occupied''. It then follows that, in $\alpha_1$, $\gamma_{0,z}$ must ``come'' from $\gamma_{0,z-1}$; more precisely, in $\alpha_1$, $\gamma_{0,z-1}$ is smaller than $\gamma_{0,z}$ on the $0$-th path and there is no other merging between them on this path. Now, all the edges on $\gamma_{0,z-1} \rightarrow \gamma_{0,z} \rightarrow R_1$ are occupied.

Inductively, only considering unoccupied edges, one can check that for $0 \le i \le z-2$, $\gamma_{i,z-i}$ must come from $\gamma_{i,z-i-1}$; in other words,  for $0 \le i \le z-2$, the $i$-th Menger's path must end with $\gamma_{i,z-i-1} \rightarrow \gamma_{i,z-i}\rightarrow R_1$. It then follows that the $(z-1)$-th Menger's path must come from $\gamma_{z-1,2} \leftarrow \gamma_{z-1,3} \leftarrow \cdots \leftarrow \gamma_{z-1,z} \leftarrow \omega_{z-1}$; so,  the $(z-1)$-th Menger's path is fixed as $S \rightarrow \omega_{z-1} \rightarrow \gamma_{z-1,z} \rightarrow \gamma_{z-1,z-1} \rightarrow \cdots \rightarrow \gamma_{z-1,2} \rightarrow \gamma_{z-1,1} \rightarrow R_1$.

We now proceed by induction on $j$, $j=z-2, z-3, \ldots, 1$. Suppose that, for $j+1 \le i \le z$, the $i$-th Menger's path is already fixed (and hence the edges on these paths are all occupied), and for $0 \le i \le j$, the $i$-th Menger's path ends with $\gamma_{i,j-i+1} \rightarrow \gamma_{i,j-i+2} \rightarrow \cdots \rightarrow \gamma_{i,z-i}\rightarrow R_1$ (so, the edges on these paths are all occupied). Only considering the unoccupied edges, one checks that for $0 \le i \le j-1$, $\gamma_{i,j-i+1}$ must come from $\gamma_{i,j-i}$. It then follows that the \mbox{$j$-th} Menger's path, which ends with $\gamma_{j,1}\rightarrow \gamma_{j,2}\rightarrow \cdots \rightarrow \gamma_{j,z-j}\rightarrow R_1$, must come from $\gamma_{j,z-j+1} \leftarrow \gamma_{j,z-j+2} \leftarrow \cdots \leftarrow \gamma_{j,z} \leftarrow \omega_j$. So, the $j$-th Menger's path can now be fixed as $S \rightarrow \omega_{j}\rightarrow \gamma_{j,z}\rightarrow \gamma_{j,z-1} \rightarrow \cdots \rightarrow \gamma_{j,z-j+1}\rightarrow \gamma_{j,1}\rightarrow \gamma_{j,2}\rightarrow \cdots \rightarrow \gamma_{j,z-j}\rightarrow R_1$. Now, for $j \le i \le z$, the $i$-th Menger's path is fixed, and for $0 \le i \le j-1$, the $i$-th Menger's path must end with $\gamma_{i,j-i} \rightarrow \gamma_{i,j-i+1} \rightarrow \cdots \rightarrow \gamma_{i,z-i}\rightarrow R_1$.

It follows from the above inductive argument that for $1 \le i \le z$, the $i$-th Menger's path is fixed, and the $0$-th Menger's path must end with $\gamma_{0,1}$ $\rightarrow \gamma_{0,2}$ $\rightarrow \cdots \rightarrow \gamma_{0,z}$ $\rightarrow R_1$. One then checks that the $\gamma_{0, 1}$ must come from $\omega_0$, which implies that the $0$-th Menger's path is fixed as $S \rightarrow \omega_0 \rightarrow \gamma_{0,1} \rightarrow \gamma_{0,2} \rightarrow \cdots \rightarrow \gamma_{0,z}\rightarrow R_1$. The proof of uniqueness of Menger's path set from $S$ to $R_1$ is then complete.
\end{proof}

The above lemma then immediately implies that
\begin{thm}\label{lowerboundmstar}
$$\mathcal{M}^*(n,n)\ge (n-1)^2.$$
\end{thm}

The following theorem gives an upper bound on $\mathcal{M}^*(n, n)$. First, we remind the reader that, by Proposition $3.6$ in~\cite{ha2011}, $\mathcal{M}^*(m,n)=\mathcal{M}^*(n, n)$ for any $m \geq n$.

\begin{thm}
\begin{equation*} \mathcal{M}^*(n, n) \leq \ceil{\dfrac{n}{2}}(n^2-4n+5).
\end{equation*}
\end{thm}

\begin{proof}
Consider any $(n,n)$-graph $G$ with one source $S$, sinks $R_1, R_2$, a set of Menger's paths $\phi=\{\phi_1, \phi_2, \ldots, \phi_n\}$ from $S$ to $R_1$, a set of Menger's paths $\psi=\{\psi_1, \psi_2, \ldots, \psi_n\}$ from $S$ to $R_2$.

As discussed in Section~\ref{AA-sequences}, we assume that, for $1\le i \le n$, paths $\phi_i$ and $\psi_i$ share a starting subpath, and paths $\phi_n$ and $\psi_0$ do not merge with any other paths, directly flowing to the sinks (then, necessarily, each $\psi$-AA-sequence is of positive length, and by Lemma~\ref{lengthofAAsequence}, the shortest $\psi$-AA-sequence is of length $1$). We say that the path pair $(\phi_i,\psi_j)$ is \textit{matched} if $i=j$, otherwise, \textit{unmatched}. Apparently, each starting subpath corresponds to a matched path pair; and among the set of all path pairs, each of which corresponds some merging in $G$, there are at most $(n-2)$ matched and at most $(n^2-3n+3)$ unmatched.

We then consider the following two cases (note that the following two cases may not be mutually exclusive):

\underline{Case 1:} there exists a shortest $\psi$-AA-sequence associated with a matched path pair. By Lemma~\ref{each-path-pair-at-most-once} and the fact that each starting subpath corresponds to a matched path pair,  there are at most $\floor{\frac{n-1}{2}}$ mergings corresponding to this path pair, at most $\floor{\frac{n-2}{2}}$ corresponding to any other matched path pair, and at most $\floor{\frac{n-1}{2}}$ mergings corresponding to any unmatched. So, the number of mergings is upper bounded by
\begin{equation} \label{mstarbound-I}
\floor{\frac{n-1}{2}}+(n-3)\floor{\frac{n-2}{2}}+(n^2-3n+3)\floor{\frac{n-1}{2}}.
\end{equation}

\underline{Case 2:} there exists a shortest $\psi$-AA-sequence associated with an unmatched path pair. Again, by Lemma~\ref{each-path-pair-at-most-once} and the fact that each starting subpath corresponds to a matched path pair, there are at most $\floor{\frac{n}{2}}$ mergings corresponding to this path pair, at most $\floor{\frac{n-1}{2}}$ mergings corresponding to any other unmatched path pair, and at most $\floor{\frac{n-2}{2}}$ mergings corresponding to any matched. So, the number of mergings is upper bounded by
\begin{equation} \label{mstarbound-II}
\floor{\frac{n}{2}}+(n-2)\floor{\frac{n-2}{2}}+(n^2-3n+2)\floor{\frac{n-1}{2}}.
\end{equation}
Then $\mstar(n,n)\le \max\{(\ref{mstarbound-I}),(\ref{mstarbound-II})\}$. For odd $n$, (\ref{mstarbound-I}) is larger than (\ref{mstarbound-II}), so we have
\begin{align*}
\mstar(n,n)&\le\rbrac{\frac{n-1}{2}}+(n-3)\rbrac{\frac{n-3}{2}}+(n^2-3n+3)\rbrac{\frac{n-1}{2}}\\
&=(n^2-4n+5)\rbrac{\frac{n+1}{2}}.
\end{align*}
For even $n$, (\ref{mstarbound-II}) is larger than (\ref{mstarbound-I}), so we have
\begin{align*}
\mstar(n,n)&\le\rbrac{\frac{n}{2}}+(n-2)\rbrac{\frac{n-2}{2}}+(n^2-3n+2)\rbrac{\frac{n-2}{2}}\\
&=(n^2-4n+5)\rbrac{\frac{n}{2}}.
\end{align*}
The proof is then complete.
\end{proof}

\subsection{Bounds on $\mathcal{M}(m,n)$}

Consider the following $(n,n)$-graph $\mathcal{F}(n,n)$ with distinct sources $S_1, S_2$, sinks $R_1, R_2$, a set of Menger's paths $\phi=\{\phi_1, \phi_2, \ldots, \phi_n\}$ from $S_1$ to $R_1$, a set of Menger's paths $\psi=\{\psi_1, \psi_2, \ldots, \psi_n\}$ from $S_2$ to $R_2$, and a merging sequence $\Omega=[\Omega_k: 1 \le k \le 2n^2-3n+2]$, where
\begin{displaymath}
\Omega_k=
\left\{
      \begin{array}{ll}
            ([j-i]_n,i+1) &\mathrm{if}\   k=2i(n-1)+j\quad \\
            &\mathrm{for} \ (0\le i\le n-1, 1\le j\le n-1)\ \mathrm{or}\ (i=n-1,j=n),\\
            (n-i,[i-j+2]_n) &\mathrm{if}\  k=(2i+1)(n-1)+j\quad \mathrm{for} \ 0\le i\le n-2, 1\le j\le n-1,
      \end{array}
\right.
\end{displaymath}
where, for any integer $x$, $[x]_n$ denotes the least strictly positive residue of $x$ modulo $n$.
For a quick example, see $\mathcal{F}(3,3)$ in Figure~\ref{pic33blockwithsplitted}(a), whose merging sequence is
\begin{align*}
\Omega=&[(1,1),(2,1),(3,1),(3,3),(3,2),(1,2),(2,2),(2,1),(2,3),(3,3),(1,3)].
\end{align*}
Then, similar to the proof of Lemma~\ref{Fnn}, through verifying the uniqueness of the set of Menger's paths from $S_i$ to $R_i$, we can show that

\begin{lem}
$\mathcal{F}(n,n)$ is non-reroutable.
\end{lem}

Consider a non-reroutable $(k,n)$-graph $\mathcal{G}(k, n)$ with distinct sources $\widehat{S}_1, \widehat{S}_2$, sinks $\widehat{R}_1, \widehat{R}_2$, a set of Menger's paths $\hat{\phi}=\{\hat{\phi}_1, \hat{\phi}_2, \ldots, \hat{\phi}_k\}$ from $\widehat{S}_1$ to $\widehat{R}_1$, a set of Menger's paths $\hat{\psi}=\{\hat{\psi}_1, \hat{\psi}_2, \ldots, \hat{\psi}_n\}$ from $\widehat{S}_2$ to $\widehat{R}_2$. For a fixed merging sequence of $\mathcal{G}(k, n)$, assume, without loss of generality, that the first element is $(\hat{\phi}_1, \hat{\psi}_n)$. Now, we consider the following procedure of concatenating graphs $\mathcal{F}(n,n)$ and $\mathcal{G}(k,n)$ to obtain a new graph:
\begin{enumerate}
\item split $R_1$ into $n$ copies $R_1^{(1)},R_1^{(2)},\ldots,R_1^{(n)}$ such that path $\phi_i$ has the ending point $R_1^{(i)}$; split $R_2$ into $n$ copies $R_2^{(1)},R_2^{(2)},\ldots,R_2^{(n)}$ such that path $\psi_i$ has the ending point $R_2^{(i)}$;
\item split $\widehat{S}_1$ into $k$ copies $\widehat{S}_1^{(1)},\widehat{S}_1^{(2)},\ldots,\widehat{S}_1^{(k)}$ such that path $\hat{\phi}_i$ has the starting point $\widehat{S}_1^{(i)}$; split $\widehat{S}_2$ into $n$ copies $\widehat{S}_2^{(1)},\widehat{S}_2^{(2)},\ldots,\widehat{S}_2^{(n)}$ such that path $\psi_i$ has the starting point $\widehat{S}_2^{(i)}$;
\item delete all edges on $\phi_1$ and all edges on $\psi_n$, each of which is larger than merging $(\phi_1,\psi_n)$ to obtain new $\phi_1$ and $\psi_n$;
\item delete all edges on $\hat{\phi}_1$ and all edges on $\hat{\psi}_n$, each of which is smaller than merging $(\hat{\phi}_1,\hat{\psi}_n)$ to obtain new $\hat{\phi}_1$ and $\hat{\psi}_n$;
\item concatenate $\phi_1$ and $\hat{\phi_1}$ to obtain $\phi_1\circ\hat{\phi}_1$ (so, necessarily, $\psi_n$ and $\hat{\psi}_n$ are concatenated simultaneously and we obtain $\psi_n\circ\hat{\psi}_n$);
\item identify $S_1,\widehat{S}_1^{(2)},\widehat{S}_1^{(3)},\ldots,\widehat{S}_1^{(k)}$; identify $\widehat{R}_1,R_1^{(2)},R_1^{(3)},\ldots,R_1^{(k)}$; identify $R_2^{(i)}$ and $\widehat{S}_2^{(i)}$ for $1\le i\le n-1$.
\end{enumerate}

Obviously, such procedure produces a $(k+n-1,n)$-graph with two distinct sources $S_1,S_2$ and two sinks $\widehat{R}_1$ and $\widehat{R}_2$, a set of Menger's paths $\{\phi_1\circ\hat{\phi_1},\phi_2,\phi_3,\ldots,\phi_n,\hat{\phi}_2,\hat{\phi}_3,\ldots,\hat{\phi}_k\}$ from $S_1$ to $\widehat{R}_1$ and a set of Menger's paths $\{\psi_1\circ\hat{\psi}_1,\psi_2\circ\hat{\psi}_2,\ldots,\psi_n\circ\hat{\psi}_n\}$ from $S_2$ to $\widehat{R}_2$.

For example, in Figure~\ref{picconcatenation}, we concatenate $\mathcal{F}(2,2)$ and a non-reroutable $(2,2)$-graph to obtain a $(3,2)$-graph. We have the following lemma, whose proof is similar to Lemma~\ref{Fnn} and thus omitted.
\begin{figure}
\psfrag{axx}{$\footnotesize\textrm{(a)}$}\psfrag{bxx}{$\footnotesize\textrm{(b)}$}\psfrag{cxx}{$\footnotesize\textrm{(c)}$}
\psfrag{s1a}{$S_1$}\psfrag{r1a}{$R_1$}
\psfrag{s1a1}{$S_1^{(1)}$}\psfrag{s1a2}{$S_1^{(2)}$}\psfrag{s1a3}{$S_1^{(3)}$}
\psfrag{r1a1}{$R_1^{(1)}$}\psfrag{r1a2}{$R_1^{(2)}$}\psfrag{r1a3}{$R_1^{(3)}$}
\psfrag{s2a}{$S_2$}
\psfrag{r2a}{$R_2$}
\psfrag{phi1}{$\phi_1$}\psfrag{phi2}{$\phi_2$}\psfrag{phi3}{$\phi_3$}
\centering
\includegraphics[width=0.65\textwidth]{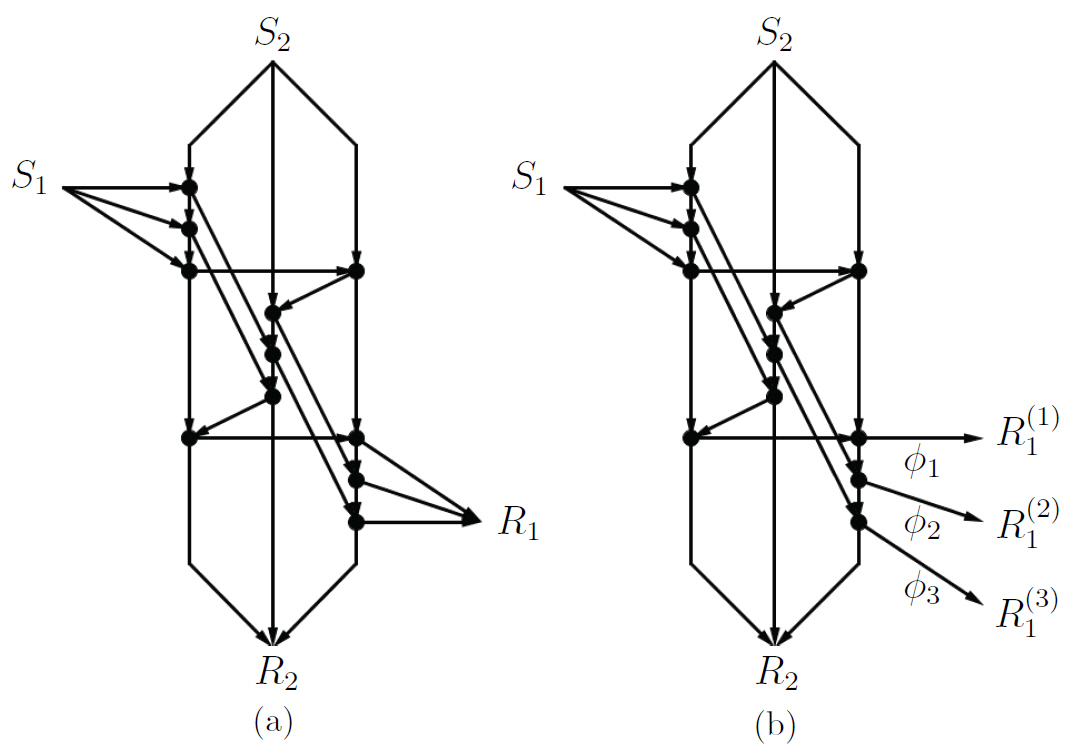}
\caption{(a) Graph $\mathcal{F}(3,3)$ with $11$ mergings (b) Splitting of $R_1$ in $\mathcal{F}(3,3)$}\label{pic33blockwithsplitted}
\end{figure}

\begin{lem} \label{concatenationlemma}
The concatenated graph as above is a non-reroutable $(k+n-1, n)$-graph with the number of mergings equal to $|\mathcal{F}(n, n)|_{\mathcal{M}}+|\mathcal{G}(k, n)|_{\mathcal{M}}-1$.
\end{lem}

We are now ready for the following theorem, which gives us a lower bound on $\mathcal{M}(m, n)$.
\begin{thm}\label{Mmnlb}
$$
\mathcal{M}(m,n)\ge 2mn-m-n+1.
$$
\end{thm}

\begin{proof}\ Without loss of generality, assume that $m \leq n$. For $1 \leq m' \leq m$ and $1\leq n' \leq n$, we will iteratively construct a sequence of non-reroutable $(m', n')$-graphs with $2m'n'-m'-n'+1$ mergings, which immediately implies the theorem.

First, for any $k$, $\mathcal{H}(1, k)$, a non-reroutable $(1, k)$-graph can be given by specifying its mergings sequence
$$
\Omega=[(1,1),(1,2),\ldots,(1,k)].
$$

Next, consider the case $2 \leq m \leq n$. Assume that for any $m',n'$ such that $m'\leq n'$, $m' \leq m$, $n' \leq n$, however $(m', n') \neq (m, n)$, we have constructed a non-reroutable $(m', n')$-graph, which is effectively a non-reroutable $(n', m')$-graph as well. We obtain a new $(m, n)$-graph through the following procedure:
\begin{enumerate}
\item if $m=n$, concatenate $\mathcal{F}(m, m)$ and an already constructed non-reroutable $(1, m)$-graph $\mathcal{H}(1,m)$;
\item if $m < n$, concatenate $\mathcal{F}(m, m)$ and an already constructed non-reroutable $(n-m+1, m)$-graph.
\end{enumerate}

For the first case, according to Lemma~\ref{concatenationlemma}, the obtained graph is non-reroutable $(m, m)$-graph with the number of mergings
$$
(2m^2-3m+2)+m-1=2m^2-2m+1.
$$

Similarly, for the second case, the obtained graph is a non-reroutable $(m, n)$-graph with the number of mergings
\begin{displaymath}
(2m^2-3m+2)+(2(n-m+1)m-(n-m+1)-m+1)-1=2mn-m-n+1.\\
\end{displaymath}

We then have established the theorem.
\end{proof}

\begin{figure}
\psfrag{s1a}{$S_1$}
\psfrag{r1a}{$R_1$}
\psfrag{s2a}{$S_2$}
\psfrag{r2a}{$R_2$}
\psfrag{s1p}{$\widehat{S}_1$}
\psfrag{r1p}{$\widehat{R}_1$}
\psfrag{s2p}{$\widehat{S}_2$}
\psfrag{r2p}{$\widehat{R}_2$}
\psfrag{psi1}{$\psi_1$}
\psfrag{psi2}{$\psi_2$}
\psfrag{phi1}{$\phi_1$}
\psfrag{phi2}{$\phi_2$}
\psfrag{psi1p}{$\hat{\psi}_1$}
\psfrag{psi2p}{$\hat{\psi}_2$}
\psfrag{phi1p}{$\hat{\phi}_1$}
\psfrag{phi2p}{$\hat{\phi}_2$}
\psfrag{phi1phi1p}{$\phi_1 \circ \hat{\phi}_1$}
\psfrag{psi1psi1p}{$\psi_1 \circ \hat{\psi}_1$}
\psfrag{psi2psi2p}{$\psi_2 \circ \hat{\psi}_2$}
\centering
\includegraphics[width=0.7\textwidth]{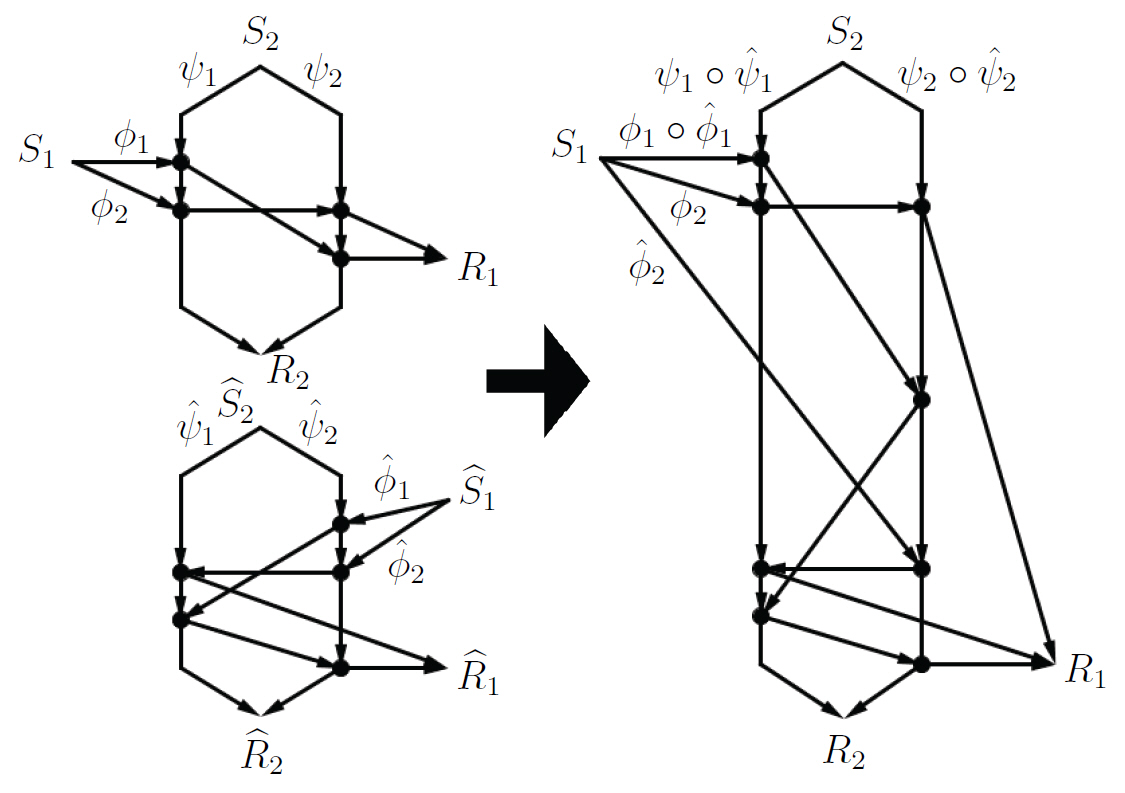}
\caption{Concatenation of $\mathcal{F}(2,2)$ and a non-reroutable $(2,2)$-graph}\label{picconcatenation}
\end{figure}

\begin{exmp}
To construct a non-reroutable $(4, 6)$-graph with $39$ mergings, one can concatenate $\mathcal{F}(4, 4)$ and a non-reroutable $(3, 4)$-graph, which can be obtained by concatenating $\mathcal{F}(3, 3)$ and a non-reroutable $(2, 3)$-graph. The latter can be obtained by concatenating $\mathcal{F}(2, 2)$ and a non-reroutable $(2, 2)$-graph. Finally, a non-reroutable $(2, 2)$-graph can be obtained by concatenating $\mathcal{F}(2, 2)$ and $\mathcal{H}(1, 2)$. One readily checks that the number of mergings in the eventually obtained graph is
\begin{displaymath}
|\mathcal{F}(4,4)|_{ \mathcal{M}}+|\mathcal{F}(3,3)|_{ \mathcal{M}}+|\mathcal{F}(2,2)|_{ \mathcal{M}}+
|\mathcal{F}(2,2)|_{ \mathcal{M}}+|\mathcal{H}(1,2)|_{ \mathcal{M}}-4=22+11+4+4+2-4=39.
\end{displaymath}
\end{exmp}

\begin{thm}\label{mupperbound}
\begin{displaymath}
\mathcal{M}(m, n) \leq (m+n-1)+(mn-2)\floor{\dfrac{m+n-2}{2}}.
\end{displaymath}
\end{thm}

\begin{proof}
Consider any $(m, n)$-graph $G$ with distinct sources $S_1, S_2$, sinks $R_1, R_2$, a set of Menger's paths $\phi=\{\phi_1, \phi_2, \ldots, \phi_m\}$ from $S_1$ to $R_1$, a set of Menger's paths $\psi=\{\psi_1, \psi_2, \ldots, \psi_n\}$ from $S_2$ to $R_2$. As discussed in Section~\ref{AA-sequences}, we assume that all the AA-sequences are of positive lengths. By Lemma~\ref{lengthofAAsequence}, the shortest $\phi$-AA-sequence and $\psi$-AA-sequence are both of length $1$. We then consider the following two cases (note that the following two cases may not be mutually exclusive):

\underline{Case 1:} there exists a shortest $\phi$-AA-sequence and a shortest $\psi$-AA-sequence, which are associated with the same path pair. By Lemma~\ref{each-path-pair-at-most-once}, there are at most $\floor{\frac{m+n}{2}}$ mergings corresponding to this path pair, and at most $\floor{\frac{m+n-2}{2}}$ mergings corresponding to any other path pair. So, the number of mergings is upper bounded by\
\begin{equation}\label{mneqI}
\floor{\frac{m+n}{2}}+(mn-1)\floor{\frac{m+n-2}{2}}.
\end{equation}

\underline{Case 2:} there exists a shortest $\phi$-AA-sequence and a shortest $\psi$-AA-sequence, which are associated with two distinct path pairs. Again, by Lemma~\ref{each-path-pair-at-most-once}, there are at most $\floor{\frac{m+n-1}{2}}$ mergings corresponding to each of these two path pairs, and at most $\floor{\frac{m+n-2}{2}}$ mergings corresponding to any other path pair. So, the number of mergings is upper bounded by
\begin{equation}\label{mneqII}
2\floor{\frac{m+n-1}{2}}+(mn-2)\floor{\frac{m+n-2}{2}}.
\end{equation}
Then, $\mcm(m,n)\le \max\{(\ref{mneqI}),(\ref{mneqII})\}$. Straightforward computations then lead to the theorem.
\end{proof}

\begin{rem}
It has been established in~\cite{la2006} that
$$
n(n-1)/2 \le \mstar(n,n)\le n^3.
$$
Summarizing all the four bounds we obtain, we have
\begin{displaymath}
\begin{array}{rcl}
\vspace{0.1cm}(n-1)^2 \le&\hspace{-1.5ex}\mstar(n,n)\hspace{-1.5ex}&\le \ceil{\dfrac{n}{2}}(n^2-4n+5) ,\\
2mn-m-n+1 \le&\hspace{-1.5ex}\mcm(m,n)\hspace{-1.5ex}&\le (m+n-1)+(mn-2)\floor{\dfrac{m+n-2}{2}}.
\end{array}
\end{displaymath}
\end{rem}

\subsection{Bounds on $\mcm(3, n)$}
It has been shown in \cite{ha2011} that for any $k$, there exists $C_k$ such that $\mcm(k,n)\le C_kn$ for all $n$, where $C_k$ can be rather loose. The following result refines the above result for the case when $k=3$.

\begin{thm}
$$\mcm(3,n)\leq 14 n.$$
\end{thm}

\begin{proof}

Consider any non-reroutable $(3, n)$-graph $G$ with distinct sources $S_1, S_2$, sinks $R_1, R_2$, a set of Menger's paths $\phi=\{{\phi}_1,{\phi}_2,{\phi}_{3}\}$ from $S_1$ to $R_1$ and a set of Menger's paths $\psi=\{{\psi}_1,{\psi}_2,...,{\psi}_{n}\}$ from $S_2$ to $R_2$. If a merging is the smallest (the largest) one on a ${\psi}$-path, we say it is an \emph{$x$-terminal} (\emph{$y$-terminal}) merging on the ${\psi}$-path, or simply a \emph{${\psi}$-terminal} merging.

Consider the following iterative procedure (Figures~\ref{pic14mergings},~\ref{pic3ncase1} and~\ref{pic3ncase2} roughly illustrate the procedure), where, for notational simplicity, we treat a graph as a union of its vertex set and edge set. Initially set $\mathbb{S}^{(0)}=\emptyset$, and $\mathbb{R}^{(0)}=G$. Now for each $j=1,2,3$, pick a merging $\gamma_{0,j}$ such that $\gamma_{0,j}$ belongs to path ${\phi}_j$ and
$$
|\mathbb{R}^{(0)}|t(\gamma_{0,1}),t(\gamma_{0,2}),t(\gamma_{0,3}))|_{\mathcal{M}}=14,
$$
where one can choose $\gamma_{0,j}$ to be $S_1$ if such merged
subpath does not exist on ${\phi}_j$. Now set
$$
\mathbb{L}_1= \mathbb{R}^{(0)}|t(\gamma_{0,1}),t(\gamma_{0,2}),t(\gamma_{0,3})),
$$
and
$$
\mathbb{S}^{(1)}=\mathbb{S}^{(0)} \cup \mathbb{L}_1, \qquad \mathbb{R}^{(1)}=G \setminus \mathbb{S}^{(1)}.
$$
Suppose that we already obtain
$$
\mathbb{L}_i= \mathbb{R}^{(i-1)}|t(\gamma_{i-1,1}),t(\gamma_{i-1,2}),t(\gamma_{i-1,3})),
$$
and
$$
\mathbb{S}^{(i)}=\mathbb{S}^{(i-1)} \cup \mathbb{L}_i, \qquad \mathbb{R}^{(i)}=G \setminus \mathbb{S}^{(i)},
$$
where $\mathbb{L}_i$ contains exactly 14 mergings and at least two ${\psi}$-terminal merged subpaths. We then continue to pick merged subpath $\gamma_{i,j}$ on ${\phi}_j$ from $\mathbb{R}^{(i)}$ such that
$$
|\mathbb{R}^{(i)}|t(\gamma_{i,1}),t(\gamma_{i,2}),t(\gamma_{i,3}))|_{\mathcal{M}}=14
$$
and there are at least two ${\psi}$-terminal mergings in $\mathbb{R}^{(i)}|t(\gamma_{i,1}),t(\gamma_{i,2}),t(\gamma_{i,3}))$. If such $\gamma_{i, j}$'s exist, set
$$
\mathbb{L}_{i+1}=\mathbb{R}^{(i)}|t(\gamma_{i,1}),t(\gamma_{i,2}),t(\gamma_{i,3})),
$$
and if $|\mathbb{R}^{(i)}|< 14$, set $\mathbb{L}_{i+1}=\mathbb{R}^{(i)}$ and terminate the iterative procedure. So far, for any obtained ``block'' $\mathbb{L}_{i+1}$, either we have $|\mathbb{L}_{i+1}|_{\mathcal{M}} < 14$ or ($|\mathbb{L}_{i+1}|_{\mathcal{M}} = 14$ and there are at least two ${\psi}$-terminal mergings in $\mathbb{L}_{i+1}$); such block $\mathbb{L}_{i+1}$ is said to be \emph{normal}. If $|\mathbb{R}^{(i)}| \geq 14$, however, we cannot find a normal block, we continue the procedure and define a \emph{singular} $\mathbb{L}_{i+1}$ in the following.
\begin{figure}
\psfrag{s1a}{$S_1$}
\psfrag{r1a}{$R_1$}
\psfrag{l1a}{$\mathbb{L}_1$}
\psfrag{l2a}{$\mathbb{L}_2$}
\psfrag{s1x}{$\mathbb{S}^{(1)}$}
\psfrag{s2x}{$\mathbb{S}^{(2)}$}
\psfrag{psi1}{$\phi_1$}
\psfrag{psi2}{$\phi_2$}
\psfrag{psi3}{$\phi_3$}
\psfrag{g01}{$\gamma_{0,1}$}
\psfrag{g02}{$\gamma_{0,2}$}
\psfrag{g03}{$\gamma_{0,3}$}
\psfrag{g11}{$\gamma_{1,1}$}
\psfrag{g12}{$\gamma_{1,2}$}
\psfrag{g13}{$\gamma_{1,3}$}
\psfrag{vdots}{$\vdots$}
\centering
\includegraphics[width=0.55\textwidth]{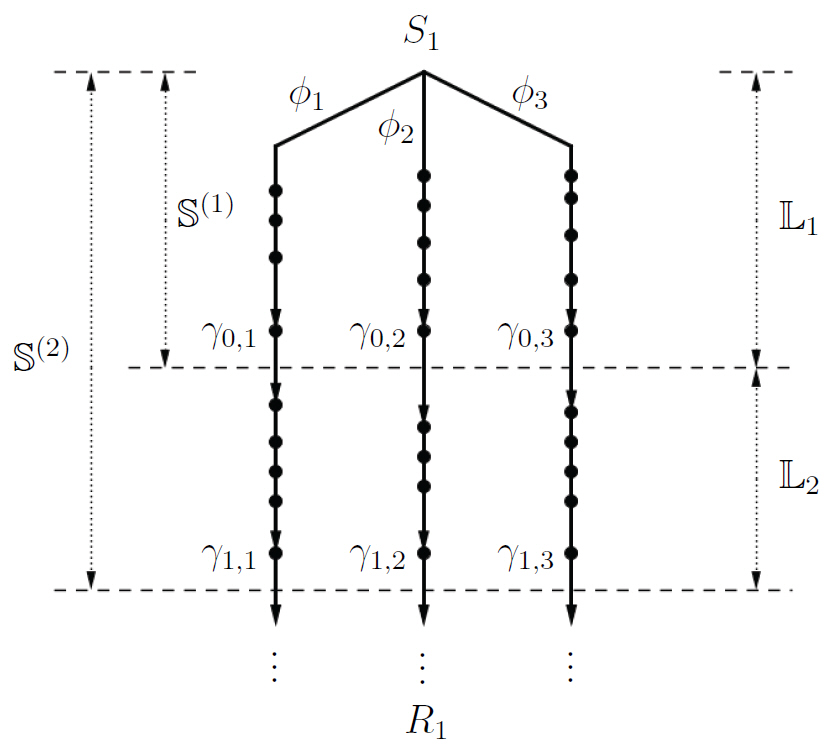}
\caption{Partition a $(3,n)$-graph $G$ into blocks}\label{pic14mergings}
\end{figure}

Note that $\mathbb{S}^{(i)}=\mathbb{L}_1\cup \mathbb{L}_2\cup \cdots \cup \mathbb{L}_i$. Let $z_i=\sum_{j=1}^{i}(x_j-y_j)$, where $x_i$ and $y_i$ denote the number of $x$-terminal and $y$-terminal mergings in the ${\psi}$-paths in $\mathbb{L}_i$, respectively; then $z_i$ is the number of $\psi$-paths which can continue to merge within $\mathbb{R}^{(i)}$. If a normal block does not exist after $i$ iterations, necessarily we will have $z_i\geq 3$ (suppose $z_i \leq 2$, by the fact that $\mcm(3,3)=13$ (see Theorem~\ref{m33}), we would be able to obtain a normal block $\mathbb{L}_{i+1}$, which contains two $x$-terminals or (an $x$-terminal and a $y$-terminal)). We say a merged subpath is {\em critical} within a subgraph of $G$ if the corresponding $\psi$-path, after merging at this merged subpath, does not merge anymore within this subgraph. It then follows that the number of the critical merged subpaths within  $\mathbb{S}^{(i)}$ is $z_i$.

Now, let $\mathbb{K}_i$ denote the set of all the merged subpaths within $\mathbb{R}^{(i)}$, each of which can semi-reach the tail of some critical merged subpath within $\mathbb{S}^{(i)}$ against $\psi$. One checks at least one of those $\phi$-paths, each of which contains at least one critical merged subpath within $\mathbb{S}^{(i)}$, does not contain any merged subpath within $\mathbb{K}_i$. Without loss of generality, we assume that ${\phi}_3 \cap \mathbb{K}_i = \emptyset$. Now we consider the following two cases:

\underline{Case 1:} ${\phi}_1\cap \mathbb{K}_i\neq\emptyset$ and ${\phi}_2\cap \mathbb{K}_i\neq\emptyset$. As shown in Figure~\ref{pic3ncase1}, assume that within $\mathbb{K}_i$, $\lambda_{i,1}, \lambda_{i,2}$ are the largest merged subpaths on $\phi_1,\phi_2$, respectively. Now, set
$$
\mathbb{L}_{i+1}=\mathbb{R}^{(i)}|t(\lambda_{i,1}),t(\lambda_{i,2})), \qquad
\mathbb{Q}_i={\phi}_1[t(\gamma_{i-1,1}),t(\lambda_{i,1})]\cup
{\phi}_2[t(\gamma_{i-1,2}),t(\lambda_{i,2})].
$$
Note that for $\lambda_{i,j}$, $j=1,2$, the associated $\psi$-path, from $\lambda_{i,j}$, may merge outside $\mathbb{Q}_i$ next time. If this $\psi$-path merges within $\mathbb{Q}_i$ again after a number of mergings
outside $\mathbb{Q}_i$, we call it an {\em excursive} $\psi$-path. One checks that there are at most one excursive $\psi$-path (since, otherwise, we can find a cycle in $G$, which is a contradiction). On the other hand, for any merged subpath from $\mathbb{K}_i$ other than $\lambda_{i,1}, \lambda_{i,2}$ , say $\mu$, the associated $\psi$-path, from $\mu$, can only merge within $\mathbb{Q}_i$ and will not merge outside $\mathbb{Q}_i$. So, the number of {\bf connected} $\psi$-paths that contain at least one merged subpath within $\mathbb{L}_{i+1}\cap \mathbb{Q}_i$ is upper bounded by $y_{i+1}+2$. Then, by the fact that $\mcm(2,n)=3n-1$ (see Theorem~\ref{ThreeProofs}), we have
\begin{equation} \label{with-T}
|\mathbb{L}_{i+1}\cap \mathbb{Q}_i|_{\mathcal{M}} \leq 3(y_{i+1}+2)-1.
\end{equation}
It is clear that all non-excursive $\psi$-paths that contain at least one merged subpath within $\mathbb{L}_{i+1}\setminus \mathbb{Q}_i$ must have $x$-terminals in $\mathbb{L}_{i+1}$. Thus, again by the fact that $\mcm(2,n)=3n-1$, we have
\begin{equation} \label{not-with-T}
|\mathbb{L}_{i+1}\setminus \mathbb{Q}_i|_{\mathcal{M}} \leq 3(x_{i+1}+1)-1.
\end{equation}
It then immediately follows from (\ref{with-T}) and (\ref{not-with-T}) that
$$
|\mathbb{L}_{i+1}|_{\mathcal{M}}= |\mathbb{L}_{i+1}\cap \mathbb{Q}_i|_{\mathcal{M}}+|\mathbb{L}_{i+1}\setminus \mathbb{Q}_i|_{\mathcal{M}} \leq 3(x_{i+1}+y_{i+1})+7.
$$

\begin{figure}
\psfrag{s1a}{$S_1$}
\psfrag{r1a}{$R_1$}
\psfrag{lia}{$\mathbb{L}_i$}
\psfrag{li1a}{$\mathbb{L}_{i+1}$}
\psfrag{psi1}{$\phi_1$}
\psfrag{psi2}{$\phi_2$}
\psfrag{psi3}{$\phi_3$}
\psfrag{ller}{$y$}
\psfrag{ster}{$x$}
\psfrag{gi11}{$\gamma_{i-1\hspace{-0.01cm},\hspace{-0.01cm}1}$}
\psfrag{gi12}{$\gamma_{i-1\hspace{-0.01cm},\hspace{-0.01cm}2}$}
\psfrag{gi13}{$\gamma_{i-1\hspace{-0.01cm},\hspace{-0.01cm}3}$}
\psfrag{xii1}{$\lambda_{i\hspace{-0.01cm},\hspace{-0.01cm}1}$}
\psfrag{xii2}{$\lambda_{i\hspace{-0.01cm},\hspace{-0.01cm}2}$}
\psfrag{vdots}{$\vdots$}
\psfrag{lia}{$\mathbb{L}_i$}
\psfrag{li1}{$\mathbb{L}_{i+1}$}
\psfrag{hti}{$\mathbb{Q}_i$}
\psfrag{oth}{$\mathbb{L}_{i+1}\backslash\mathbb{Q}_i$}
\psfrag{exc}{$\scriptsize\textrm{excursive}$}
\psfrag{pat}{$\scriptsize\textrm{path}$}
\centering
\includegraphics[width=0.55\textwidth]{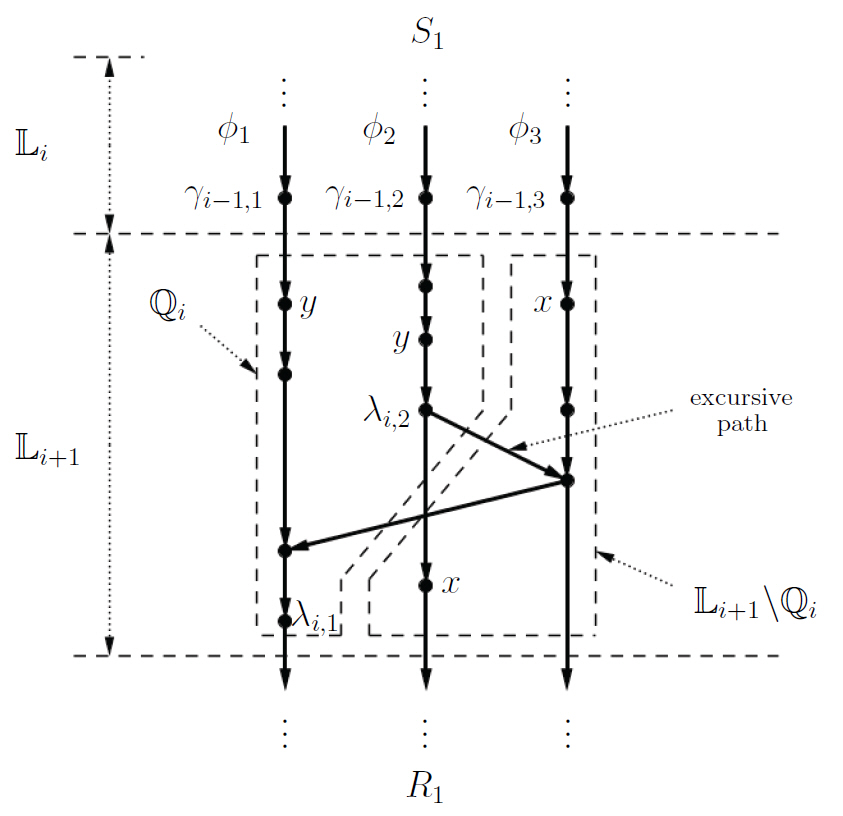}
\caption{Case 1}\label{pic3ncase1}
\end{figure}
Next, we claim that $x_{i+1}+y_{i+1}\geq 3$. To see this, suppose, by contradiction, that $x_{i+1}+y_{i+1}\leq 2$. Observing that $y_{i+1}\geq z_i-2 \geq 1$, we then consider the following two cases:

If $x_{i+1}+y_{i+1}= 2$, we have
$$
|\mathbb{L}_{i+1}|_{\mathcal{M}} \leq 3(x_{i+1}+y_{i+1})+7=13,
$$
which implies that we can continue to choose a normal block (with two ${\psi}$-terminal merged
subpaths), a contradiction.

If $x_{i+1}+y_{i+1}= 1$, we have $x_{i+1}=0$, $y_{i+1}=1$. Note that if there is no excursive $\psi$-path, we have $z_i\le y_{i+1}+2$; if there is one excursive $\psi$-path, then $z_i\le y_{i+1}+1$. This, together with $z_i \geq 3$, implies that $z_i=3$ and there is no excursive $\psi$-path. Consequently, we have
$$
|\mathbb{L}_{i+1}\cap \mathbb{Q}_i|_{\mathcal{M}} \leq \mcm(2,3)=8, \qquad |\mathbb{L}_{i+1}\setminus \mathbb{Q}_i|_{\mathcal{M}}=0.
$$
But this, together with $\mcm(3,3)=13$, implies that we can continue to choose a normal block with an $x$-terminal and a $y$-terminal merged subpaths, which is a contradiction.

\underline{Case 2:} ${\phi}_1\cap \mathbb{K}_i\neq\emptyset$ and ${\phi}_2\cap \mathbb{K}_i=\emptyset$. As shown in Figure~\ref{pic3ncase2}, assume that within $\mathbb{K}_i$, $\lambda_{i,1}$ is the largest merged subpath on ${\phi}_1$. Apparently, there is no excursive $\psi$-path. By the fact that $\mathcal{M}(1, n)=n$ (see Example $2.15$ of~\cite{ha2011}) and $\mathcal{M}(2, n)=3n-1$, we have
$$
|\mathbb{L}_{i+1}\cap \mathbb{Q}_i|_\mcm\leq y_{i+1}+1, \qquad |\mathbb{L}_{i+1}\setminus \mathbb{Q}_i|_\mcm\leq 3x_{i+1}-1.
$$
It then immediately follows that $|\mathbb{L}_{i+1}|_\mcm\leq 3x_{i+1}+y_{i+1}$.

Similarly as before, we claim that $x_{i+1}+y_{i+1} \geq 3$. To see this, suppose, by contradiction, that $x_{i+1}+y_{i+1}\leq 2$. From $y_{i+1}+1\geq z_i\geq 3$, we infer that
$y_{i+1}=2$ and $x_{i+1}=0$, and further
$$
|\mathbb{L}_{i+1}|_\mcm \leq 3x_{i+1}+y_{i+1}=2,
$$
which implies that we can in fact obtain a normal block, a contradiction.
\begin{figure}
\psfrag{s1a}{$S_1$}
\psfrag{r1a}{$R_1$}
\psfrag{lia}{$\mathbb{L}_i$}
\psfrag{li1a}{$\mathbb{L}_{i+1}$}
\psfrag{psi1}{$\phi_1$}
\psfrag{psi2}{$\phi_2$}
\psfrag{psi3}{$\phi_3$}
\psfrag{ller}{$y$}
\psfrag{ster}{$x$}
\psfrag{gi11}{$\gamma_{i-1\hspace{-0.01cm},\hspace{-0.01cm}1}$}
\psfrag{gi12}{$\gamma_{i-1\hspace{-0.01cm},\hspace{-0.01cm}2}$}
\psfrag{gi13}{$\gamma_{i-1\hspace{-0.01cm},\hspace{-0.01cm}3}$}
\psfrag{xii1}{$\lambda_{i\hspace{-0.01cm},\hspace{-0.01cm}1}$}
\psfrag{xii2}{$\lambda_{i\hspace{-0.01cm},\hspace{-0.01cm}2}$}
\psfrag{vdots}{$\vdots$}
\psfrag{lia}{$\mathbb{L}_i$}
\psfrag{li1}{$\mathbb{L}_{i+1}$}
\psfrag{hti}{$\mathbb{Q}_i$}
\psfrag{oth}{$\mathbb{L}_{i+1}\backslash\mathbb{Q}_i$}
\psfrag{exc}{$\scriptsize\textrm{excusive}$}
\psfrag{pat}{$\scriptsize\textrm{path}$}
\centering
\includegraphics[width=0.55\textwidth]{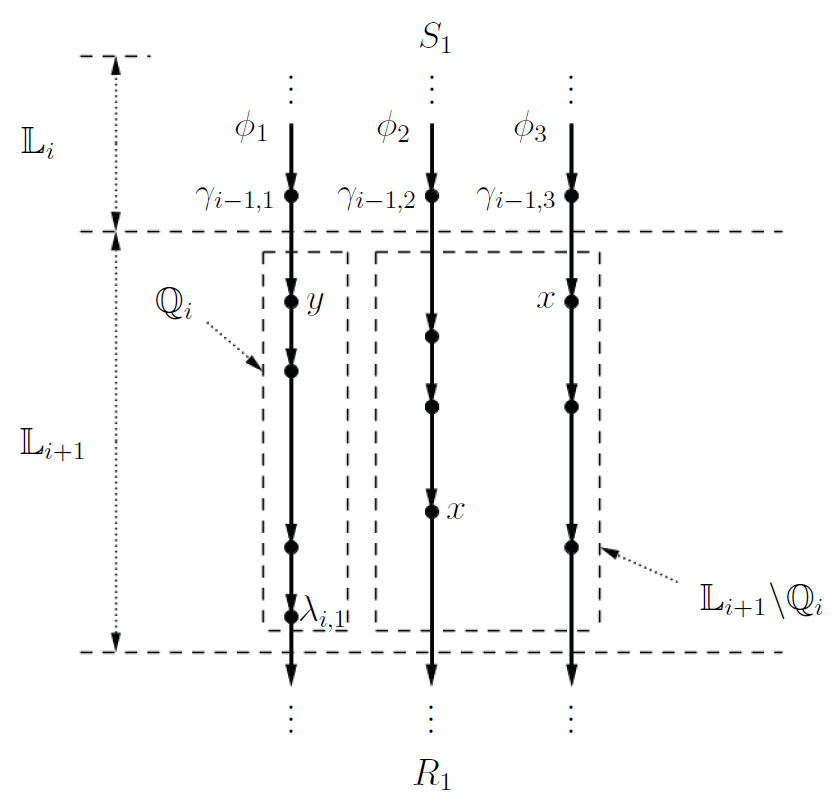}
\caption{Case 2}\label{pic3ncase2}
\end{figure}

Combining the above two cases, we conclude that the number of merged subpaths within the singular block $\mathbb{L}_{i+1}$ is upper bounded by $3(x_{i+1}+y_{i+1})+7$, where $x_{i+1}+y_{i+1} \geq 3$.

We continue these operations in an iterative fashion to further obtain normal blocks and singular blocks until there are no merged subpaths left in the graph. Suppose there are $n_1$ singular blocks $\mathbb{L}_{j_1}, \mathbb{L}_{j_2}, \ldots, \mathbb{L}_{j_{n_1}}$ and $n_2$ normal blocks. Note that each singular block has at least three ${\psi}$-terminal merged subpaths and each normal block except the last one has at least two ${\psi}$-terminal merged subpaths. If the last normal block has at least two $\psi$-terminal merged subpaths, we then have
$$
3n_1 \leq \sum_{i=1}^{n_1}(x_{j_i}+y_{j_i})\leq 2n-2n_2.
$$
It then follows that
\begin{equation} \label{last-one-has-more-than-one}
|G|_{\mathcal{M}} \leq 14n_2+\sum_{i=1}^{n_1}[3(x_{j_i}+y_{j_i})+7] \leq 14n_2+7n_1+3(2n-2n_2)=6n+8n_2+7n_1 \leq 14n.
\end{equation}
If the last normal block has only one $\psi$-terminal merged subpath, necessarily, there are at most three mergings in the last normal block, we then have
$$
3n_1 \leq \sum_{i=1}^{n_1}(x_{j_i}+y_{j_i})\leq 2n-2(n_2-1)-1.
$$
It then follows that
\begin{equation} \label{last-one-has-only-one}
|G|_{\mathcal{M}} \leq 14(n_2-1)+3+\sum_{i=1}^{n_1}[3(x_{j_i}+y_{j_i})+7] \leq 6n+8n_2+7n_1-8 \leq 14n.
\end{equation}
Combining (\ref{last-one-has-more-than-one}) and (\ref{last-one-has-only-one}), we then have established the theorem.

\end{proof}

\section{Inequalities}

Consider two non-reroutable $(n,n)$-graph $G^{(1)}, G^{(2)}$. For $j=1,2$, assume that $G^{(j)}$ has one source $S^{(j)}$, two sinks $R_1^{(j)},R_2^{(j)}$. Let $\phi^{(j)}=\{\phi_1^{(j)},\phi_2^{(j)},\ldots,\phi_n^{(j)}\}$ denote the set of Menger's paths from $S^{(j)}$ to $R_1^{(j)}$ and $\psi^{(j)}=\{\psi_1^{(j)},\psi_2^{(j)},\ldots,\psi_n^{(j)}\}$ denote the set of Menger's paths from $S^{(j)}$ to $R_2^{(j)}$. As before, we assume that, for $1\le i \le n$, paths $\phi_i^{(j)}$ and $\psi_i^{(j)}$ share a starting subpath.

Now, consider the following procedure of concatenating graphs $G^{(1)}$ and $G^{(2)}$:
\begin{enumerate}
\item reverse the direction of each edge in $G^{(2)}$ to obtain a new graph $\widehat{G}^{(2)}$ (for $1 \leq i \leq n$, path $\phi_i^{(2)}$ in $G^{(2)}$ becomes path $\hat{\phi}_i^{(2)}$ in $\widehat{G}^{(2)}$ and path $\psi_i^{(2)}$ in $G^{(2)}$ becomes path $\hat{\psi}_i^{(2)}$ in $\widehat{G}^{(2)}$);
\item split $S^{(1)}$ into $n$ copies $S_1^{(1)},S_2^{(1)},\ldots,S_n^{(1)}$ in $G^{(1)}$ such that paths $\phi_i^{(1)}$ and $\psi_i^{(1)}$ have the same starting point $S_i^{(1)}$; split $S^{(2)}$ into $n$ copies $S_1^{(2)},S_2^{(2)},\ldots,S_n^{(2)}$ in $\widehat{G}^{(2)}$ such that paths $\hat{\phi}_i^{(2)}$ and $\hat{\psi}_i^{(2)}$ have the same ending point $S_i^{(2)}$;
\item for $1\le i\le n$, identify $S_i^{(1)}$ and $S_i^{(2)}$.
\end{enumerate}

Obviously, such procedure produces an $(n,n)$-graph with two distinct sources $R_1^{(2)}, R_2^{(2)}$, two sinks $R_1^{(1)}, R_2^{(1)}$, a set of Menger's paths $\{\hat{\phi}_1^{(2)}\circ \phi_1^{(1)},\hat{\phi}_2^{(2)}\circ \phi_2^{(1)},\ldots,\hat{\phi}_n^{(2)}\circ \phi_n^{(1)}\}$ from $R_1^{(2)}$ to $R_1^{(1)}$ and a set of Menger's paths $\{\hat{\psi}_1^{(2)}\circ \psi_1^{(1)},\hat{\psi}_2^{(2)}\circ \psi_2^{(1)},\ldots,\hat{\psi}_n^{(2)}\circ \psi_n^{(1)}\}$ from $R_2^{(2)}$ to $R_2^{(1)}$. See Figure~\ref{pictwomstar} for an example where we concatenate two $(3,3)$-graphs.

We then have the following lemma.
\begin{lem}
The concatenated graph as above is a non-reroutable $(n,n)$-graph with $|G^{(1)}|_\mcm+|G^{(2)}|_\mcm+n$ mergings.
\end{lem}
The following theorem then immediately follows.
\begin{thm}\label{mandmstarrelation}
$$
\mcm(n,n)\ge 2\mstar(n,n)+n.
$$
\end{thm}
\begin{figure}
\psfrag{saa}{$S^{(2)}$}
\psfrag{sbb}{$S^{(1)}$}
\psfrag{r1a}{$R_1^{(2)}$}
\psfrag{r2a}{$R_2^{(2)}$}
\psfrag{r1b}{$R_1^{(1)}$}
\psfrag{r2b}{$R_2^{(1)}$}
\centering
\includegraphics[width=0.55\textwidth]{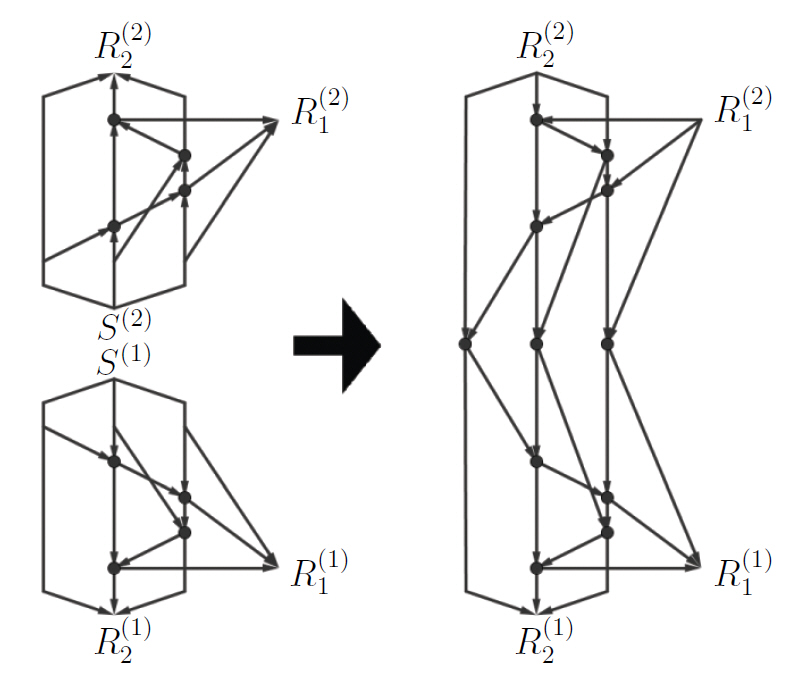}
\caption{Concatenation of two $(3,3)$-graphs}\label{pictwomstar}
\end{figure}

Consider a non-reroutable $(n+1,n+1)$-graph $G^{(1)}$ and a non-reroutable $(n-1,n-1)$-graph $G^{(2)}$. The graph $G^{(1)}$ has one source $S^{(1)}$, two sinks $R_1^{(1)},R_2^{(1)}$, a set of Menger's paths $\phi=\{\phi_0,\phi_1,\ldots,\phi_n\}$ from $S^{(1)}$ to $R_1^{(1)}$ and a set of Menger's paths $\psi=\{\psi_0,\psi_1,\ldots,\psi_n\}$ from $S^{(1)}$ to $R_2^{(1)}$. As discussed in Section~\ref{AA-sequences}, we assume paths $\phi_i$ and $\psi_i$ share a starting subpath  $\omega_i$, and paths $\phi_n$, $\psi_0$ do not merge with any other paths in $G^{(1)}$, directly flowing to the sinks. The graph $G^{(2)}$ has one source $S^{(2)}$, two sinks $R_1^{(2)},R_2^{(2)}$, a set of Menger's paths $\xi=\{\xi_1,\xi_2,\ldots,\xi_{n-1}\}$ from $S^{(2)}$ to $R_1^{(2)}$ and a set of Menger's paths $\eta=\{\eta_1,\eta_2,\ldots,\eta_{n-1}\}$ from $S^{(2)}$ to $R_2^{(2)}$. Again, assume paths $\xi_i$ and $\eta_i$ share a starting subpath.

Now, we consider the following procedure of concatenating graphs $G^{(1)}$ and $G^{(2)}$:
\begin{enumerate}
\item reverse the direction of each edge in $G^{(2)}$ to obtain a new graph $\widehat{G}^{(2)}$ (for $1\le i\le n-1$, path $\xi_i$ in $G^{(2)}$ becomes path $\hat{\xi}_i$ in $\widehat{G}^{(2)}$ and path $\eta_i$ in $G^{(2)}$ becomes path $\hat{\eta}_i$ in $\widehat{G}^{(2)}$);
\item split $S^{(1)}$ into $n+1$ copies $S^{(1)}_0,S^{(1)}_1,\ldots,S^{(1)}_{n}$ in $G^{(1)}$ such that paths $\phi_i$ and $\psi_i$ have the same starting point $S^{(1)}_i$; split $S^{(2)}$ into $n-1$ copies $S^{(2)}_1,S^{(2)}_2,\ldots,S^{(2)}_{n-1}$ in $\widehat{G}^{(2)}$ such that paths $\hat{\xi}_i$ and $\hat{\eta}_i$ have the same ending point $S^{(2)}_i$;
\item delete all edges on $\phi_n$, each of which is larger than $\omega_n$; delete all edges on $\psi_0$, each of which is larger than $\omega_0$;
\item identify $R_1^{(2)}$ and $S_0^{(1)}$; for $1\le i\le n-1$, identify $S_i^{(2)}$ and $S_i^{(1)}$; identify $R_2^{(2)}$ and $S_n^{(1)}$.
\end{enumerate}

Obviously, such procedure produces an $(n,n)$-graph with two distinct sources $R_1^{(2)}, R_2^{(2)}$, two sinks $R_1^{(1)},R_2^{(1)}$, a set of Menger's paths $\{\phi_0,\hat{\xi}_1\circ \phi_1,\hat{\xi}_2\circ \phi_2,\ldots,\hat{\xi}_{n-1}\circ \phi_{n-1},\}$ from $R_1^{(2)}$ to $R_1^{(1)}$ and a set of Menger's paths $\{\hat{\eta}_1\circ \psi_1, \hat{\eta}_2\circ \psi_2, \ldots, \hat{\eta}_{n-1}\circ \psi_{n-1},\psi_n\}$ from $R_2^{(2)}$ to $R_2^{(1)}$. For example, in Figure~\ref{picmstar4422}, we concatenate a $(2,2)$-graph and a $(4,4)$-graph to obtain a $(3,3)$-graph.

We then have the following lemma.

\begin{lem}
The concatenated graph as above is a non-reroutable $(n,n)$-graph with $|G^{(1)}|_\mathcal{M}+|G^{(2)}|_\mathcal{M}+(n-1)$ mergings.
\end{lem}
It immediately follows that
\begin{thm}
$$
\mcm(n,n) \geq \mstar(n+1,n+1)+\mstar(n-1,n-1)+(n-1).
$$
\end{thm}
\begin{figure}
\psfrag{ssg}{\footnotesize $S^{(1)}$}\psfrag{ssh}{\footnotesize $S^{(2)}$}
\psfrag{r1g}{\footnotesize $R^{(1)}_1$}\psfrag{r2g}{\footnotesize $R^{(1)}_2$}
\psfrag{r1h}{\footnotesize $R^{(2)}_1$}\psfrag{r2h}{\footnotesize $R^{(2)}_2$}
\psfrag{eta1}{\footnotesize $\eta_1$}\psfrag{eta2}{\footnotesize $\eta_2$}
\psfrag{xi1}{\footnotesize $\xi_1$}\psfrag{xi2}{\footnotesize $\xi_2$}
\psfrag{psi0}{\footnotesize $\psi_0$}\psfrag{psi1}{\footnotesize $\psi_1$}\psfrag{psi2}{\footnotesize $\psi_2$}\psfrag{psi3}{\footnotesize $\psi_3$}
\psfrag{phi0}{\footnotesize $\phi_0$}\psfrag{phi1}{\footnotesize $\phi_1$}\psfrag{phi2}{\footnotesize $\phi_2$}\psfrag{phi3}{\footnotesize $\phi_3$}
\psfrag{qq0}{\footnotesize $\omega_0$}\psfrag{qq1}{\footnotesize $\omega_1$}\psfrag{qq2}{\footnotesize $\omega_2$}\psfrag{qq3}{\footnotesize $\omega_3$}
\psfrag{xph1}{\footnotesize $\hat{\xi}_1\circ \phi_1$}\psfrag{xph2}{\footnotesize $\hat{\xi}_2\circ \phi_2$}
\psfrag{tps1}{\footnotesize $\hat{\eta}_1\circ \psi_1$}\psfrag{tps2}{\footnotesize $\hat{\eta}_2\circ \psi_2$}
\centering
\includegraphics[width=0.75\textwidth]{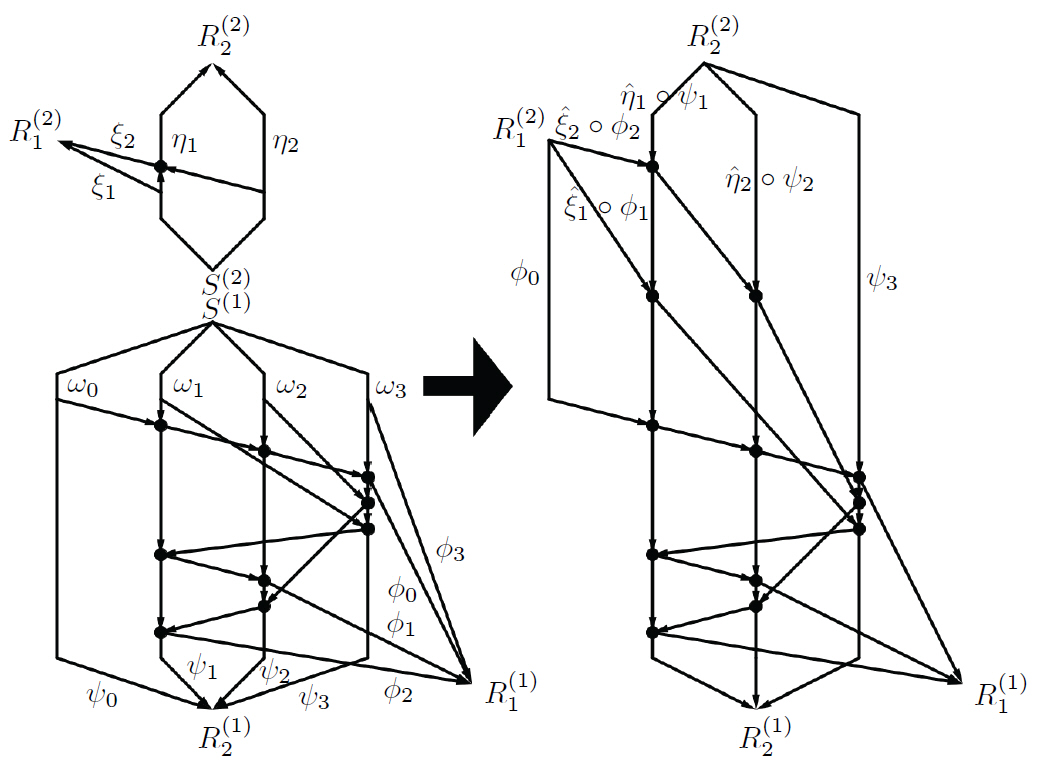}
\caption{Concatenation of a $(2,2)$-graph and a $(4,4)$-graph}\label{picmstar4422}
\end{figure}

Consider $n_1 \leq n_2 \leq \cdots \leq n_k$. For $j=1, 2, \ldots, k-1$, consider a non-reroutable $(n_j,n_k)$-graph $G^{(j)}$ with one source $S^{(j)}$, two sinks $R^{(j)},\widehat{R}^{(j)}$, a set of Menger's paths $\{\phi^{(j)}_1,\phi^{(j)}_2,\ldots,\phi^{(j)}_{n_j}\}$ from $S^{(j)}$ to $R^{(j)}$ and a set of Menger's paths $\{\psi^{(j)}_1,\psi^{(j)}_2,\ldots,\psi^{(j)}_{n_k}\}$ from $S^{(j)}$ to $\widehat{R}^{(j)}$. As before, we assume that paths $\phi^{(j)}_i$ and $\psi^{(j)}_i$ share a starting subpath for $1\le i\le n_j$.

Now, consider the following procedure of concatenating graphs $G^{(1)},G^{(2)},\ldots,G^{(k-1)}$ (see Figure~\ref{picMstar333} for an example):
\begin{enumerate}
\item for $1\le j\le k-2$, split $\widehat{R}^{(j)}$ into $n_k$ copies $\widehat{R}^{(j)}_1,\widehat{R}^{(j)}_2,\ldots,\widehat{R}^{(j)}_{n_k}$ such that path $\psi^{(j)}_i$ has the ending point $\widehat{R}^{(j)}_i$;
\item for $2\le j\le k-1$, split $S^{(j)}$ into $n_k$ copies $S^{(j)}_1,S^{(j)}_2,\ldots,S^{(j)}_{n_k}$ such that paths $\phi^{(j)}_i$ and $\psi^{(j)}_i$ have the same starting point $S^{(j)}_i$;
\item for $1\le j\le k-2$ and $1\le i\le n_k$, identify $\widehat{R}^{(j)}_i$ and $S^{(j+1)}_i$.
\end{enumerate}

Relabel $S^{(1)}, \widehat{R}^{(k-1)}$ as $S, R^{(k)}$, respectively. We then have an $(n_1,n_2,\ldots,n_k)$-graph with one source $S$, $k$ sinks $R^{(1)},R^{(2)},\ldots,R^{(k)}$ and a set of Menger's paths $\{\delta^{(j)}_1,\delta^{(j)}_2,\ldots,\delta^{(j)}_{n_j}\}$ from $S$ to $R^{(j)}$ for $1\le j\le k$, where
\begin{displaymath}
\delta^{(j)}_i=\left\{\begin{array}{ll}
\psi^{(1)}_i\circ\psi^{(2)}_i\circ\cdots\circ\psi^{(j-1)}_i\circ\phi^{(j)}_i &\ \textrm{ if } 1\le j<k,\\
\psi^{(1)}_i\circ\psi^{(2)}_i\circ\cdots\circ\psi^{(j-1)}_i\circ\psi^{(j)}_i &\ \textrm{ if } j=k.\\
\end{array}\right.
\end{displaymath}

We then have the following lemma.
\begin{lem}
The concatenated graph as above is a non-reroutable $(n_1,n_2,\ldots,n_k)$-graph with $|G^{(1)}|_\mcm+|G^{(2)}|_\mcm+\cdots+|G^{(k-1)}|_\mcm$ mergings.
\end{lem}
It immediately follows that
\begin{thm}
For $n_1\le n_2\le \cdots \le n_k$,
$$
\mstar(n_1,n_2,\ldots,n_k)\ge\sum_{i=1}^{k-1}\mstar(n_i,n_i).
$$
\end{thm}
\begin{figure}
\psfrag{sss}{$S$}
\psfrag{s1a}{$S^{(1)}$}\psfrag{s2a}{$S^{(2)}$}
\psfrag{r1a}{$R^{(1)}$}\psfrag{r2a}{$R^{(2)}$}\psfrag{r3a}{ $R^{(3)}$}
\psfrag{h1a}{$\widehat{R}^{(1)}$}\psfrag{h2a}{$\widehat{R}^{(2)}$}
\centering
\includegraphics[width=0.65\textwidth]{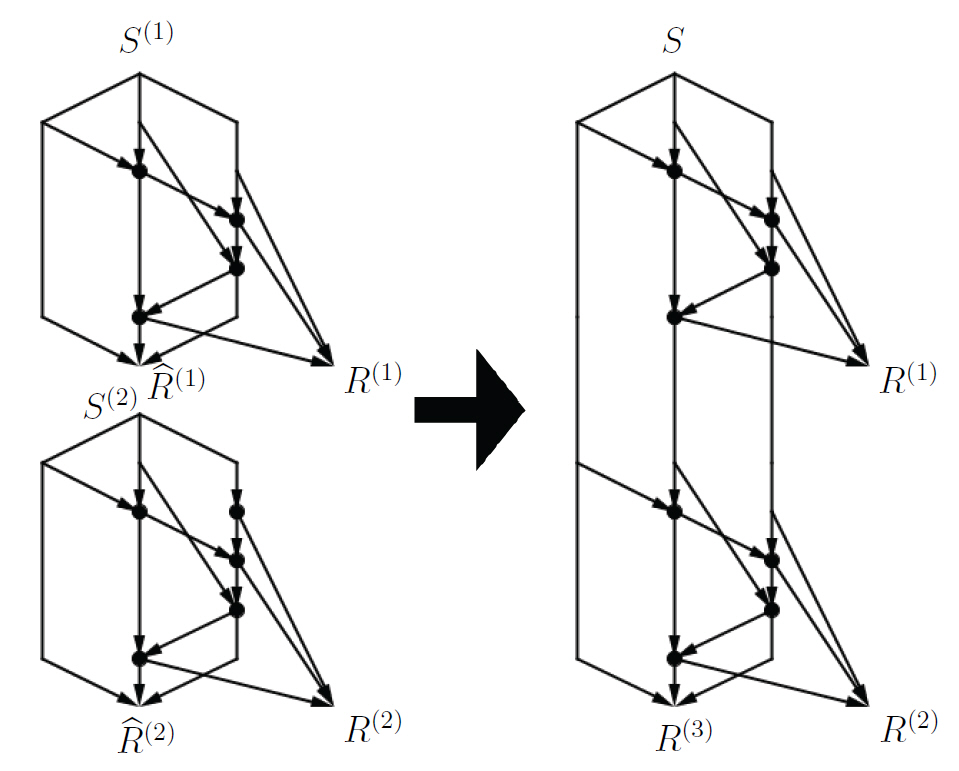}
\caption{Concatenate two $(3,3)$-graphs to obtain a $(3,3,3)$-graph} \label{picMstar333}
\end{figure}

\end{document}